\documentclass[conference]{IEEEtran}
\IEEEoverridecommandlockouts

\usepackage{cite}
\usepackage{url}
\usepackage{amsmath,amssymb,amsfonts}
\usepackage{algorithmic}
\usepackage{graphicx}
\usepackage{subcaption}
\usepackage{textcomp}
\usepackage{xcolor}
\def\BibTeX{{\rm B\kern-.05em{\sc i\kern-.025em b}\kern-.08em
    T\kern-.1667em\lower.7ex\hbox{E}\kern-.125emX}}
\usepackage{makecell}
\usepackage{preamble}
\usepackage{cleveref}
\usepackage{dsfont}
\newcommand{\Ones}{\mathds{1}}

\newcommand*{\eratio}{\rho}

\newcommand{\ZGn}{Z_n^{\Gibbs}}
\newcommand{\ZDBn}{Z_n^{\Bethe,2}}

\newcommand{\CGn}{C^{\Gibbs}}
\newcommand{\CDBn}{C^{\Bethe,2}}

\newcommand{\PGn}{P_{\Gibbs}}
\newcommand{\PDBn}{P_{\Bethe,2}}

\usepackage{tikz}
\usepackage{mathtools}
\usepackage{cuted}

\usetikzlibrary{
    calc,
    fit,
    arrows,
    chains,
    matrix,
    positioning,
    scopes,
    shapes,
    decorations.pathreplacing
}

\pgfdeclarelayer{background}
\pgfdeclarelayer{behind}
\pgfdeclarelayer{above}
\pgfdeclarelayer{glass}
\pgfsetlayers{background,behind,main,above,glass}

\makeatletter
\tikzset{join/.code=\tikzset{after node path={%
    \ifx\tikzchainprevious\pgfutil@empty\else(\tikzchainprevious)%
    edge[every join]#1(\tikzchaincurrent)\fi}}
}
\makeatother

\tikzset{
    >=stealth',
    every on chain/.append style={join},
    every join/.style={->},
    labeled/.style={
        execute at begin node=$\scriptstyle,
        execute at end node=$
    }
}

\begin{document}

\title{Double-Cover-Based\,Analysis\,of\,the\,Bethe\,Permanent of Block-Structured Positive Matrices
}

\author{\IEEEauthorblockN{Binghong Wu}
\IEEEauthorblockA{\textit{Department of Information Engineering} \\
\textit{The Chinese University of Hong Kong}\\
Shatin, N.T., Hong Kong\\
wb024@ie.cuhk.edu.hk}
\and
\IEEEauthorblockN{Pascal O. Vontobel}
\IEEEauthorblockA{\textit{Department of Information Engineering} \\
\textit{The Chinese University of Hong Kong}\\
Shatin, N.T., Hong Kong\\
pascal.vontobel@ieee.org}
}

\maketitle

\pagestyle{plain}
\thispagestyle{plain}

\begin{abstract}
We consider the permanent of a square matrix with non-negative entries. A tractable approximation is given by the so-called Bethe permanent, which can be efficiently computed by running the sum-product algorithm on a suitable factor graph. While the ratio of the permanent of a matrix to its Bethe permanent is, in the worst case, upper and lower bounded by expressions that are exponentially far apart in the matrix size, in practice it is observed for many ensembles of matrices of interest that this ratio is strongly concentrated around some value that depends only on the matrix size. In this paper, for an ensemble of block-structured matrices where entries in a block take the same value, we numerically study the ratio of the permanent of a matrix to its Bethe permanent. It is observed that also for this ensemble the ratio is strongly concentrated around some value depending only on a few key parameters of the ensemble. We use graph-cover-based approaches, combined with multivariate analytic combinatorics, to explain the reasons for this behavior and to quantify the observed value.
\end{abstract}


\section{Introduction}
\label{sec:intro}
Let $n$ be a positive integer and let $\matr{A} \defeq (a_{i,j}) \in \mathbb{R}_{\ge 0}^{n\times n}$. The permanent of the matrix $\matr{A}$ is defined to be
\begin{align}
    \perm(\matr{A})
    &\defeq
    \sum_{\sigma\in\set{S}_n}\ \prod_{i \in [n]} a_{i,\sigma(i)} ,
    \label{eq:intro:perm}
\end{align}
where $\set{S}_n$ is the set of all permutations of $[n] \defeq \{ 1, \ldots, n \}$. The permanent is a fundamental quantity in combinatorics and statistical physics~\cite{Minc:78}. In graphical-model terms~\cite{Kschischang:Frey:Loeliger:01,Loeliger:04:1,Yedidia:Freeman:Weiss:05:1}, one can formulate a graphical model such that its partition function equals $\perm(\matr{A})$. Equivalently, $\perm(\matr{A})$ is characterized by the minimum of the Gibbs free energy function associated with the graphical model, and therefore we will refer to it as the Gibbs permanent. Computing $\perm(\matr{A})$ exactly is \#P-hard, motivating efficient surrogates. The Bethe permanent $\perm_{\Bethe}(\matr{A})$ is obtained by minimizing the Bethe free energy function (or, equivalently, by running the sum-product algorithm (SPA) on a suitable graphical model)~\cite{Yedidia:Freeman:Weiss:05:1,6352911}. While often remarkably accurate, the Bethe approximation incurs a systematic gap, making the ratio $\perm(\matr{A})/\perm_{\Bethe}(\matr{A})$ a key quantity to understand.

A main motivation for the investigations in the present paper is the pattern maximum likelihood (PML) problem~\cite{10.5555/1036843.1036895}. Several works~\cite{6283654,6804280,pmlr-v134-anari21a} show that determining the PML is equivalent to computing $\arg\max_{\matr{A}}\perm(\matr{A})$ over a structured family of non-negative matrices. Since exact permanents are intractable for large matrices, Vontobel~\cite{6283654,6804280} proposed the practical surrogate objective $\arg\max_{\matr{A}}\ \perm_{\Bethe}(\matr{A})$, which makes it essential to understand the approximation ratio $\perm(\matr{A})/\perm_{\Bethe}(\matr{A})$ on PML-induced instances. More recently, Anari \emph{et al.}~\cite{pmlr-v134-anari21a} highlighted that these instances possess pronounced low-complexity structure (e.g., few distinct rows/columns, low non-negative rank) and leveraged this viewpoint to justify and analyze scalable surrogates, including both Bethe- and Sinkhorn-type approximations, for approximate PML optimization.

Therefore, studying the ratio $\perm(\matr{A})/\perm_{\Bethe}(\matr{A})$ for highly structured matrices is of particular interest. Specifically, block-structured matrices, where all entries in a block take the same value, directly model the low-complexity regime of having only a few distinct rows and columns. A simple instance for such a matrix with $n \! = \! 5$ and $m \! = \! 2$ is
\begin{align}
    \hspace{-0.25cm}
    \matr{A}
    &\defeq
    \matr{A}\bigl( \matr{B}, \bm{k}, \bm{\ell} \bigr)
    \defeq
    \mbox{\tiny $
    \renewcommand{\arraystretch}{1.3}
    \left(
        \begin{array}{cccc|c}
            b_{11} & b_{11} & b_{11} & b_{11} & b_{12}\\
            b_{11} & b_{11} & b_{11} & b_{11} & b_{12}\\
            b_{11} & b_{11} & b_{11} & b_{11} & b_{12}\\
            \hline
            b_{21} & b_{21} & b_{21} & b_{21} & b_{22}\\
            b_{21} & b_{21} & b_{21} & b_{21} & b_{22}
        \end{array}
    \right)
    $}                                    
    \in \mathbb{R}_{> 0}^{n \times n},
    \label{eq:block:structured:matrix:example:1}
\end{align}
which is based on the matrix $\matr{B} \defeq (b_{i,j}) \in \mathbb{R}_{> 0}^{m \times m}$, row block sizes $\bm{k} = (k_1; k_2) = (3; 2)$ and column block sizes $\bm{\ell} = (\ell_1; \ell_2) = (4;1)$. In the case of the PML, we have $b_{i,j} \defeq q_i^{\mu_j}$, with probabilities $\{q_i\}_i$ and frequencies $\{\mu_j\}_j$. (See, e.g., \cite{6283654,6804280} for details.)

The Bethe permanent has a substantial literature, ranging from variational viewpoints to approximation guarantees (see, e.g., \cite{Chertkov:Kroc:Vergassola:08:1,Huang:Jebara:09:1,Watanabe:Chertkov:10:1,Gurvits:11:2,Gurvits:Samorodnitsky:14:1,8948682,Straszak:Vishnoi:19:1,pmlr-v134-anari21a}). In particular, for an arbitrary matrix $\matr{A} \in \mathbb{R}_{\ge 0}^{n\times n}$ it was proven that
\begin{align}
    1
    &\leq
    \frac{\perm(\matr{A})}{\perm_{\Bethe}(\matr{A})}
    \leq
    2^{n/2},
\end{align}
with the lower and upper bounds being tight. In contrast to this substantial gap between the worst case upper and lower bounds, in practice it is observed that for many matrix ensembles of interest the ratio $\perm(\matr{A}) / \perm_{\Bethe}(\matr{A})$ is strongly concentrated. See, e.g., the paper~\cite{Ng:Vontobel:22:2}, that studied matrices with i.i.d. entries.

In the present paper, we make a similar observation for block-structured matrices $\matr{A}\bigl( \matr{B}, \bm{k}, \bm{\ell} \bigr)$, see Fig.~\ref{fig:ratiocomparison_all}. The results in this figure are obtained as follows: we fix $n \! = \! 5$ and $m \! = \! 2$, and construct $\matr{A} \defeq \matr{A}\bigl( \matr{B}, \bm{k}, \bm{\ell} \bigr)$ based on $\matr{B}$-matrices whose entries are generated i.i.d. according to some distribution and based on block sizes $\bm{k} = (k_1; k_2)$ and $\bm{\ell} = (\ell_1; \ell_2)$ that are randomly generated. For each such matrix, we draw a blue triangle at the location $(x,y) = \bigl( \perm(\matr{A}), \perm_{\Bethe}(\matr{A}) \bigr)$. We see that the resulting blue triangles follow closely the line $x \mapsto y \defeq \sqrt{\e / (2 \pi n)} \cdot x$, i.e., the ratio $\perm(\matr{A}) / \perm_{\Bethe}(\matr{A})$ is close to $\sqrt{2 \pi n / \e}$. Interestingly, this ratio is the same as the ratio that was observed for the ensemble under consideration in~\cite{Ng:Vontobel:22:2}. These observations reinforce the practical reliability of Bethe-based surrogates on PML-type structured matrices.

While $\perm_{\Bethe}(\matr{A})$ can be computed efficiently, it turns out to be rather difficult to characterize this quantity analytically except for a few special cases. This is the reason why in this paper we analyze the degree-$2$ Bethe permanent $\perm_{\Bethe,2}(\matr{A})$ of $\matr{A}$ instead of $\perm_{\Bethe}(\matr{A})$, similar to what was done in~\cite{Ng:Vontobel:22:2}. (See also the recent related work~\cite{zhou2026complexvaluedmatrixpermanentsspabasedapproximations} for complex-valued matrices.) While $\perm_{\Bethe,2}(\matr{A})$ is different from $\perm_{\Bethe}(\matr{A})$, these investigations are valuable because
\begin{itemize}
    \item the ratio $\perm(\matr{A}) / \perm_{\Bethe,2}(\matr{A})$ is predictive of the ratio $\perm(\matr{A}) / \perm_{\Bethe}(\matr{A})$,

    \item the analysis of $\perm_{\Bethe,2}(\matr{A})$ gives insight into the value taken by the ratio $\perm(\matr{A}) / \perm_{\Bethe,2}(\matr{A})$, and with that insight into the value taken by the ratio $\perm(\matr{A}) / \perm_{\Bethe}(\matr{A})$.

\end{itemize}
Continuing the above simulation study, Fig.~\ref{fig:ratiocomparison_all} also shows the numerical results for the degree-$2$ Bethe permanent $\perm_{\Bethe,2}(\matr{A})$, where a red circle is drawn at the location $(x,y) = \bigl( \perm(\matr{A}), \perm_{\Bethe,2}(\matr{A}) \bigr)$. We see that the resulting red circles follow closely the line $x \mapsto y \defeq \sqrt[4]{\e / (\pi n)} \cdot x$, i.e., the ratio $\perm(\matr{A}) / \perm_{\Bethe,2}(\matr{A})$ is close to $\sqrt[4]{\pi n / \e}$.

Recall that the degree-$2$ Bethe permanent $\perm_{\Bethe,2}(\matr{A})$, and, more generally, the degree-$M$ Bethe permanent $\perm_{\Bethe,M}(\matr{A})$ are defined as follows~\cite{6570731,6352911}. Namely,
\begin{align}
    \perm_{\Bethe,M}(\matr{A})
    &\defeq
    \sqrt[M]{\Big\langle\,\perm\!\big(\matr{A}^{\uparrow \matr{\tilde{P}}}\big)\,\Big\rangle_{\matr{\tilde{P}}\in\tilde\Phi_M}} ,
    \label{eq:intro:cover}
\end{align}
where $\langle \, \cdot \, \rangle$ is the arithmetic average over all degree-$M$ covers (equivalently, over all cover configurations $\matr{\tilde{P}}\in\tilde\Phi_M$), and $\matr{A}^{\uparrow \matr{\tilde{P}}}$ is the corresponding $Mn\!\times\! Mn$ lifted matrix. (See~\cite{6570731,6352911} for details.) As was shown in~\cite{6570731,6352911},
\begin{align}
    \perm_{\Bethe,1}(\matr{A})
    &= \perm(\matr{A}), \\
    \limsup_{M\to\infty}\ \perm_{\Bethe,M}(\matr{A})
    &= \perm_{\Bethe}(\matr{A}).
\end{align}
It is straightforward to see that one can write
\begin{align}
    \frac{\perm(\matr{A})}{\perm_{\Bethe}(\matr{A})}
    =
    \frac{\perm(\matr{A})}{\perm_{\Bethe,2}(\matr{A})}
    \cdot
    \frac{\perm_{\Bethe,2}(\matr{A})}{\perm_{\Bethe}(\matr{A})}.
\label{eq:intro:ratio-decomp}
\end{align}
A recurring empirical observation, and a useful working heuristic for asymptotics, is that the two factors on the right-hand side are often very similar~\cite{7746637,Ng:Vontobel:22:2,Vontobel2025ITW}. Consequently,
\begin{align}
    \frac{\perm(\matr{A})}{\perm_{\Bethe}(\matr{A})}
    \approx
    \left(\frac{\perm(\matr{A})}{\perm_{\Bethe,2}(\matr{A})}\right)^{\!\! 2},
    \label{eq:tworatio}
\end{align}
and so understanding the ratio $\perm(\matr{A}) / \perm_{\Bethe,2}(\matr{A})$ goes a long way toward understanding the ratio $\perm(\matr{A}) / \perm_{\Bethe}(\matr{A})$.

\begin{figure}[t]
    \vspace{0.05in}
    \centering
    \includegraphics[width=0.9\linewidth]{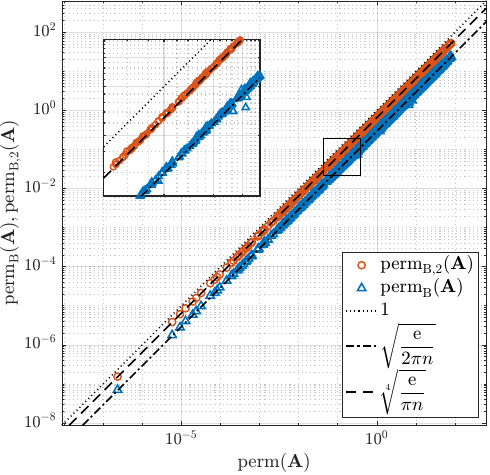}
    \caption{Numerical results for block-structured matrices $\matr{A}$, along with theoretical asymptotics, for $n\!=\!5$ and $m\!=\!2$. (See Section~\ref{sec:intro} for details.)}
    \label{fig:ratiocomparison_all}
\end{figure}

In particular, Ng and Vontobel~\cite{Ng:Vontobel:22:2} established the following results. For the all-one matrix $\matr{A} \defeq \Ones_{n\times n}$, they proved that\footnote{The notation $a(n) \sim b(n)$ stands for $\lim\limits_{n \to \infty} \frac{a(n)}{b(n)} = 1$.}
\begin{align}
    \frac{\perm(\matr{A})}{\perm_{\Bethe}(\matr{A})}
    \sim
    \sqrt{\frac{2\pi n}{\e}},
    \qquad
    \frac{\perm(\matr{A})}{\perm_{\Bethe,2}(\matr{A})}
    \sim
    \sqrt[4]{\frac{\pi n}{\e}},
    \label{eq:intro:allone}
\end{align}
showing analytically that the expression in~\eqref{eq:tworatio} holds up to a factor $\sqrt{2}$. Moreover, for matrices with i.i.d.\ entries, they showed that
\begin{align}
    \gamma_{\Bethe,2}(n)
    \defeq
      \frac{\sqrt{\mathbb{E}\bigl[ \perm(\matr{A})^2 \bigr]}}
           {\sqrt{\mathbb{E}\bigl[ \perm_{\Bethe,2}(\matr{A})^2 \bigr]}}
    \;\sim\;
    \gamma'_{\Bethe,2}(n)
    \defeq
    \sqrt[4]{\frac{\pi n}{\e}}.
\label{eq:intro:expectation}
\end{align}
Notably, the ratios $\sqrt{2\pi n/\e}$ and $\sqrt[4]{\pi n/\e}$ that appear in~\eqref{eq:intro:allone}--\eqref{eq:intro:expectation} also play a key role for the matrices considered in Fig.~\ref{fig:ratiocomparison_all}.

The main result of the present paper is, in the asymptotic setting stated in Assumption~\ref{assump:general_block} below, that for a matrix $\matr{A} \defeq \matr{A}\bigl( \matr{B}, \bm{k}, \bm{\ell} \bigr)$ it holds that
\begin{align}
    \frac{\perm(\matr{A})}{\perm_{\Bethe,2}(\matr{A})}
    &\sim
    \sqrt[4]{\frac{\pi n}{\e}}
    \cdot
    \prod_{i=2}^{m} \sqrt[4]{\frac{\e^{-\eratio_i}}{1-\eratio_i}},
    \label{eq:intro:main}
\end{align}
where $\{ \rho_i \}_{i=2}^{m}$ are real numbers in the interval $[0,1)$ computed based on $\matr{B}$, $\bm{k}$, and $\bm{\ell}$. For many matrix ensembles $\matr{A}\bigl( \matr{B}, \bm{k}, \bm{\ell} \bigr)$ of interest (in particular the matrix ensembles of interest in the PML setting), we observe that $\rho_i \ll 1$, \ $2 \leq i \leq m$, and so
\begin{align}
    \frac{\perm(\matr{A})}{\perm_{\Bethe,2}(\matr{A})}
    \sim
    \sqrt[4]{\frac{\pi n}{\e}}
    \cdot
    \Biggl( 1 + O\Biggl( \sum_{i=2}^{m}\eratio_i^2 \Biggr) \Biggr).
    \label{eq:intro:main-smallr}
\end{align}
Note that, whereas prior results in~\eqref{eq:intro:expectation} characterized ``only'' the expectation-value ratio $\mathbb{E}\bigl[ \perm(\matr{A})^2 \bigr] \ / \ \mathbb{E}\bigl[\perm_{\Bethe,2}(\matr{A})^2 \bigr]$ for the matrix ensemble under consideration in~\cite{Ng:Vontobel:22:2}, the new results in~\eqref{eq:intro:main}--\eqref{eq:intro:main-smallr} characterize the ratio $\perm(\matr{A}) \ / \ \perm_{\Bethe,2}(\matr{A})$ for specific matrices $\matr{A} \defeq$ $\matr{A}\bigl( \matr{B}, \bm{k}, \bm{\ell} \bigr)$, albeit in the asymptotic setting of Assumption~\ref{assump:general_block}.

The rest of this paper is structured as follows:
\begin{itemize}

    \item In \Cref{sec:bg-nfg-dc-cycleindex}, we review the normal factor graph (NFG) representations of the Gibbs and degree-$2$ Bethe permanents, respectively. We also revisit the cycle-index framework that is used to compute quantities of interest.

    \item In \Cref{sec:general_setup}, we formulate the block constraints via multivariate generating functions and extend the ``analytic combinatorics in several variables (ACSV)'' framework~\cite{pemantle2024analytic} to obtain the asymptotics.

    \item In \Cref{sec:numerics}, we present numerical experiments on representative structured instances from the PML setup, demonstrating the robustness of the predicted ratio.

    \item In \Cref{sec:conclusion:outlook}, we conclude the paper and present an outlook.
  
    \item All technical proofs and auxiliary derivations are deferred to the appendices.

\end{itemize}

\noindent
\textbf{Notation.}
Unless stated otherwise, all vectors are column vectors. We use $(\bm{x}; \bm{y})$ to denote $[\bm{x}^\transp\!, \bm{y}^\transp]^\transp$, we use $\bm{1}_m$ to denote the all-one vector of size $m \times 1$, and we use $\Ones_{a \times b}$ to denote the all-one matrix of size $a \!\times\! b$. For $\bm{k} \in \mathbb{Z}_{\geq 0}^{m}$, we define $\bm{k}! \!\defeq\! \prod_{i \in [m]} (k_i!)$. For $\bm{w} \in \mathbb{R}_{\geq 0}^{m}$ and $\bm{k} \in \mathbb{Z}_{\geq 0}^{m}$, we define $\bm{w}^{\bm{k}} \!\defeq\! \prod_{i \in [m]} w_i^{k_i}$.

\section{Background: \\ NFGs, Double Covers, and the Cycle Index}
\label{sec:bg-nfg-dc-cycleindex}

\subsection{Normal factor graphs and partition sums}

A factor graph $\graphN$ can be used to depict a multivariate function that is the product of (simpler) multivariate functions. The former function is called the global function and the latter functions are called local functions. The corresponding factor graph consists of function nodes and variable nodes, where for each local function one draws a function node, for every variable one draws a variable node, and for every variable appearing as an argument of a local function one draws an edge between the corresponding variable and function nodes. A valid configuration is an assignment of values to the variable nodes such that the global function takes on a nonzero
value.

The partition sum $Z(\graphN)$ is defined to be the sum of the global function over all possible variable assignments, or, equivalently, the sum of the global function over all valid configurations. With suitable local functions, partition sums encode many combinatorial quantities of interest.

In the present paper, we use normal factor graphs (NFGs) in the sense of~\cite{Kschischang:Frey:Loeliger:01,Loeliger:04:1}, where the factorization is formulated such that all variable nodes have degree two or one, and because of this the variable nodes are then usually omitted in drawings, i.e., there are no variable nodes in drawings, and variables are simply associated with edges or half-edges.

Our analysis builds upon the double-cover NFG framework for the permanent~\cite{6352911,6570731,7746637,Ng:Vontobel:22:2,Vontobel2025ITW}. We adopt the formulation from these papers, restating here only the key definitions and propositions required for our derivation. (We refer to Appendix~A for a much more detailed discussion.) The corresponding graphs are shown in Fig.~\ref{fig:ffg:permanent:1}.

\begin{figure}[t]
    \vspace{0.05in}
    \begin{subfigure}{0.45\linewidth}
        \centering
        \includegraphics[width=0.85\linewidth]{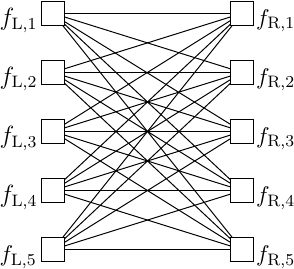}
    \end{subfigure}
    \begin{subfigure}{0.45\linewidth}
        \centering
        \includegraphics[width=0.85\linewidth]{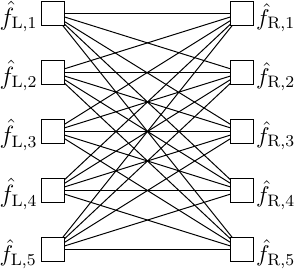}
    \end{subfigure}
    \caption{NFGs used to represent $\perm(\matr{A})$ (left) and $\perm_{\Bethe,2}(\matr{A})$ (right).}
    \label{fig:ffg:permanent:1}
\end{figure}

\subsection{NFGs for permanents}
\label{subsec:bg:nfg-doublecover}

Given $\matr{A} = (a_{i,j})\in\mathbb{R}_{\geq 0}^{n\times n}$, one builds an NFG $\graphN(\matr{A})$ on the complete bipartite graph, see Fig.~\ref{fig:ffg:permanent:1}~(left), with function nodes $\{ f_{\mathrm{L},i} \}_{i \in [n]}$ on the left-hand side, function nodes $\{ f_{\mathrm{R},j} \}_{j \in [n]}$ on the right-hand side, and, for every $(i,j) \in [n]^2$, the variable $x_{i,j} \in \{ 0, 1 \}$ is associated with the edge connecting $f_{\mathrm{L},i}$ and $f_{\mathrm{R},j}$. The local functions are defined such that they take on a nonzero value if and only if exactly one of the variables associated with the incident edges equals~$1$.\footnote{In order to avoid awkward language, we assume that $a_{i,j} > 0$, $(i,j) \in [n]^2$ in this description.} With this, every valid configuration of $\graphN(\matr{A})$ corresponds to a perfect matching in the bipartite graph. In particular, if the perfect matching corresponds to a permutation $\sigma \in \set{S}_n$, then, because of the way that the local functions are defined, the global function equals $\prod_{i \in [n]} a_{i,\sigma(i)}$, and so the partition function of $\graphN(\matr{A})$ satisfies $Z\!\bigl(\graphN(\matr{A})\bigr) = \perm(\matr{A})$.

The paper~\cite{Ng:Vontobel:22:2} defines an NFG $\hat\graphN(\matr{A})$, shown in Fig.~\ref{fig:ffg:permanent:1}~(right), whose partition sum satisfies $Z\bigl( \hat\graphN(\matr{A}) \bigr) = \bigl( \perm_{\Bethe,2}(\matr{A}) \bigr)^2$. While the graph underlying $\hat\graphN(\matr{A})$ is the same as the graph underlying $\graphN(\matr{A})$ in Fig.~\ref{fig:ffg:permanent:1}~(left), its variable alphabets and local functions are different. Namely, for every $(i,j) \in [n]^2$,
\begin{align}
  \hat{x}_{i,j}
    &\in \{0,1\}^2
     = \{(0,0),(0,1),(1,0),(1,1)\},
    \label{eq:bg:edge-pair}
\end{align}
and with the definition of the local functions appearing in Appendix~A. The properties of these local functions turn out to have the following important implication for the valid configurations of $\hat\graphN(\matr{A})$. Namely, after disregarding edges whose associated variable takes the value $(0,0)$, every valid configuration of $\hat\graphN(\matr{A})$ consists of disconnected components corresponding to
\begin{itemize}
    \item single edges whose assoc.\ variable takes the value $(1,1)$,
    \item simple cycles whose assoc.\ variables take the value $(0,1)$,
\end{itemize}
and with every function node being part of one of these disconnected components. Crucially, the value $(1,0)$ does not appear in any valid configuration.

Similar to the construction of $\hat\graphN(\matr{A})$, an NFG can be constructed such that its partition sum equals $\bigl( \perm(\matr{A}) \bigr)^2$. This NFG has local function nodes that have the following important implication for the valid configurations. Namely, after disregarding edges whose associated variable takes the value $(0,0)$, every valid configuration of this NFG consists of disconnected components corresponding to
\begin{itemize}
    \item single edges whose assoc.\ variable takes the value $(1,1)$,
    \item simple cycles whose assoc.\ variables take the value $(0,1)$,
    \item simple cycles whose assoc.\ variables take the value $(1,0)$,
\end{itemize}
and with every function node being part of one of these disconnected components.

With the help of $\hat\graphN(\matr{A})$, Ng and Vontobel~\cite[Proposition~2]{Ng:Vontobel:22:2} show that
\begin{align}
    \hspace{-1.75cm}
    \frac{\bigl( \perm(\matr{A}) \bigr)^2}{\bigl( \perm_{\Bethe,2}(\matr{A}) \bigr)^2}
    &\! = \!
        \Biggl(
            \sum_{\sigma_1,\sigma_2\in\set{S}_n} \!\!\!
            p(\sigma_1)\,p(\sigma_2)
            \, 2^{-c(\sigma_1 \circ \sigma_2^{-1})}
        \Biggr)^{-1} \!\!\! , \hspace{-1cm}
    \label{eq:intro:cycle-penalty}
\end{align}
where $p(\sigma) \defeq \bigl( \prod_{i \in [n]} a_{i,\sigma(i)} \bigr) /\perm(\matr{A})$ and where $c(\sigma)$ is the number of cycles of length at least two in $\sigma$.\footnote{If $\sigma \in \set{S}_6$ is such that $\sigma(1) = 2$, $\sigma(2) = 3$, $\sigma(3) = 1$, $\sigma(4) = 5$, $\sigma(5) = 4$, $\sigma(6) = 6$, then in cycle notation $\sigma$ reads $(123)(45)(6)$, showing that $\sigma$ consists of one cycle of length $3$, one cycle of length $2$, and one cycle of length $1$, and with this $c(\sigma) = 2$.} Note that the term $2^{-c(\sigma_1 \circ \sigma_2^{-1})}$ is the result of the non-existence of $(1,0)$-cycles in the valid configurations of $\hat\graphN(\matr{A})$. In particular,
\begin{itemize}
    \item without the term $2^{-c(\sigma_1 \circ \sigma_2^{-1})}$, the expression on the right-hand side of Eq.~\eqref{eq:intro:cycle-penalty} is equal to $1$,
    \item with the term $2^{-c(\sigma_1 \circ \sigma_2^{-1})}$, the expression on the right-hand side of Eq.~\eqref{eq:intro:cycle-penalty} is at least $1$.
\end{itemize}
With this, Eq.~\eqref{eq:intro:cycle-penalty} shows the clear influence of the cycles in $\hat\graphN(\matr{A})$ on the ratio $\perm(\matr{A}) / \perm_{\Bethe,2}(\matr{A})$.

Eq.~\eqref{eq:intro:cycle-penalty} will serve as our main black-box tool. Since the weight depends on $(\sigma_1,\sigma_2)$ only through the cycle structure of $\sigma_1\circ\sigma_2^{-1}$, it naturally leads to a cycle-index organization where the factor $1/2$ per non-trivial cycle becomes a simple substitution at the level of generating functions.

\section{Positive Block-Constant Model}
\label{sec:general_setup}

We now specialize to positive block-constant matrices with fixed row and column partitions and set up the multivariate coefficient-extraction framework used in the subsequent analysis.

\begin{assumption}
\label{assump:general_block}

Let the positive integer $m$ and the matrix $\matr{B} \defeq (b_{i,j})\in\mathbb{R}_{>0}^{m\times m}$ be fixed. For a positive integer $n$, let the vectors $\bm{k}, \bm{\ell} \in \mathbb{Z}_{>0}^{m}$ be such that $\sum_{i \in [m]} k_i \! = \! \sum_{j \in [m]} \ell_j \! = \! n$. The block-constant matrix $\matr{A} \defeq \matr{A}(\matr{B},\bm{k},\bm{\ell})\in\mathbb{R}_{>0}^{n\times n}$ is defined to be the matrix that is obtained by expanding each $b_{i,j}$ into a constant-value block of size $k_i\times \ell_j$, i.e.,
\begingroup
\setlength{\abovedisplayskip}{0.5em}
\setlength{\belowdisplayskip}{0.5em}
\setlength{\abovedisplayshortskip}{0.2em}
\setlength{\belowdisplayshortskip}{0.2em}
\begin{align}
    \matr{A} \defeq
    \begin{pmatrix}
        b_{1,1} \Ones_{k_1 \times \ell_1} & \cdots & b_{1,m} \Ones_{k_1 \times \ell_m} \\
        \vdots & \ddots & \vdots \\
        b_{m,1} \Ones_{k_m \times\ell_1} & \cdots & b_{m,m} \Ones_{k_m \times \ell_m}
    \end{pmatrix},
\end{align}
\endgroup
where $\Ones_{k\times\ell}$ denotes the all-one matrix of size $k\times\ell$. In this paper we consider the asymptotic setting $n\to\infty$ with $k_i\sim\alpha_i n$ and $\ell_j\sim\beta_j n$ for all $i,j\in[m]$, where $\bm\alpha,\bm\beta\in\mathbb R_{>0}^m$ are fixed vectors satisfying $\sum_{i\in[m]}\alpha_i=\sum_{j\in[m]}\beta_j=1$.
\end{assumption}

To capture the combinatorics of closed walks on $\hat\graphN(\matr{A})$, we introduce indeterminates $\bm{t}=(t_1;\dots;t_m)$ and $\bm{u}=(u_1;\dots;u_m)$ as counting indeterminates in the generating functions, marking row and column types so that the partitioning $(\bm{k},\bm{\ell})$ is selected by $[\bm{t}^{\bm{k}}\bm{u}^{\bm{\ell}}]$. To encode the relevant walk weights in a form that is convenient for spectral analysis, we introduce the following symmetric state-transition matrix.

\begin{definition}
\label{def:S-general}
We define the symmetric state-transition matrix $\matr{S}\defeq\matr{S}(\bm{t},\bm{u})$ of size $m \times m$ by
\begin{align}
    \matr{S}
    \defeq
    \diag(\bm{u})^{1/2}
    \cdot \matr{B}^\transp
    \cdot \diag(\bm{t})
    \cdot \matr{B}
    \cdot \diag(\bm{u})^{1/2}.
    \label{eq:Sdef}
\end{align}
The total weight of length-$h$ closed walks is given by $\trace(\matr{S}^h)$.
\end{definition}

Let $\lambda_1, \ldots, \lambda_m$ be the eigenvalues of $\matr{S}$. With this, $\trace(\matr{S}^h)=\sum_{i\in[m]}\lambda_i^h$ for $h\in\mathbb{Z}_{\ge 0}$. Moreover, when $\bm{t},\bm{u}\in\mathbb{R}_{>0}^m$ (as will hold at the critical point identified later), the matrix $\matr{S}(\bm{t},\bm{u})$ is real symmetric and positive semi-definite, and so, without loss of generality, we can assume that $\lambda_1 \ge \lambda_2 \ge \cdots \ge \lambda_m \ge 0$. Under the strict-positivity assumption on $\matr{B}$ (and $\bm{t},\bm{u}$), Perron--Frobenius theory yields a leading eigenvalue $\lambda_1$ with a spectral gap, i.e., $\lambda_1 > \lambda_2$. For details from a factor-graph perspective, please see Appendix~B.

Cycle-index-based generating functions can be used to analyze $\bigl( \perm(\matr{A}) \bigr)^2$ and $\bigl( \perm_{\Bethe,2}(\matr{A}) \bigr)^2$ based on the observations before~Eq.~\eqref{eq:intro:cycle-penalty} and Eq.~\eqref{eq:intro:cycle-penalty} itself, as shown in the following two lemmas.

\begin{lemma}
\label{lem:block_trace_egf}
  
Under Assumption~\ref{assump:general_block} and Definition~\ref{def:S-general}, let $(\bm{t},\bm{u})$ be such that $\bm{t},\bm{u}\in\mathbb{R}_{>0}^m$ and $\lambda_1\!\left(\matr{S}(\bm{t},\bm{u})\right) < 1$. Define
\begingroup
\setlength{\abovedisplayskip}{0.4em}
\setlength{\belowdisplayskip}{0.4em}
\setlength{\abovedisplayshortskip}{0.2em}
\setlength{\belowdisplayshortskip}{0.2em}
\begin{align}
    \CGn(\bm{t},\bm{u})
    &\defeq \exp\!\left(\sum_{h\ge 1}\frac{\trace\bigl(\matr{S}(\bm{t},\bm{u})^h\bigr)}{h}\right),
    \label{eq:CG-trace}\\
    \CDBn(\bm{t},\bm{u})
    &\defeq \exp\!\left(\trace\bigl(\matr{S}(\bm{t},\bm{u})\bigr) + \sum_{h\ge 2}\frac{\trace\bigl(\matr{S}(\bm{t},\bm{u})^h\bigr)}{2 h}\right).
    \label{eq:CB-trace}
\end{align}
\endgroup
Then\footnote{After the original submission of the paper, we realized that~\eqref{eq:CGCB-det:1} is a special case of~\cite[Theorem 1]{chabaud2022quantum} (a generalization of MacMahon's master theorem~\cite{macmahon1915combinatory}); however, our proof differs from the one presented therein.}
\begingroup
\setlength{\abovedisplayskip}{0.4em}
\setlength{\belowdisplayskip}{0.4em}
\setlength{\abovedisplayshortskip}{0.2em}
\setlength{\belowdisplayshortskip}{0.2em}
\begin{align}
    \CGn(\bm{t},\!\bm{u})
    & = \frac{1}{\det\bigl(\matr{I}\!-\!\matr{S}(\bm{t},\!\bm{u})\bigr)},
    \label{eq:CGCB-det:1} \\
    \CDBn(\bm{t},\!\bm{u})
    & = \frac{\exp\Bigl(\tfrac{1}{2}\trace\bigl(\matr{S}(\bm{t},\!\bm{u})\bigr)\Bigr)}{\sqrt{\det\bigl(\matr{I}\!-\!\matr{S}(\bm{t},\!\bm{u})\bigr)}}.
    \label{eq:CGCB-det:2}
\end{align}
\endgroup
\end{lemma}

\begin{proof}
The result uses that for $\bm{t},\bm{u}\in\mathbb{R}_{>0}^m$ such that $\lambda_1(\matr{S}(\bm{t},\bm{u}))<1$, the expansion $\log(\matr{I}\!-\!\matr{S})=-\sum_{h\ge1}\matr{S}^h/h$ converges. Moreover, for any square matrix $\matr{M}$, it holds that $\exp\bigl( \trace(\matr{M}) \bigr) = \det\bigl( \exp(\matr{M}) \bigr)$.
\end{proof}

\begin{lemma}
\label{lem:bridge-ogf-nfg}
  
Consider the matrix $\matr{A} \defeq \matr{A}(\matr{B},\bm{k},\bm{\ell})$ from Assumption~\ref{assump:general_block}. It holds that
\begin{align}
    \bigl( \perm(\matr{A}) \bigr)^2
    &= \bm{k}! \cdot \bm{\ell}! \cdot \ZGn(\bm{k},\bm{\ell}),
    \label{eq:partition:1} \\
    \bigl( \perm_{\Bethe,2}(\matr{A}) \bigr)^2
    &= \bm{k}! \cdot \bm{\ell}! \cdot \ZDBn(\bm{k},\bm{\ell}),
    \label{eq:partition:2}
\end{align}
where
\begin{align}
    \ZGn(\bm{k}, \bm{\ell})
    &= [\bm{t}^{\bm{k}}\bm{u}^{\bm{\ell}}] \, \CGn(\bm{t}, \bm{u}),
    \label{eq:ZGZB-coeff:1} \\
    \ZDBn(\bm{k}, \bm{\ell})
    &= [\bm{t}^{\bm{k}}\bm{u}^{\bm{\ell}}] \, \CDBn(\bm{t}, \bm{u}).
    \label{eq:ZGZB-coeff:2}
\end{align}
Here, $[\bm{t}^{\bm{k}}\bm{u}^{\bm{\ell}}] \, \CGn(\bm{t}, \bm{u})$ is the coefficient $c_{\bm{k},\bm{\ell}}$ of the monomial $c_{\bm{k},\bm{\ell}} \bm{t}^{\bm{k}}\bm{u}^{\bm{\ell}}$ in the power series expansion of $\CGn(\bm{t}, \bm{u})$, with a similar interpretation for $[\bm{t}^{\bm{k}}\bm{u}^{\bm{\ell}}] \, \CDBn(\bm{t}, \bm{u})$.
\end{lemma}

\begin{proof}
The coefficients $\ZGn(\bm{k},\bm{\ell})$ and $\ZDBn(\bm{k},\bm{\ell})$ arise from the cycle-index enumeration driven by $\matr{S}(\bm{t},\bm{u})$, hence depend only on the $m\times m$ matrix $\matr{B}$. However, they do not account for the $k_i\!\times\! \ell_j$ block repetitions in $\matr{A}(\matr{B},\bm{k},\bm{\ell})$. Restoring labels within each row/column type class contributes the multiplicity factor $\bm{k}! \cdot \bm{\ell}!$ for both $\bigl( \perm(\matr{A}) \bigr)^2$ and $\bigl( \perm_{\Bethe,2}(\matr{A}) \bigr)^2$, thereby yielding~\eqref{eq:partition:1}--\eqref{eq:partition:2}.
\end{proof}

\begin{proposition}
\label{prop:block_asymp}

Let $\bm{z} \defeq (\bm{t}; \bm{u})$ denote the set of $2m$ counting indeterminates. Let $\bm{w}\defeq (\bm{t^*};\bm{u^*}) \in \mathbb{R}_{>0}^{2m}$ be the unique positive critical point within the ACSV framework governing the coefficient growth associated with the partitioning $\bm{r} \defeq (\bm{k};\bm{\ell})$. Let $\lambda_1 > \lambda_2 \geq \dots \geq \lambda_m \geq 0$ be the eigenvalues of the matrix $\matr{S}(\bm{w})$. Defining the spectral ratios $\eratio_i \defeq \lambda_i/\lambda_1$ at this critical point, the coefficient extraction yields
\begin{align}
    \ZGn & \sim
    \frac{n \sqrt{2m} \cdot \bm{w}^{-\bm{r}}}{(2\pi)^{m-1} \|\bm{r}\|_2^m \sqrt{\det(\mathcal{H})}}
    \!\cdot \!
    \prod_{i=2}^{m}\frac{1}{1\!-\!\eratio_i}, \\
    \ZDBn & \sim
    \frac{n \sqrt{2m} \cdot \bm{w}^{-\bm{r}}}{(2\pi)^{m-1} \|\bm{r}\|_2^m \sqrt{\det(\mathcal{H})}}
    \!\cdot \! 
    \sqrt{\frac{\e}{\pi n}}
    \!\cdot \!
    \prod_{i=2}^{m}\sqrt{\frac{\e^{\eratio_i}}{1\!-\!\eratio_i}},
\end{align}
as $n\to\infty$, where $\det(\mathcal{H})$ denotes the determinant of the $(2m\!-\!2) \times (2m\!-\!2)$ Hessian matrix restricted to the tangent space of the singular manifold at $\bm{w}$.
\end{proposition}

\begin{proof}
See Appendix~C for the ACSV derivation, and Appendix~D for the Sinkhorn-scaling-algorithm-based characterization of the positive critical point.
\end{proof}

Combining the above results yields the following theorem.

\begin{theorem}
\label{thm:general_ratio}

Consider the matrix $\matr{A} \defeq \matr{A}(\matr{B},\bm{k},\bm{\ell})$ from Assumption~\ref{assump:general_block}. Let $\eratio_i \defeq \lambda_i/\lambda_1$, $i = 2,\dots,m$, be defined as in Proposition~\ref{prop:block_asymp}. As $n\to\infty$, it holds that
\begin{align}
    \frac{\perm(\matr{A})}{\perm_{\Bethe,2}(\matr{A})}
    \;\sim\;
    \sqrt[4]{\frac{\pi n}{\e}}
    \cdot
    \prod_{i=2}^{m} \sqrt[4]{\frac{\e^{-\eratio_i}}{1-\eratio_i}}.
    \label{eq:ratio_general_exact}
\end{align}
\end{theorem}

\begin{proof}
The statement follows by combining Proposition~\ref{prop:block_asymp} with \eqref{eq:partition:1}--\eqref{eq:partition:2} and taking the square root.
\end{proof}

\begin{corollary}
\label{cor:small_ratio}

We consider the same setup as in Theorem~\ref{thm:general_ratio}. If $\eratio_2 \ll 1$, then
\begin{align}
    \frac{\perm(\matr{A})}{\perm_{\Bethe,2}(\matr{A})}
    \;\sim\;
    \sqrt[4]{\frac{\pi n}{\e}}
    \cdot
    \Biggl(1+O\Biggl(\sum_{i=2}^m \eratio_i^2\Biggr)\Biggr).
    \label{eq:ratio_block_approx}
\end{align}
\end{corollary}

\begin{proof}
Obtained by applying a Taylor expansion to \eqref{eq:ratio_general_exact} in the regime $\eratio_2\ll 1$.
\end{proof}

\section{Numerical Validation}
\label{sec:numerics}

\begin{figure}[t]
    \centering
    \includegraphics[width=0.8\linewidth]{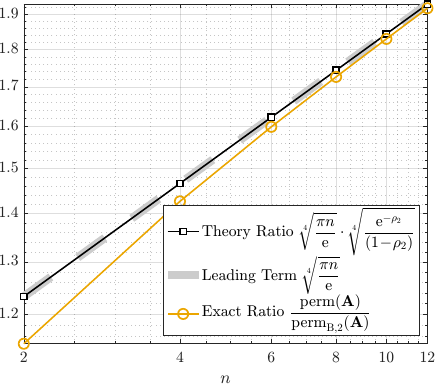}
    \caption{Numerical results discussed in Section~\ref{sec:numerics}.}
    \label{fig:ratiocomparison_full}
\end{figure}

To validate our theory, we consider the following setup:
\begin{itemize}
    \item The integer $n$ takes the values $2, 4, 6, 8, 10, 12$.
    \item The integer $m$ takes the value $2$.
    \item The partitionings take the values $\bm{k} = \bm{\ell} = (n/2; n/2)$.
    \item The matrix $\matr{B}$ has entries $b_{i,j} = q_i^{\mu_j}$ with $(q_1; q_2) = (0.6; 0.4)$ and $(\mu_1; \mu_2) = (2;1)$.
\end{itemize}
For each $n$, we numerically determine the positive saddle point $\bm{w}$ associated with the partitioning $(\bm{k};\bm{\ell})$, and evaluate the spectrum of the symmetric state-transition matrix $\matr{S}(\bm{w})$. For the above setup, we observe that the spectral ratio is small (numerically, $\eratio_2 \approx 0.0102$), making the $\rho_i$-based correction factor in Theorem~\ref{thm:general_ratio} negligible. As shown in Fig.~\ref{fig:ratiocomparison_full}, the ratio $\perm(\matr{A})/\perm_{\Bethe,2}(\matr{A})$ closely follows the predicted $\sqrt[4]{\pi n/\e}$ asymptotic scaling already for small values of $n$.

\section{Conclusion and Outlook}
\label{sec:conclusion:outlook}

The results in this paper significantly extend the class of matrices for which the ratio $\perm(\matr{A})/\perm_{\Bethe,2}(\matr{A})$ can be analyzed asymptotically.

After the initial submission of the present paper, we found that precise asymptotic results can be given not only for $\perm(\matr{A})/\perm_{\Bethe,2}(\matr{A})$, but also for $\perm(\matr{A})/\perm_{\Bethe}(\matr{A})$. These results follow from the fact that the critical point $\bm{w}\defeq (\bm{t^*};\bm{u^*})$ is obtained with the help of the Sinkhorn scaling algorithm and the fact that the SPA behaves similarly to the Sinkhorn scaling algorithm for the matrices under consideration in this paper. These findings will be presented in a forthcoming paper.

\clearpage

\IEEEtriggeratref{15}
\bibliographystyle{IEEEtran}
\bibliography{refs}

\vfill
\noindent
\begin{minipage}{0.95\columnwidth}
\footnotesize
\textsuperscript{*}Compared with the official ISIT version, this arXiv version corrects a few minor typographical issues in the main text and references: an incorrect occurrence of $\hat{\graphN}(\matr A)$ is replaced by ``this NFG'' when referring to the NFG for $(\perm(\matr A))^2$; the block-size condition $k_i,\ell_j\sim n$ is replaced by $k_i\sim\alpha_i n$ and $\ell_j\sim\beta_j n$; and the capitalization of ``SPA-based'' is corrected in the title of reference [18]. The results are unchanged.
\end{minipage}

\clearpage

\appendices

\section{Detailed NFG and Double-Cover Definitions}
\label{app:nfg-detailed}

This section collects the normal factor graph (NFG) definitions and the double-cover construction that underlie the cycle-index viewpoint used in the main text. We follow \cite{Yedidia:Freeman:Weiss:05:1,6352911,6570731,7746637,Ng:Vontobel:22:2,Vontobel2025ITW} and mainly focus on~\cite{Ng:Vontobel:22:2}.

\subsection{NFG $\graphN(\matr{A})$ whose partition function equals $\perm(\matr{A})$}
Let $\matr{A} = (a_{i,j})$ be an arbitrary non-negative matrix of size $n \times n$. We use the NFG $\graphN(\matr{A})$ shown in Fig.~\ref{fig:ffg:permanent:1}~(left), which is the same as in~\cite{Ng:Vontobel:22:2}:

\begin{itemize}
    \item $\graphN(\matr{A})$ is based on the complete bipartite graph with $n$ left nodes $\{f_{\mathrm{L},i}\}_{i\in[n]}$ and $n$ right nodes $\{f_{\mathrm{R},j}\}_{j\in[n]}$.
    
    \item For $i,j\in[n]$, the edge $(f_{\mathrm{L},i},f_{\mathrm{R},j})$ carries a binary variable $x_{i,j}\in\setset{X}\defeq\{0,1\}$.

    \item For $i \in [n]$, let
    \begin{align*}
        \!\!\!\!\!
        f_{\mathrm{L},i}(x_{i,1}, \ldots, x_{i,n})
        &\defeq
        \begin{cases}\!
            \sqrt{a_{i,j}} 
            & \!\!\!\!
            \begin{array}{c}
                \text{$\exists j \in [n]$ s.t. $x_{i,j} = 1$;} \\
                \text{$x_{i,j'} = 0$, $\forall j' \in [n] \setminus \{ j \}$}
            \end{array} \\[0.50cm]
            0 & \text{otherwise}
        \end{cases}
    \end{align*}
    For $j \in [n]$, let
    \begin{align*}
        \!\!\!\!\!
        f_{\mathrm{R},j}(x_{1,j}, \ldots, x_{n,j})
        &\defeq
        \begin{cases}\!
            \sqrt{a_{i,j}} 
            & \!\!\!\!
            \begin{array}{c}
                \text{$\exists i \in [n]$ s.t. $x_{i,j} = 1$;} \\
                 \text{$x_{i',j} = 0$, $\forall i' \in [n] \setminus \{ i \}$}
            \end{array} \\[0.50cm]
            0 & \text{otherwise}
        \end{cases}
    \end{align*}
    
    \item The global function and the partition function are
    \begin{align*}
        g(x_{1,1},\ldots,x_{n,n})
        &\defeq \Big(\prod_{i\in[n]} f_{\mathrm{L},i}\Big)\cdot \Big(\prod_{j\in[n]} f_{\mathrm{R},j}\Big),\\
        Z(\graphN) &\defeq \sum_{x_{1,1},\ldots,x_{n,n}} g(x_{1,1},\ldots,x_{n,n}).
    \end{align*}
\end{itemize}

One checks that $g(x)=\prod_{i\in[n]} a_{i,\sigma(i)}$ if and only if the $1$-entries $\{x_{i,\sigma(i)}\}_{i\in[n]}$ form a permutation pattern, and $g(x)=0$ otherwise. Therefore
\begin{align}
    Z\bigl(\graphN(\matr{A})\bigr)=\perm(\matr{A}).
\end{align}

\subsection{Graph Cover Characterization of the Bethe Approximation}
For NFGs with non-negative local functions, the Bethe approximation $Z_{\Bethe}(\graphN)$ admits the graph-cover characterization \cite{Ng:Vontobel:22:2}
\begin{align*}
    Z_{\Bethe}(\graphN)
    & = \limsup_{M\to\infty} Z_{\Bethe,M}(\graphN), \\
    Z_{\Bethe,M}(\graphN)
    & \defeq \sqrt[M]{\big\langle Z(\graph{\tilde{N}}) \big\rangle_{\graph{\tilde{N}}\in\setset{\tilde{N}}_M}},
\end{align*}
where the average is over all $M$-covers $\graph{\tilde{N}}$ of $\graphN$, $M\ge1$. Specializing to matrices yields the (degree-$M$) Bethe permanent
\begin{align*}
    \perm_{\Bethe,M}(\matr{A})
    &\defeq
    \sqrt[M]{\Big\langle \! \perm(\matr{A}^{\! \uparrow \matr{\tilde{P}}}) \! \Big\rangle_{\matr{\tilde{P}} \in \tilde{\Phi}_M}},
\end{align*}
with
\begin{align*}
    \matr{A}^{\! \uparrow \matr{\tilde{P}}}
    & \defeq
    \begin{pmatrix}
        a_{1,1} \matr{\tilde{P}}^{(1,1)} & \cdots & a_{1,n} \matr{\tilde{P}}^{(1,n)} \\
        \vdots & & \vdots \\
        a_{n,1} \matr{\tilde{P}}^{(n,1)} & \cdots & a_{n,n} \matr{\tilde{P}}^{(n,n)}
    \end{pmatrix}, \\
    \tilde{\Phi}_M
    &\defeq
    \left\{
        \matr{\tilde{P}} 
        = \big\{ \matr{\tilde{P}}^{(i,j)} \big\}_{i \in [n], j \in [n]} \ \middle| \ \matr{\tilde{P}}^{(i,j)} \in \setset{P}_{M \times M}
    \right\},
\end{align*}
where $\setset{P}_{M \times M}$ denotes the set of $M\times M$ permutation matrices. Note that $\matr{A}^{\! \uparrow \matr{\tilde{P}}}$ has size $(Mn)\times(Mn)$.

\subsection{The explicit double cover $\hat\graphN(\matr{A})$ and its partition function}

For the specific case of $M=2$, there exists an explicit NFG $\hat\graphN(\matr{A})$ defined on the same bipartite graph structure as $\graphN(\matr{A})$ shown in Fig.~\ref{fig:ffg:permanent:1}~(right). The construction of this graph relies on a specific change of basis detailed in~\cite{Ng:Vontobel:22:2}.

Conceptually, the degree-$2$ Bethe permanent involves averaging over graph covers where edge variables represent $2 \!\times\! 2$ permutation matrices, allowing for parallel (no-cross) or crossed connections. The construction in~\cite{Ng:Vontobel:22:2} applies an invertible transformation to the local functions, mapping the state space of these coupled variables into a new basis. A key property of this transformation is that, due to the constraints enforcing perfect matchings, the local functions evaluate to zero for the value $(1,0)$. Consequently, cycles whose associated variables take the value $(1,0)$ effectively carry zero weight and do not appear in any valid configuration.

Formally, this results in an NFG $\hat\graphN(\matr{A})$ whose partition function satisfies~\cite[Proposition~1]{Ng:Vontobel:22:2}
\begin{align}
    \perm_{\Bethe,2}(\matr{A}) & = \sqrt{Z\bigl( \hat \graphN(\matr{A}) \bigr)}.
    \label{app:eq:double:cover:app:1}
\end{align}
The edge variables are defined as pairs
\begin{align*}
    \hat{x}_{i,j} \in \hat{\setset{X}} \defeq \setset{X} \times \setset{X} = \{(0,0), (0,1), (1,0), (1,1)\},
\end{align*}
where $\setset{X} \defeq \{0,1\}$. Based on the cancellation of the $(1,0)$ terms described above, the local functions $\hat{f}_{\mathrm{L},i}$ and $\hat{f}_{\mathrm{R},j}$ are given explicitly as follows:
\begin{itemize}
    \item For $i\in[n]$, define $\hat{f}_{\mathrm{L},i}(\hat{x}_{i,1},\ldots,\hat{x}_{i,n})$ by
    \begin{align*}
        \hat{f}_{\mathrm{L},i} \defeq
        \begin{cases}
            a_{i,j}
            & \begin{array}{c}
                \text{$\exists j \in [n]$ s.t.  $\hat{x}_{i,j} = (1,1)$;} \\
                \text{$\hat{x}_{i,j'} = (0,0)$, $\forall j' \in [n] \setminus \{ j \}$}
            \end{array} \\[0.50cm]
            \sqrt{a_{i,j} a_{i,j'}}
            & \begin{array}{c}
                \text{$\exists j, j' \in [n]$, $j \neq j'$, s.t.} \\
                \text{$\hat{x}_{i,j} = (0,1)$, $\hat{x}_{i,j'} = (0,1)$;} \\
                \text{$\hat{x}_{i,j''} = (0,0)$, $\forall j'' \in [n] \setminus \{ j, j' \}$}
            \end{array} \\[0.75cm]
            0 & \text{otherwise}
        \end{cases}
    \end{align*}
    
    \item For $j\in[n]$, define $\hat{f}_{\mathrm{R},j}(\hat{x}_{1,j},\ldots,\hat{x}_{n,j})$ by
    \begin{align*}
        \hat{f}_{\mathrm{R},j} \defeq
        \begin{cases}
            a_{i,j}
            & \begin{array}{c}
                \text{$\exists i \in [n]$ s.t. $\hat{x}_{i,j} = (1,1)$;} \\
                \text{$\hat{x}_{i',j} = (0,0)$, $\forall i' \in [n] \setminus \{ i \}$}
            \end{array} \\[0.50cm]
            \sqrt{a_{i,j} a_{i',j}}
            & \begin{array}{c}
                \text{$\exists i, i' \in [n]$, $i \neq i'$, s.t.} \\
                \text{$\hat{x}_{i,j} = (0,1)$, $\hat{x}_{i',j} = (0,1)$;} \\
                \text{$\hat{x}_{i'',j} = (0,0)$, $\forall i'' \in [n] \setminus \{ i, i' \}$}
            \end{array} \\[0.75cm]
            0 & \text{otherwise}
        \end{cases}
    \end{align*}

    \item The global function and partition function are
    \begin{align*}
        \hat{g}(\hat{x}_{1,1},\ldots,\hat{x}_{n,n})
        &\defeq \Big(\prod_{i\in[n]} \hat{f}_{\mathrm{L},i}\Big)\cdot \Big(\prod_{j\in[n]} \hat{f}_{\mathrm{R},j}\Big),\\
        Z(\hat\graphN) &\defeq \sum_{\hat{x}_{1,1},\ldots,\hat{x}_{n,n}} \hat{g}(\hat{x}_{1,1},\ldots,\hat{x}_{n,n}).
    \end{align*}
\end{itemize}

\subsection{Combinatorial Interpretation via the Cycle Penalty}

To understand the difference between $\bigl( \perm(\matr{A}) \bigr)^2$ and $\bigl( \perm_{\Bethe,2}(\matr{A}) \bigr)^2$, we compare the partition sum of the restricted NFG defined above with that of the squared permanent $\bigl( \perm(\matr{A}) \bigr)^2$.

In the specific NFG $\hat\graphN(\matr{A})$, the local functions defined above explicitly forbid the value $(1,0)$. Consequently, for every cycle formed by the underlying permutations $\sigma_1$ and $\sigma_2$, only the realization using $(0,1)$ edges contributes to the partition sum.

In contrast, similar to the construction of $\hat\graphN(\matr{A})$, an NFG can be constructed such that its partition sum equals $\bigl( \perm(\matr{A}) \bigr)^2$. Crucially, the local functions of this NFG are defined such that valid configurations include simple cycles whose associated variables take the value $(0,1)$, as well as those taking the value $(1,0)$ that are explicitly forbidden in $\hat\graphN(\matr{A})$. Consequently, in this unconstrained case, a cycle of length $h \ge 2$ can be realized by the edge variables in two ways: either via edges taking the value $(1,0)$ or via edges taking the value $(0,1)$. Both realizations contribute equally to the partition sum.

Comparing the two cases, we see that for every cycle of length $h \ge 2$, the partition sum of $\hat\graphN(\matr{A})$ captures only one of the two possible realizations available in the NFG for the squared permanent. This effectively introduces a penalty factor of $1/2$ per cycle. Summing over all permutation pairs yields the identity~\cite{Ng:Vontobel:22:2} as
\begin{align}
    \frac{\perm_{\Bethe,2}(\matr{A})}{\perm(\matr{A})}
    \;=\;
    \sqrt{\; \sum_{\sigma_1,\sigma_2\in\Sn} p(\sigma_1)\,p(\sigma_2)\,2^{-c(\sigma_1\circ\sigma_2^{-1})}\;},
    \label{app:eq:pair-perm-identity-app}
\end{align}
where $p(\sigma)\defeq (\prod_{i} a_{i,\sigma(i)})/\perm(\matr{A})$ and $c(\sigma_1\circ\sigma_2^{-1})$ counts the number of cycles of length $>1$ in the cycle decomposition of $(\sigma_1\circ\sigma_2^{-1})$. For full technical details, see~\cite{Ng:Vontobel:22:2}.

\section{Factor-Graph View of the Trace Kernel $\matr{S}$}
\label{app:fgtrace}

This section provides a short NFG-based view behind the symmetric state-transition matrix $\matr{S}$ in Definition~\ref{def:S-general}. To this end, we first introduce an asymmetric weighted transition matrix $\matr{W}$ to capture the fundamental walk operations. The goal is to justify (i) why the indeterminates $\bm{t},\bm{u}$ act as counting indeterminates in the generating functions, and (ii) why the relevant cycle weights collapse to trace powers $\trace(\matr{W}^h)$ (equivalently $\trace(\matr{S}^h)$).

\subsection{Counting indeterminates, partitioning selection, and the weighted transition}

Under Assumption~\ref{assump:general_block}, the $n\times n$ matrix $\matr{A}=\matr{A}(\matr{B},\bm{k},\bm{\ell})$ is expanded from the matrix $\matr{B}=(b_{i,j})\in\mathbb{R}_{>0}^{m\times m}$ with fixed multiplicities $\bm{k},\bm{\ell}\in\mathbb{Z}_{>0}^m$. To select a prescribed partitioning in a symmetry-reduced enumeration, we attach indeterminates $\bm{t}=(t_1;\dots;t_m)$ to row types and $\bm{u}=(u_1;\dots;u_m)$ to column types. Each use of a row (column) of type $i$ ($j$) contributes a multiplicative factor $t_i$ ($u_j$) to the total weight of the walk. Since the weight of a walk is the product of its step weights, visiting type $i$ exactly $k_i$ times results in the accumulation of the factor $t_i$ to the power of $k_i$. Consequently, configurations associated with the partitioning $(\bm{k};\bm{\ell})$ naturally contribute the monomial $\bm{t}^{\bm{k}}\bm{u}^{\bm{\ell}}$, allowing the partitioning constraints to be enforced via the coefficient extraction $[\bm{t}^{\bm{k}}\bm{u}^{\bm{\ell}}]$.

For the cycle-index enumeration, we need the aggregated weight of length-$h$ closed walks on the type graph. The basic local pattern is two consecutive columns coupled through the same row. A step from a column type $j$ to a column type $j'$ passes through an intermediate row type $i$ and contributes the product of the two elemental entries on that row, $b_{i,j}\,b_{i,j'}$, together with the counting markers $t_i$ and $u_{j'}$. Thus the step weight is $b_{i,j}\,t_i\,b_{i,j'}\,u_{j'}$. Summing (marginalizing) over the intermediate row type yields the one-step kernel $W_{jj'} \defeq \sum_{i\in[m]} b_{i,j}\,t_i\,b_{i,j'}\,u_{j'}$. This motivates the formal definition of the fundamental matrix governing the walk transitions.

\begin{definition}
\label{app:def:W-operator}
The $m \times m$ weighted transition matrix $\matr{W}$ is defined by
\begin{align}
    \matr{W} \defeq \matr{B}^\transp \cdot \diag(\bm{t})\cdot \matr{B} \cdot\diag(\bm{u}).
    \label{app:eq:Wdef}
\end{align}
\end{definition}

Summing over all closed type sequences $(j_1,\dots,j_h)$ with $j_{h+1}\equiv j_1$ gives the standard trace identity
\begin{align*}
    \sum_{j_1,\dots,j_h} W_{j_1j_2}\cdots W_{j_hj_1} \;=\; \trace(\matr{W}^h),
\end{align*}
so $\trace(\matr{W}^h)$ packages the total weight of all length-$h$ cycles, while $\bm{t},\bm{u}$ record the type counts.

For $m\!=\!2$, this viewpoint can be visualized as a trellis over the column-type state space. The structure of this trellis is illustrated in Fig.~\ref{fig:trellisgraph_t6_s2}. It consists of two states at each of the $T+1=7$ time steps ($\tau=0$ to $6$), with all transitions between consecutive time steps allowed. Each edge from state $j$ to $j'$ carries weight $W_{jj'}$. Consequently, summing the weights of all closed length-$T$ paths in the trellis is exactly $\trace(\matr{W}^T)$.

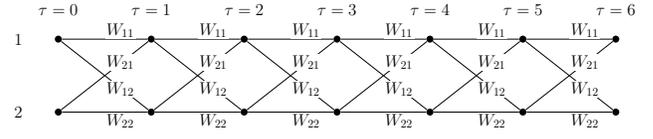
\begin{figure}[t]
    \centering
    \resizebox{0.95\linewidth}{!}{\begin{tikzpicture}[
    state/.style={circle, fill=black, inner sep=0pt, minimum size=6pt},
    line/.style={thick},
    labelstyle/.style={midway, fill=white, font=\Large, inner sep=1pt}
]

\def\T{7}
\def\Xgap{2.8}
\def\Ygap{2.2}

\foreach \t in {0,...,6} {
    \node[state] (a\t) at (\t*\Xgap, 0) {};
    \node[state] (b\t) at (\t*\Xgap, -\Ygap) {};
    \node at (\t*\Xgap, 0.9) {\Large $\tau=\t$};
}

\node at (-1.2, 0) {\Large $\mathrm{1}$};
\node at (-1.2, -\Ygap) {\Large $\mathrm{2}$};

\foreach \t in {0,...,5} {
    \pgfmathtruncatemacro{\next}{\t + 1}
    \draw[line] (a\t) -- node[labelstyle, above right] {$W_{11}$} (a\next);
    \draw[line] (a\t) -- node[labelstyle, below right, yshift=-4pt] {$W_{12}$} (b\next);
    \draw[line] (b\t) -- node[labelstyle, above right, yshift=4pt] {$W_{21}$} (a\next);
    \draw[line] (b\t) -- node[labelstyle, below right] {$W_{22}$} (b\next);
}
\end{tikzpicture}}
    \caption{Trellis graph with two states over $7$ time steps.}
    \label{fig:trellisgraph_t6_s2}
    \vspace{-0.5em}
\end{figure}

\subsection{Symmetrization and the NFG kernel factorization}
The weighted transition matrix $\matr{W}$ in \eqref{app:eq:Wdef} is a convenient walk kernel, but it is generally not symmetric because the column-type marker $u_{j'}$ is attached to the arrival type. For the spectral manipulations used in the main text, it is preferable to pass to a real symmetric kernel without changing trace powers $\trace(\matr{W}^h)$.

\begin{figure*}[t]
    \centering
    \begin{subfigure}{0.95\textwidth}
        \centering
        \resizebox{\textwidth}{!}{\begin{tikzpicture}[
    node/.style={draw, rectangle, minimum size=0.7cm, fill=white},
    dot/.style={circle, fill=black, minimum size=2pt, inner sep=0pt},
    bgbox/.style={draw=red, dashed, line width=1pt, fill=magenta!12, rounded corners=2pt},
    >=latex
]

\pgfdeclarelayer{bg}
\pgfsetlayers{bg,main}

\def\xsep{1.8}

\foreach \i in {0,1,2} {
    \pgfmathsetmacro{\xe}{2*\i*\xsep}
    \pgfmathsetmacro{\xf}{(2*\i+1)*\xsep}
    \pgfmathtruncatemacro{\label}{\i + 1}

    \node[node] (e\i) at (\xe, 0) {$=$};
    \node[node] (f\i) at (\xf, 0) {$f$};
    \draw (e\i) -- (f\i);
    \ifnum\i>0
        \pgfmathtruncatemacro{\iprev}{\i - 1}
        \draw (f\iprev) -- (e\i);
    \fi

    \node[node] (ps\i) at ($(e\i) + (0, -1.5)$) {};
    \node (pslabel\i) [below=0.2cm of ps\i] {\large $u_{j_{\label}}$};
    \draw (e\i) -- (ps\i);

    \node[node] (pq\i) at ($(f\i) + (0, 1.5)$) {};
    \node[above=0.02cm of pq\i] {\large $t_{i_{\label}}$};
    \draw (pq\i) -- (f\i);

    \begin{pgfonlayer}{bg}
        \node[bgbox, fit=(e\i)(ps\i), inner sep=6pt] {};
    \end{pgfonlayer}
}

\pgfmathsetmacro{\xeFour}{6*\xsep}
\node[node] (e3) at (\xeFour, 0) {$=$};
\draw (f2) -- (e3);
\node[node] (ps3) at ($(e3) + (0, -1.5)$) {};
\node (pslabel3) [below=0.2cm of ps3] {\large $u_{j_4}$};
\draw (e3) -- (ps3);

\begin{pgfonlayer}{bg}
    \node[bgbox, fit=(e3)(ps3), inner sep=4pt] {};
\end{pgfonlayer}

\coordinate (afterE3) at ($(e3.east) + (0.8, 0)$);
\draw (e3) -- (afterE3);

\coordinate (dotsLeft) at ($(afterE3) + (0, 0)$);
\coordinate (dotsRight) at ($(dotsLeft) + (1.8, 0)$);
\node[dot] at ($(dotsLeft)!0.3!(dotsRight)$) {};
\node[dot] at ($(dotsLeft)!0.5!(dotsRight)$) {};
\node[dot] at ($(dotsLeft)!0.7!(dotsRight)$) {};
\draw (dotsRight) -- ++(0.8, 0);

\coordinate (beforeEkm1) at ($(dotsRight) + (0.8, 0)$);
\node[node] (ekm1) at ($(beforeEkm1) + (0.35, 0)$) {$=$};
\node[node] (pskm1) at ($(ekm1) + (0, -1.5)$) {};
\node (pslabkm1) [below=0.2cm of pskm1] {\large $u_{j_{h-1}}$};
\draw (ekm1) -- (pskm1);

\begin{pgfonlayer}{bg}
    \node[bgbox, fit=(ekm1)(pskm1), inner sep=4pt] {};
\end{pgfonlayer}

\node[node] (fkm1) at ($(ekm1) + (\xsep, 0)$) {$f$};
\draw (ekm1) -- (fkm1);
\node[node] (pqkm1) at ($(fkm1) + (0,1.5)$) {};
\node[above=0.02cm of pqkm1] {\large $t_{i_{h-1}}$};
\draw (pqkm1) -- (fkm1);

\node[node] (ek) at ($(fkm1) + (\xsep, 0)$) {$=$};
\node[node] (psk) at ($(ek) + (0, -1.5)$) {};
\node (pslabk) [below=0.2cm of psk] {\large $u_{j_h}$};
\draw (fkm1) -- (ek);
\draw (ek) -- (psk);

\begin{pgfonlayer}{bg}
    \node[bgbox, fit=(ek)(psk), inner sep=4pt] {};
\end{pgfonlayer}

\node[node] (fk) at ($(ek) + (\xsep, 0)$) {$f$};
\draw (ek) -- (fk);
\node[node] (pqk) at ($(fk) + (0, 1.5)$) {};
\node[above=0.02cm of pqk] {\large $t_{i_h}$};
\draw (pqk) -- (fk);

\draw
    (fk.east) -- ++(0.8,0) -- ++(0,2.8) -- ++(-22.6,0) -- ++(0,-2.8) -- (e0.west);

\end{tikzpicture}}
        \caption{NFG whose partition sum aggregates length-$h$ closed walks on the type graph. The $\nu$-th visited column or row contributes markers $t_{i_\nu}$ and $u_{j_\nu}$ with $i_\nu,j_\nu\in\{1,\dots,m\}$. The factor node $f$ encodes the local elemental weight $b_{i_\nu j_\nu}\,t_{i_\nu}\,b_{i_\nu j_{\nu+1}}$ along one step. Red Open The Box (OTB) regions split each $u_{j_\nu}$ into two adjacent factors $\sqrt{u_{j_\nu}}$.}

        \label{fig:nfg_TransitionMatrix_Z_OTB_P_S-app}
    \end{subfigure}
    \vspace{1em}

    \begin{subfigure}{0.95\textwidth}
        \centering
        \resizebox{\textwidth}{!}{\begin{tikzpicture}[
  node/.style={draw, rectangle, minimum size=0.7cm, fill=white},
  dot/.style={circle, fill=black, minimum size=2pt, inner sep=0pt},
  bgblue/.style={
    draw=blue, dashed, line width=1pt, fill=blue!15, rounded corners=2pt,
    inner xsep=2pt, inner ysep=2pt
  },
  >=latex
]

\pgfdeclarelayer{bg}
\pgfsetlayers{bg,main}

\def\xsep{1.8}
\def\xhalfsep{1.4}

\foreach \i in {0,1,2} {
  \pgfmathsetmacro{\xel}{3*\i*\xhalfsep}
  \pgfmathsetmacro{\xer}{(2+3*\i)*\xhalfsep}
  \pgfmathsetmacro{\xf}{(3*\i+1)*\xhalfsep}
  \pgfmathtruncatemacro{\label}{\i + 1}

  \node[node] (el\i) at (\xel, 0) {$=$};
  \node[node] (f\i)  at (\xf, 0)  {$f$};
  \node[node] (er\i) at (\xer, 0) {$=$};
  
  \draw (el\i) -- (f\i);
  \draw (f\i) -- (er\i);
  \ifnum\i>0
    \pgfmathtruncatemacro{\iprev}{\i - 1}
    \draw (er\iprev) -- (el\i);
  \fi

  \node[node] (psl\i) at ($(el\i) + (0, -1.5)$) {};
  \node[node] (psr\i) at ($(er\i) + (0, -1.5)$) {};
  \node (psllabel\i) [below=0.2cm of psl\i, xshift = 0.1cm] {\Large $\sqrt{u_{j_{\label}}}$};
  \pgfmathtruncatemacro{\nextlabel}{\label + 1}
  \node (psrlabel\i) [below=0.2cm of psr\i, xshift = -0.1cm] {\Large $\sqrt{u_{j_{\nextlabel}}}$};

  \draw (el\i) -- (psl\i);
  \draw (er\i) -- (psr\i);

  \node[node] (pq\i) at ($(f\i) + (0, 1.5)$) {};
  \node (pqlabel\i) [above=0.1cm of pq\i] {\Large $t_{i_{\label}}$};
  \draw (pq\i) -- (f\i);

  \pgfmathtruncatemacro{\slabel}{\i + 1}
  \begin{pgfonlayer}{bg}
    \node[bgblue,
      fit=(el\i)(f\i)(er\i)(psl\i)(psr\i)(pq\i)(psllabel\i)(psrlabel\i)(pqlabel\i)
    ] (group\i) {};
  \end{pgfonlayer}
  \node[text=blue, below=4pt of group\i.south] {\Large $S_{j_{\label},j_{\nextlabel}}$};
}

\pgfmathsetmacro{\xeFour}{9*\xhalfsep}
\node[node] (e3) at (\xeFour, 0) {$=$};
\draw (er2) -- (e3);
\node[node] (ps3) at ($(e3) + (0, -1.5)$) {};
\node (pslabel3) [below=0.2cm of ps3] {\Large $\sqrt{u_{j_4}}$};
\draw (e3) -- (ps3);

\coordinate (afterE3) at ($(e3.east) + (0.8, 0)$);
\draw (e3) -- (afterE3);

\coordinate (dotsLeft) at ($(afterE3) + (0, 0)$);
\coordinate (dotsRight) at ($(dotsLeft) + (1.8, 0)$);
\node[dot] at ($(dotsLeft)!0.3!(dotsRight)$) {};
\node[dot] at ($(dotsLeft)!0.5!(dotsRight)$) {};
\node[dot] at ($(dotsLeft)!0.7!(dotsRight)$) {};
\draw (dotsRight) -- ++(0.8, 0);

\node[node] (erkm2) at ($(dotsRight) + (1.15, 0)$) {$=$};
\node[node] (psrkm2) at ($(erkm2) + (0, -1.5)$) {};
\node (pslaberkm2) [below=0.2cm of psrkm2, xshift = -0.2cm] {\Large $\sqrt{u_{j_{h-1}}}$};
\draw (erkm2) -- (psrkm2);

\node[node] (elkm1) at ($(erkm2) + (\xhalfsep, 0)$) {$=$};
\draw (erkm2) -- (elkm1);
\node[node] (pslkm1) at ($(elkm1) + (0, -1.5)$) {};
\node (pslabkm1) [below=0.2cm of pslkm1, xshift = 0.2cm] {\Large $\sqrt{u_{j_{h-1}}}$};
\draw (elkm1) -- (pslkm1);

\node[node] (fkm1) at ($(elkm1) + (\xhalfsep, 0)$) {$f$};
\draw (elkm1) -- (fkm1);
\node[node] (pqkm1) at ($(fkm1) + (0,1.5)$) {};
\node (pqkm1label) [above=0.1cm of pqkm1] {\Large $t_{i_{h-1}}$};
\draw (pqkm1) -- (fkm1);

\node[node] (erkm1) at ($(fkm1) + (\xhalfsep, 0)$) {$=$};
\node[node] (psrkm1) at ($(erkm1) + (0, -1.5)$) {};
\node (psrabkm1) [below=0.2cm of psrkm1, xshift = -0.1cm] {\Large $\sqrt{u_{j_h}}$};
\draw (fkm1) -- (erkm1);
\draw (erkm1) -- (psrkm1);

\begin{pgfonlayer}{bg}
  \node[bgblue,
    fit=(elkm1)(fkm1)(erkm1)(pslkm1)(psrkm1)(pqkm1label)(pslabkm1)(psrabkm1)(pqkm1label)
  ] (groupkm1) {};
\end{pgfonlayer}

\node[text=blue, below=4pt of groupkm1.south] {\Large $S_{j_{h-1},j_{h}}$};

\node[node] (elk) at ($(erkm1) + (\xhalfsep, 0)$) {$=$};
\node[node] (pslk) at ($(elk) + (0, -1.5)$) {};
\node (pslabk) [below=0.2cm of pslk, xshift = 0.1cm] {\Large $\sqrt{u_{j_h}}$};
\draw (erkm1) -- (elk);
\draw (elk) -- (pslk);

\node[node] (fk) at ($(elk) + (\xhalfsep, 0)$) {$f$};
\draw (elk) -- (fk);
\node[node] (pqk) at ($(fk) + (0, 1.5)$) {};
\node (pqklabel) [above=0.1cm of pqk] {\Large $t_{i_h}$};
\draw (pqk) -- (fk);

\node[node] (erk) at ($(fk) + (\xhalfsep, 0)$) {$=$};
\node[node] (psrk) at ($(erk) + (0, -1.5)$) {};
\node (psrabk) [below=0.2cm of psrk, xshift = -0.2cm] {\Large $\sqrt{u_{j_1}}$};
\draw (fk) -- (erk);
\draw (erk) -- (psrk);

\begin{pgfonlayer}{bg}
    \coordinate (groupk_rpad) at ($(erk.east)+(4pt,0)$);
    \node[bgblue,
        fit=(elk)(fk)(erk)(pslk)(psrk)(pqk)(pslabk)(psrabk)(pqklabel)(groupk_rpad)
    ] (groupk) {};
\end{pgfonlayer}
\node[text=blue, below=4pt of groupk.south] {\Large $S_{j_{h},j_{1}}$};

\draw
  (erk.east) -- ++(0.5,0) -- ++(0,3.2) -- ++(-26.8,0) -- ++(0,-3.2) -- (el0.west);

\end{tikzpicture}}
        \caption{After CTB (blue), each grouped box becomes a symmetric kernel entry $S_{j_\nu j_{\nu+1}}$, and the length-$h$ cycle evaluates to $\trace(\matr{S}^h)$ (equivalently $\trace(\matr{W}^h)$).}

        \label{fig:nfg_TransitionMatrix_Z_CTB_sqrtP_S-app}
    \end{subfigure}
    \caption{OTB/CTB factor-graph derivation of the kernel factorization $\matr{W}\mapsto \matr{S}$ and the trace representation of length-$h$ closed walks.}
    \label{fig:nfg_TransitionMatrix_all-app}
\end{figure*}
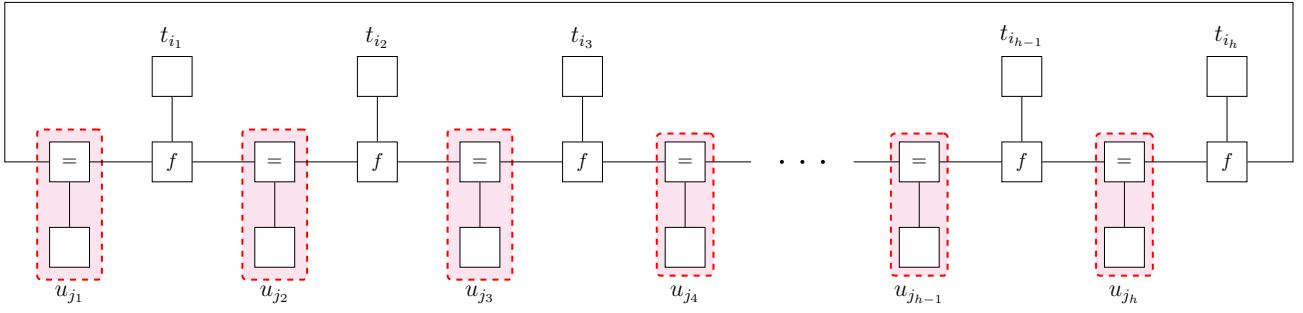
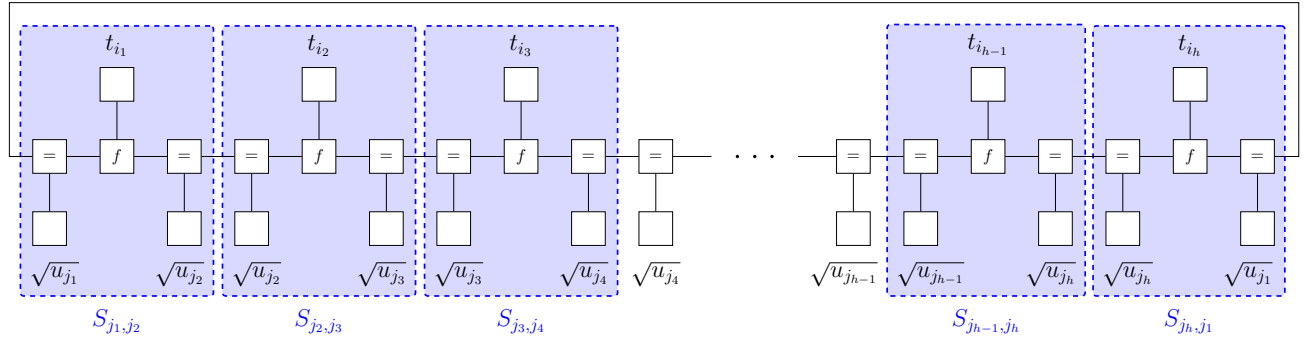

This is done by splitting each column-type marker across the two half-edges incident to a step, i.e., replacing $u_{j'}$ by $\sqrt{u_j}\sqrt{u_{j'}}$ at the level of a transition $j\to j'$. Equivalently, we apply the diagonal similarity transform to obtain the symmetric state-transition matrix $\matr{S}(\bm{t},\bm{u})$ introduced in Definition~\ref{def:S-general}, yielding
\begin{align}
    \matr{S}(\bm{t},\bm{u})
    & \;\defeq\;
    \diag(\bm{u})^{1/2}\cdot\matr{W}\cdot\diag(\bm{u})^{-1/2}\nonumber\\
    & \;=\;
    \diag(\bm{u})^{1/2}\cdot\matr{B}^\transp \cdot\diag(\bm{t})\cdot\matr{B}\cdot\diag(\bm{u})^{1/2}.
    \tag{\ref{eq:Sdef}}
\end{align}
Since $\matr{S}$ and $\matr{W}$ are similar, they preserve all trace powers,
\begin{align}
    \trace(\matr{S}^h)=\trace(\matr{W}^h),\qquad h\ge 1.
\label{app:eq:traceeq}
\end{align}
Moreover, writing $\bm{a}_j$ for the $j$-th column of $\matr{B}$, \eqref{eq:Sdef} yields the Gram form
\begin{align*}
    S_{j,j'}
    =
    \sqrt{u_j u_{j'}}\,\langle \bm{a}_j,\bm{a}_{j'}\rangle_{\bm{t}},
    \qquad
    \langle \bm{x},\bm{y}\rangle_{\bm{t}}\defeq \bm{x}^\transp \diag(\bm{t})\,\bm{y},
\end{align*}
so whenever $\bm{t},\bm{u}\in\mathbb{R}_{>0}^m$ (as holds in Proposition~\ref{prop:block_asymp}), the matrix $\matr{S}(\bm{t},\bm{u})$ is real symmetric and positive semidefinite.

Figure~\ref{fig:nfg_TransitionMatrix_all-app} illustrates this procedure in NFG language. One first applies Open The Box (OTB) to the nodes representing the column indeterminates (red regions). By factorizing the column-type weight $u_j$ symmetrically into $\sqrt{u_j}$ on the adjacent half-edges, we prepare the edges to carry the combined interaction weights. The factor node $f$ represents the local row-mediated interaction: for a fixed intermediate row type $i$, it contributes the elemental weight $b_{i,j}\,t_i\,b_{i,j'}$ that couples two consecutive column types $j\to j'$ through that row type. One then applies Close The Box (CTB) to regroup the components (blue regions). Upon this regrouping, each resulting box abstracts exactly to one entry of the symmetric kernel $\matr{S}$ in~\eqref{eq:Sdef}. Consequently, a length-$h$ cycle becomes a closed sequence of $h$ such kernel boxes, whose total contribution is $\trace(\matr{S}^h)$ as in~\eqref{app:eq:traceeq}.

\section{Proof of Proposition~\ref{prop:block_asymp}}
\label{app:acsv_derivation}

This section provides the derivation of the asymptotic partition sums using analytic combinatorics in several variables (ACSV)~\cite{pemantle2024analytic}.

First, recall the relevant setup and definitions. The generating functions for the partition sums, derived in Lemma~\ref{lem:block_trace_egf}, take the explicit forms as, respectively,
\begin{align}
    \CGn(\bm{t},\!\bm{u})
    & = \frac{1}{\det\bigl(\matr{I}\!-\!\matr{S}(\bm{t},\!\bm{u})\bigr)}, \tag{\ref{eq:CGCB-det:1}} \\
    \CDBn(\bm{t},\!\bm{u})
    & = \frac{\exp\Bigl(\tfrac{1}{2}\trace\bigl(\matr{S}(\bm{t},\!\bm{u})\bigr)\Bigr)}{\sqrt{\det\bigl(\matr{I}\!-\!\matr{S}(\bm{t},\!\bm{u})\bigr)}}. \tag{\ref{eq:CGCB-det:2}}
\end{align}
Here, $\bm{z} \defeq (\bm{t}; \bm{u})$ denotes the collection of $2m$ indeterminates.

As established in Lemma~\ref{lem:bridge-ogf-nfg}, our primary quantities of interest correspond to the coefficients of these functions associated with the partitioning $\bm{r} \defeq (\bm{k}; \bm{\ell})$ given by, respectively,
\begin{align}
    \ZGn(\bm{k}, \bm{\ell}) & = [\bm{z}^{\bm{r}}] \, \CGn(\bm{z}), \\
    \ZDBn(\bm{k}, \bm{\ell}) & = [\bm{z}^{\bm{r}}] \, \CDBn(\bm{z}).
\end{align}

To treat both cases within a unified framework, we adopt the standard ACSV form in~\cite{pemantle2024analytic}, i.e.,
\begin{align}
    F(\bm{z}) = P_0(\bm{z})\,Q_0(\bm{z})^{-\alpha}.
\end{align}
Here $\alpha$ distinguishes the Gibbs ($\alpha \!=\! 1$) and Bethe ($\alpha \!=\! 1/2$) cases. In the following, we will refine this generic form by an equivalent local parametrization near the dominant singularity, namely, we will replace the denominator $Q_0$ by a more convenient local defining function $Q$ for the same singular hypersurface and absorb the remaining non-vanishing factors into the analytic prefactor. This yields the local forms used for coefficient asymptotics, with $\alpha\!=\!1$ and $\alpha\!=\!1/2$ treated in parallel (see Lemma~\ref{app:lem:acsv_factorization}).

\begin{definition}
\label{app:def:singular_variety}
Define the singular function
\begin{align}
    Q_0(\bm{z}) \;\defeq\; \det\bigl(\matr{I}\!-\!\matr{S}(\bm{z})\bigr),
\end{align}
and the associated singular variety (i.e., the zero set of the denominator shared by $\CGn$ and $\CDBn$) as
\begin{align}
    \mathcal{V} \;\defeq\; \bigl\{\bm{z}\in\mathbb{C}^{2m}\mid Q_0(\bm{z})=0\bigr\}.
\end{align}
\end{definition}

\subsection{Spectral Analysis and Singular Decomposition}

We first establish the existence of a minimal singularity in the positive orthant and then analyze its local geometry to justify the use of smooth critical-point analysis.

\begin{lemma}
\label{app:lem:vp_minimal}
Let $F(\bm{z})$ denote either $\CGn(\bm{z})$ in~\eqref{eq:CGCB-det:1} or $\CDBn(\bm{z})$ in~\eqref{eq:CGCB-det:2}. Then there exists a $\bm{w}\in \mathcal{V}$ whose moduli vector $(|w_1|;\ldots;|w_{2m}|)$ is minimal among all points in $\mathcal{V}$. Moreover, such a minimal point may be taken in the positive real orthant, i.e., $\bm{w}\in\mathcal{V}\cap\mathbb{R}_{>0}^{2m}$.
\end{lemma}

\begin{proof}
As established in Lemma~\ref{lem:bridge-ogf-nfg}, the quantities $\ZGn(\bm{k},\bm{\ell})$ and $\ZDBn(\bm{k},\bm{\ell})$ are the coefficients of the monomial $\bm{t}^{\bm{k}}\bm{u}^{\bm{\ell}}$ in the expansions of $\CGn(\bm{z})$ and $\CDBn(\bm{z})$, respectively. In particular, these generating functions admit power series expansions about the origin with non-negative coefficients. Consequently, the multivariate Vivanti--Pringsheim theorem~\cite[Proposition~6.38]{pemantle2024analytic} yields a minimal point on the boundary of the domain of convergence whose moduli are minimal; moreover, one such minimal point lies in the positive real orthant. Hence $\bm{w}\in\mathcal{V}\cap\mathbb{R}_{>0}^{2m}$.
\end{proof}

Having established the existence of a minimal point $\bm{w}\in\mathcal{V}\cap\mathbb{R}_{>0}^{2m}$, we now characterize the local structure of $\mathcal{V}$ around $\bm{w}$.

\begin{definition}
\label{app:def:Q_pf}
Define the function
\begin{align}
    Q(\bm{z}) \;\defeq\; 1\!-\!\lambda_1(\bm{z}),
\end{align}
where $\lambda_1(\bm{z})$ is the Perron eigenvalue of $\matr{S}(\bm{z})$ for $\bm{z}\in\mathbb{R}_{>0}^{2m}$.
\end{definition}

\begin{lemma}
\label{app:lem:variety_geometry}
Let $\bm{w}\in\mathcal{V}\cap\mathbb{R}_{>0}^{2m}$ be such that $\lambda_1(\bm{w})=1$. Then, in a neighborhood of $\bm{w}$, the singular variety $\mathcal{V}$ is locally equivalent to the hypersurface $\bigl\{ \bm{z}\!\in\!\mathbb{C}^{2m}\! \bigm| Q(\bm{z}) \!=\! 0 \bigr\}$. Moreover, $\mathcal{V}$ is smooth at $\bm{w}$.
\end{lemma}

\begin{proof}
For $\bm{z} \in \mathbb{R}_{>0}^{2m}$, the matrix $\matr{S}(\bm{z})$ has strictly positive entries that are strictly increasing functions of the entries of $\bm{z}$. Consequently, the Perron eigenvalue $\lambda_1(\bm{z})$ is strictly increasing along any ray in $\mathbb{R}_{>0}^{2m}$ extending from the origin.

Consider a minimal positive real singularity $\bm{w}$ whose existence is ensured by Lemma~\ref{app:lem:vp_minimal}. Since $\bm{w} \in \mathcal{V} \cap \mathbb{R}_{>0}^{2m}$, we have $\lambda_1(\bm{w}) \!=\! 1$. For $2 \!\leq\! i \!\leq\! m$, the strict dominance of the Perron eigenvalue implies that $|\lambda_i(\bm{w})| < 1$, and by continuity, the factors $1 \!-\! \lambda_i(\bm{z})$ remain non-zero in a neighborhood of $\bm{w}$. Consequently, for $\bm{z}\in\mathbb{C}^{2m}$ in a neighborhood of $\bm{w}$ we have
\begin{align}
    Q_0(\bm{z})=0 \quad \Longleftrightarrow \quad Q(\bm{z})=0,
\end{align}
hence $\mathcal{V}$ is locally given by the hypersurface $Q(\bm{z})=0$.

Furthermore, the monotonicity of the matrix entries implies that $\nabla \lambda_1(\bm{w})$ is non-vanishing (equivalently, $\nabla Q(\bm{w}) \neq \bm{0}$). By the smoothness criterion in~\cite[Lemma~7.6]{pemantle2024analytic}, this guarantees that the singular variety $\mathcal{V}$ is smooth at $\bm{w}$.
\end{proof}

Based on the local geometry established in Lemma~\ref{app:lem:variety_geometry}, we can refine the generic ACSV form $F=P_0 Q_0^{-\alpha}$ by isolating the dominant singular factor $Q(\bm{z}) = 1\!-\!\lambda_1(\bm{z})$ from the determinant.

\begin{lemma}
\label{app:lem:acsv_factorization}
Let $Q(\bm{z})\defeq 1\!-\!\lambda_1(\bm{z})$ and let $\bm{w}$ be the minimal smooth critical point. In a neighborhood of $\bm{w}$, the generating functions can be written in the ACSV form as, respectively,
\begin{align}
    \CGn(\bm{z}) & = \PGn(\bm{z})\,Q(\bm{z})^{-1},
    \\
    \CDBn(\bm{z}) & = \PDBn(\bm{z})\,Q(\bm{z})^{-1/2},
\end{align}
where the analytic factors are given by
\begin{itemize}
    \item Gibbs Case ($\alpha=1$):
    \begin{align}
        \PGn(\bm{z}) = \prod_{i=2}^m \frac{1}{1\!-\!\lambda_i(\bm{z})}.
        \label{app:eq:pgn}
    \end{align}
    \item Bethe Case ($\alpha=1/2$):
    \begin{align}
        \PDBn(\bm{z}) = \e^{\lambda_1(\bm{z})/2} \prod_{i=2}^m \sqrt{\frac{\e^{\lambda_i(\bm{z})}}{1\!-\!\lambda_i(\bm{z})}}.
        \label{app:eq:pbn}
    \end{align}
\end{itemize}
Moreover, $\PGn$ and $\PDBn$ are holomorphic and non-zero in a neighborhood of $\bm{w}$.
\end{lemma}
\begin{proof}
By Lemma~\ref{app:lem:variety_geometry}, we have $|\lambda_i(\bm{w})|<1$ for all $2 \!\leq\! i \!\leq\! m$, hence $1\!-\!\lambda_i(\bm{w})\neq0$ and therefore $\prod_{i=2}^m(1\!-\!\lambda_i(\bm{w}))\neq0$. Using the eigenvalue factorization, i.e.,
\begin{align}
    \det\bigl(\matr{I}\!-\!\matr{S}(\bm{z})\bigr) = \bigl(1\!-\!\lambda_1(\bm{z})\bigr) \prod_{i=2}^m \bigl(1\!-\!\lambda_i(\bm{z})\bigr),
\end{align}
and substituting this into the explicit forms in Lemma~\ref{lem:block_trace_egf} yields \eqref{app:eq:pgn} and \eqref{app:eq:pbn}. The non-vanishing of the ($1\!-\!\lambda_i$)-factor for $2 \!\leq\! i \!\leq\! m$ at $\bm{w}$ implies that $\PGn$ and $\PDBn$ are holomorphic and non-zero in a neighborhood of $\bm{w}$.
\end{proof}

\begin{definition}
\label{def:log_gradient}
We define the logarithmic gradient of $Q(\bm{z})$ to be
\begin{align}
    \nabla_{\log} Q(\bm{z})
    \;\defeq\;
    \left(\frac{\partial Q(\bm{z})}{\partial \log z_i}\right)_{\!i=1}^{\!2m}
    \;=\;
    \left(z_i\frac{\partial Q(\bm{z})}{\partial z_i}\right)_{\!i=1}^{\!2m}.
\end{align}
\end{definition}

Given the smoothness and local characterization of $\mathcal{V}$ at $\bm{w}$, the asymptotic behavior of the coefficients is determined by the smooth critical point equations (see, e.g., \cite[Lemma~7.8]{pemantle2024analytic}).

\begin{lemma}
\label{app:lem:critical_point}
There exists a Lagrange scaling parameter $\eta$ such that the minimal point $\bm{w}\in\mathcal{V}\cap\mathbb{R}_{>0}^{2m}$ satisfies the system of equations
\begin{align}
    \begin{cases}
        Q(\bm{w}) = 0 & \text{(singularity condition)} \\
        \bm{r} = -\eta \nabla_{\log} Q(\bm{w}), \quad \eta > 0 & \text{(criticality condition)}
    \end{cases}
    \label{app:eq:critical_eq}
\end{align}
\end{lemma}
\begin{proof}
Let $\bm{w}\in\mathbb{R}_{>0}^{2m}$ be the minimal positive real singularity from Lemma~\ref{app:lem:vp_minimal}. By Lemma~\ref{app:lem:variety_geometry}, $\mathcal{V}$ is smooth at $\bm{w}$ and is locally defined by $Q(\bm{z})=0$. The smooth critical point equations~\cite[Lemma~7.8]{pemantle2024analytic} require the direction $\bm{r}$ to be parallel to the logarithmic gradient $\nabla_{\log}Q(\bm{w})$. Combining the two conditions yields the system~\eqref{app:eq:critical_eq}.
\end{proof}

With the critical point equations established, we now solve for the Lagrange multiplier $\eta$ by exploiting the homogeneity of the transition matrix.

\begin{lemma}
\label{app:lem:scaling_invariance}
For any $a \in \mathbb{C}\setminus\! \{0\}$, the transition matrix $\matr{S}(\bm{t},\bm{u})$ defined in~\eqref{eq:Sdef} satisfies the identity
\begin{align}
    \matr{S}(a\bm{t}, a^{-1}\bm{u}) = \matr{S}(\bm{t}, \bm{u}).
    \label{app:eq:scaling_tu}
\end{align}
\end{lemma}
\begin{proof}
Using~\eqref{eq:Sdef}, we have
\begin{align*}
    & \matr{S}(a\bm{t}, a^{-1}\bm{u}) \\
    & \quad = \diag(a^{-1}\bm{u})^{1/2}\,\matr{B}^\transp\,\diag(a\bm{t})\,\matr{B}\,\diag(a^{-1}\bm{u})^{1/2} \\
    & \quad = a^{-1/2}\diag(\bm{u})^{1/2}\,\matr{B}^\transp\bigl(a\,\diag(\bm{t})\bigr)\,\matr{B}\,a^{-1/2}\diag(\bm{u})^{1/2} \\
    & \quad = \diag(\bm{u})^{1/2}\,\matr{B}^\transp\,\diag(\bm{t})\,\matr{B}\,\diag(\bm{u})^{1/2} \\
    & \quad = \matr{S}(\bm{t},\bm{u}),
\end{align*}
which proves~\eqref{app:eq:scaling_tu}. Consequently, all eigenvalues $\lambda_i(\bm{t},\bm{u})$ are invariant, and so are any functions of them, such as $Q(\bm{z})=1-\lambda_1(\bm{z})$ and the analytic factors $P(\bm{z})$ built from $\{\lambda_i(\bm{z})\}$.
\end{proof}

\begin{lemma}
\label{app:lem:eta_value}
Let $\bm{w}$ be a minimal smooth critical point satisfying~\eqref{app:eq:critical_eq}. Then the Lagrange scaling parameter in~\eqref{app:eq:critical_eq} is
\begin{align}
    \eta = n.
    \label{eq:eta-equals-n}
\end{align}
\end{lemma}

\begin{proof}
On $\mathcal{V}$ we have $Q=1\!-\!\lambda_1$, hence $\nabla_{\log}Q(\bm{w})=-\nabla_{\log}\lambda_1(\bm{w})$. Therefore the criticality condition~\eqref{app:eq:critical_eq} is equivalent to
\begin{align}
    \bm{r}=\eta\,\nabla_{\log}\lambda_1(\bm{w}),
\end{align}
i.e.,
\begin{align}
    k_i & = \eta\,\left.\frac{\partial\log\lambda_1}{\partial\log t_i}\right|_{\bm{w}}, \qquad i\in[m],
    \label{app:eq:eta-component-k}\\
    \ell_j & = \eta\,\left.\frac{\partial\log\lambda_1}{\partial\log u_j}\right|_{\bm{w}}, \qquad j\in[m].
    \label{app:eq:eta-component-ell}
\end{align}
Since $\matr{S}(\bm{t},\bm{u})$ is degree-$1$ homogeneous in $\bm{t}$ for fixed $\bm{u}$ and likewise in $\bm{u}$ for fixed $\bm{t}$, the Perron eigenvalue $\lambda_1$ inherits the same separate homogeneity. Euler's homogeneous function theorem then gives
\begin{align}
    \sum_{i\in[m]} \left.\frac{\partial\log\lambda_1}{\partial\log t_i}\right|_{\bm{w}} & = 1,
    \label{app:eq:euler-t-log}
    \\
    \sum_{j\in[m]} \left.\frac{\partial\log\lambda_1}{\partial\log u_j}\right|_{\bm{w}} & = 1.
    \label{app:eq:euler-u-log}
\end{align}
Summing~\eqref{app:eq:eta-component-k} and~\eqref{app:eq:eta-component-ell} over $i$ and $j$, respectively, and using~\eqref{app:eq:euler-t-log} and~\eqref{app:eq:euler-u-log}, gives
\begin{align}
    n=\sum_{i\in[m]}k_i = \sum_{i\in[m]}\eta\,\left.\frac{\partial\log\lambda_1}{\partial\log t_i}\right|_{\bm{w}} = \eta,\\
    n=\sum_{j\in[m]}\ell_j = \sum_{j\in[m]}\eta\,\left.\frac{\partial\log\lambda_1}{\partial\log u_j}\right|_{\bm{w}} = \eta,
\end{align}
yielding $\eta\!=\!n$.
\end{proof}

\subsection{Proof Sketch: Geometric Intuition and Derivation Strategy}
\label{app:proof_sketch}

We briefly sketch the geometric intuition behind the asymptotic derivation before proceeding to the technical details. The multivariate Cauchy integral for the coefficient localizes around the critical point $\bm{w}$ in logarithmic coordinates. The evaluation strategy involves decomposing the $2m$-dimensional integration space into three geometrically distinct subspaces:

\begin{itemize}
    \item \textbf{The Constraint Direction:} 
    First, we identify a specific direction corresponding to the inherent structural constraint of the matrix model (the row and column type counts are tied to the same total size, as imposed in Assumption~\ref{assump:general_block}). The integrand remains constant along this direction. We isolate this redundancy, which reduces the effective dimension of the problem by one and contributes a global geometric factor to the final result.

    \item \textbf{The Normal Component:} 
    Among the remaining dimensions, the direction normal to the singular manifold is the most critical. This direction captures the singularity of the defining function. By introducing a local normal coordinate and deforming the integration path into a (truncated) Hankel contour, we extract the primary asymptotic growth rate governed by the nature of the singularity. Crucially, this deformation offers a unified framework that handles both simple poles and algebraic branch cuts.
    
    \item \textbf{The Tangential Component:} 
    The remaining $2m\!-\!2$ directions are tangential to the singular manifold. Along these directions, the phase function is dominated by a non-degenerate quadratic form. Integration over this subspace yields a standard Gaussian integral, contributing a scalar prefactor determined by the curvature (i.e., the determinant of the Hessian matrix) of the manifold at the critical point.
\end{itemize}

Combining the geometric factor from the constraint removal, the growth factor from the normal singularity, and the curvature term from the tangential Gaussian integration, yields the final asymptotic formulas presented in Proposition~\ref{prop:block_asymp}.

\bigskip
\noindent
\textit{This concludes the proof sketch. We now proceed with the detailed derivation.}

\subsection{Invariance and Dimension Reduction}
\label{app:subsec:invariance_reduction}

Having established the standard form $F(\bm{z})=P(\bm{z})\,Q(\bm{z})^{-\alpha}$ and the smoothness of the singular variety at the dominant critical point $\bm{w}$, we now evaluate the coefficient extraction integral.

We begin by stating the standard integral representation for the coefficients.

\begin{lemma}
\label{app:lem:cauchy_std}
The coefficient $a_{\bm{r}}$ associated with the partitioning $\bm{r}$ is given by the multivariate Cauchy integral in $\mathbb{C}^{2m}$ as
\begin{align}
    a_{\bm{r}} = \frac{1}{(2\pi i)^{2m}} \oint_{\mathcal{C}} F(\bm{z})\,\bm{z}^{-\bm{r}-\bm{1}}\,\d\bm{z},
    \label{app:eq:cauchy_std}
\end{align}
where the integration contour $\mathcal{C}$ is a polytorus (i.e., a $2m$-dimensional torus), deformed to pass through the neighborhood of the dominant critical point $\bm{w}$.
\end{lemma}

Under Assumption~\ref{assump:general_block}, the partitioning $\bm{r}\!=\!(\bm{k};\bm{\ell})$ satisfies the balance constraint $\bm{c}\cdot \bm{r}=0$ with $\bm{c}\defeq(\bm{1}_m;-\bm{1}_m)$, where $\bm{1}_m$ denotes the all-one vector (see Notation in Section~\ref{sec:intro}), since $\sum_{i\in[m]}k_i \!=\! \sum_{j\in[m]}\ell_j \!=\! n$. This implies that the integrand exhibits a shift invariance along the direction $\bm{c}$. To exploit this invariance and the underlying geometry, we transition to logarithmic coordinates $\bm{\zeta} \in \mathbb{C}^{2m}$ by setting $\bm{z} \!=\! \e^{\bm{\zeta}}$. In this domain, the integration contour unfolds into a flat torus defined by $\zeta_j = \log R_j \!+\! i\theta_j$, where $R_j\in\mathbb{R}$ and $\theta_j$ ranges over $[0, 2\pi)$.

We first define the relevant functions in this new coordinate system.

\begin{definition}
\label{app:def:log_transformed_funcs}
Given the analytic factors $P(\bm{z})$ and $Q(\bm{z})$ from Lemma~\ref{app:lem:acsv_factorization}, define their log-transformed counterparts as
\begin{align}
    \tilde{P}(\bm{\zeta}) \defeq P(\e^{\bm{\zeta}}), \quad \tilde{Q}(\bm{\zeta}) \defeq Q(\e^{\bm{\zeta}}).
\end{align}
\end{definition}

\begin{lemma}
\label{app:lem:log_integral_form}
Under the substitution $\bm{z} = \e^{\bm{\zeta}}$, the term $\bm{z}^{-\bm{r}-\bm{1}}\d\bm{z}$ transforms as $\e^{-\bm{r}\cdot\bm{\zeta}}\d\bm{\zeta}$, with $\d\bm{z} \!=\! \bm{z}\d\bm{\zeta}$ canceling the $-\bm{1}$ in the exponent. Consequently, the coefficient $a_{\bm{r}}$ in \eqref{app:eq:cauchy_std} can be expressed as
\begin{align}
    a_{\bm{r}} = \frac{1}{(2\pi i)^{2m}} \oint_{\mathcal{C}} \e^{-\bm{r}\cdot\bm{\zeta}}\,\tilde{P}(\bm{\zeta})\,\tilde{Q}(\bm{\zeta})^{-\alpha} \d\bm{\zeta},
    \label{app:eq:cauchy_log}
\end{align}
where the integral is taken over the logarithmic parametrization of the polytorus $\mathcal{C}$ in Lemma~\ref{app:lem:cauchy_std}, i.e., $\zeta_j = \log R_j \!+\! i\theta_j$ with $(\theta_1; \dots; \theta_{2m}) \in [0, 2\pi)^{2m}$.
\end{lemma}

The linear structure of the logarithmic domain allows us to explicitly characterize the invariance imposed by the balance constraint. We distinguish between two types of invariance, i.e., radial scaling (real shifts) and angular translation (imaginary shifts).

\begin{lemma}
\label{app:lem:log_real_invariance}
Fix a point $\bm{\zeta}(0)$ on the logarithmic parametrization of the polytorus $\mathcal{C}$, and define $\bm{\zeta}(\varphi) \defeq \bm{\zeta}(0) + \varphi\bm{c}$ with $\varphi\in\mathbb{R}$. Then the logarithmic integrand is invariant under varying $\varphi$. Geometrically, this corresponds to the freedom of deforming the contour's radial coordinates $\log\!\bm{R}$ without changing the integral's value.
\end{lemma}
\begin{proof}
Write $\bm{\zeta}(\varphi)=\log\!\bm{R}(\varphi)+i\bm{\theta}$. Since $\varphi$ and $\bm{c}$ are real, varying $\varphi$ changes only $\log\!\bm{R}(\varphi)$, while $\bm{\theta}$ is fixed. The integrand consists of the exponential factor $\e^{-\bm{r}\cdot\bm{\zeta}}$ and the factors $\tilde{P}, \tilde{Q}$. The exponential factor is invariant because the balance constraint $\bm{c} \cdot \bm{r} = 0$ implies $\e^{-\bm{r}\cdot\bm{\zeta}(\varphi)} \!=\! \e^{-\bm{r}\cdot(\bm{\zeta}(0)+\varphi\bm{c})} \!=\! \e^{-\bm{r}\cdot\bm{\zeta}(0)}$. To establish the invariance of the factors $\tilde{P}, \tilde{Q}$, we map the shift back to the original variables $\bm{z} \!=\! (\bm{t}; \bm{u})$. Given the constraint $\bm{c} \!=\! (\bm{1}_m; -\bm{1}_m)$, the shift acts on the respective logarithmic coordinates as $\bm{\zeta}_{\bm{t}}(\varphi) = \bm{\zeta}_{\bm{t}}(0) \!+\! \varphi\bm{1}_m$ and $\bm{\zeta}_{\bm{u}}(\varphi) = \bm{\zeta}_{\bm{u}}(0) \!-\! \varphi\bm{1}_m$. Taking exponentials, we obtain $\bm{t}(\varphi) \!=\! \e^\varphi \bm{t}(0)$ and $\bm{u}(\varphi) \!=\! \e^{-\varphi} \bm{u}(0)$. By Lemma~\ref{app:lem:scaling_invariance}, the factors $P$ and $Q$ (and consequently $\tilde{P}$ and $\tilde{Q}$) are invariant under such scaling. Thus, the integrand remains constant under this radial deformation.
\end{proof}

\begin{lemma}
\label{app:lem:log_imag_continuous}
For a fixed polytorus $\mathcal{C}$ (i.e., fixed radii $\log\!\bm{R}$), fix a point $\bm{\zeta}(0)$ on its logarithmic parametrization, and define $\bm{\zeta}(\varphi) \defeq \bm{\zeta}(0) + i\varphi\bm{c}$ with $\varphi\in\mathbb{R}$. Then the logarithmic integrand is invariant under varying $\varphi$.
\end{lemma}
\begin{proof}
Write $\bm{\zeta}(\varphi)=\log\!\bm{R}+i\bm{\theta}(\varphi)$, so $\bm{\theta}(\varphi)=\bm{\theta}(0)+\varphi\bm{c}$ while $\log\!\bm{R}$ is fixed. The rest follows exactly as in the proof of Lemma~\ref{app:lem:log_real_invariance}. The exponential factor is invariant by $\bm{c}\cdot\bm{r}=0$, and exponentiating $\bm{\zeta}(\varphi)$ yields the scaling $\bm{t}(\varphi)=\e^{i\varphi}\bm{t}(0)$ and $\bm{u}(\varphi)=\e^{-i\varphi}\bm{u}(0)$, under which $P$ and $Q$ (hence $\tilde P$ and $\tilde Q$) are invariant by Lemma~\ref{app:lem:scaling_invariance}. Therefore, the logarithmic integrand remains constant as $\varphi$ varies.
\end{proof}

Since the integrand is constant under translations of $\bm{\theta}$ in the direction of $\bm{c}$, the integral along this specific direction evaluates simply to the length of its integration interval. This length is precisely determined by the fundamental period of the angular torus.

\begin{lemma}
\label{app:lem:log_imag_periodicity}
Consider $\bm{\zeta}(\varphi) \defeq \bm{\zeta}(0) + i\varphi\bm{c}$ defined in Lemma~\ref{app:lem:log_imag_continuous}, then $\bm{\zeta}(\varphi)$ is periodic on the polytorus, and the fundamental period is $\varphi=2\pi$.
\end{lemma}
\begin{proof}
On the polytorus we identify $\bm{\theta}$ with $\bm{\theta}+2\pi\bm{h}$ for any $\bm{h}\in\mathbb{Z}^{2m}$. Write $\bm{\zeta}(\varphi)=\log\!\bm{R}+i\bm{\theta}(\varphi)$, so $\bm{\theta}(\varphi)=\bm{\theta}(0)+\varphi\bm{c}$ and the angular displacement is $\Delta\bm{\theta}=\varphi\bm{c}$. Hence, $\bm{\zeta}(\varphi)$ closes if and only if $\Delta\bm{\theta}\in 2\pi\mathbb{Z}^{2m}$, i.e.,
\begin{align}
    \varphi\,\bm{c} = 2\pi\,\bm{h}, \qquad \bm{h}\in\mathbb{Z}^{2m}.
    \label{app:eq:period_condition}
\end{align}
For $\bm{c}\!=\!(\bm{1}_m;-\bm{1}_m)$ we have $c_j\in\{\pm1\}$ for all $j$, so the condition implies $\varphi=2\pi h_j/c_j$ for each coordinate, hence $\varphi/2\pi\in\mathbb{Z}$. Therefore, the minimal positive period is $\varphi=2\pi$.
\end{proof}

To exploit this periodic invariance for dimension reduction, we proceed in two steps. We first rotate to an orthogonal coordinate system aligned with the invariance direction $\bm{c}$, and then choose an adapted fundamental domain so that the integration interval along $\xi_{\hat{\bm{c}}}$ (defined below) has a constant length independent of the transverse coordinate.

\begin{definition}
\label{app:def:orthogonal_decomposition}
Given the invariance vector $\bm{c} \!\defeq\! (\bm{1}_m; -\bm{1}_m)$ and its normalization $\hat{\bm{c}} \defeq \bm{c}/\|\bm{c}\|_2$, let $\matr{U}_{\bm{c}}\in\mathbb{R}^{2m\times2m}$ be an orthogonal matrix having $\hat{\bm{c}}$ as its first column and satisfying $\det(\matr{U}_{\bm{c}})=1$. We define the rotated coordinates $\bm{\xi}\in\mathbb{C}^{2m}$ via
\begin{align}
    \bm{\zeta} = \bm{\zeta}^* + \matr{U}_{\bm{c}}\cdot\bm{\xi},
    \qquad \bm{\zeta}^*\defeq \log\bm{w}.
    \label{app:eq:zeta_xi_transform}
\end{align}
\end{definition}

\begin{remark}
Such a choice is always possible: if an orthogonal matrix with first column $\hat{\bm{c}}$ has determinant $-1$, flipping the sign of any transverse column yields another valid choice with determinant $1$.
\end{remark}

\begin{definition}
\label{app:def:orthogonal_slice}
Adopt the rotated coordinates $\bm{\xi}\in\mathbb{C}^{2m}$ in~\eqref{app:eq:zeta_xi_transform} and write $\bm{\xi}=(\xi_{\hat{\bm{c}}};\bm{\xi}_\perp)$, where $\xi_{\hat{\bm{c}}}$ is the coordinate along $\hat{\bm{c}}$. We define the orthogonal slice
\begin{align}
    H_{\bm{c}}\;\defeq\;\bigl\{\bm{\zeta}\in\mathbb{C}^{2m}\bigm| \xi_{\hat{\bm{c}}}=0\bigr\}.
\end{align}
\end{definition}

\begin{lemma}
\label{app:lem:work_on_slice}
In the saddle-point analysis near $\bm{\zeta}^*$, we may, without loss of generality, carry out the local deformation and the subsequent local analysis in a coordinate system restricted to the slice $H_{\bm{c}}$.
\end{lemma}
\begin{proof}
By Lemma~\ref{app:lem:log_real_invariance} and Lemma~\ref{app:lem:log_imag_continuous}, moving the contour locally in the $\bm{c}$ direction in either the real part or the imaginary part does not affect the integral. Hence, after separating (and integrating out) the $\bm{c}$-direction, the remaining contribution depends only on the transverse coordinates $\bm{\xi}_\perp$, and we are free to choose the transverse slice $\xi_{\hat{\bm{c}}}=0$, i.e., $H_{\bm{c}}$, for the subsequent local analysis.
\end{proof}

\begin{lemma}
\label{app:lem:volume_factorization}
Under the orthogonal change of variables in~\eqref{app:eq:zeta_xi_transform}, the holomorphic volume form factorizes as
\begin{align}
    \d \bm{\zeta} = \d \zeta_1 \wedge \dots \wedge \d \zeta_{2m} = \d\xi_{\hat{\bm{c}}}\wedge \d\bm\xi_\perp,
    \label{app:eq:volume_form_factorization}
\end{align}
where $\d\bm\xi_\perp$ denotes the induced $(2m\!-\!1)$-form on the orthogonal slice $H_{\bm{c}}$.
\end{lemma}
\begin{proof}
The transformation in~\eqref{app:eq:zeta_xi_transform} is linear with Jacobian matrix $\matr{U}_{\bm{c}}$. The holomorphic volume form transforms by the determinant, and the convention $\det(\matr{U}_{\bm{c}})\!=\!1$ in Definition~\ref{app:def:orthogonal_decomposition} implies $\d\zeta_1\wedge\cdots\wedge\d\zeta_{2m}=\d\xi_1\wedge\cdots\wedge\d\xi_{2m}$. Writing $\bm{\xi}=(\xi_{\hat{\bm{c}}};\bm{\xi}_{\perp})$ yields~\eqref{app:eq:volume_form_factorization}, allowing the volume element to be decomposed into the invariance direction $\d\xi_{\hat{\bm{c}}}$ and the transverse volume form $\d\bm{\xi}_{\perp}$.
\end{proof}

\begin{remark}
Since $\matr{U}_{\bm{c}}\in\mathbb{R}^{2m\times 2m}$ is real, the change of variables~\eqref{app:eq:zeta_xi_transform} acts separately on the real and imaginary parts, i.e., $\mathrm{Re}(\bm{\zeta}) = \mathrm{Re}(\bm{\zeta}^*) \!+\! \matr{U}_{\bm{c}}\cdot\mathrm{Re}(\bm{\xi})$ and $\mathrm{Im}(\bm{\zeta}) = \mathrm{Im}(\bm{\zeta}^*) \!+\! \matr{U}_{\bm{c}}\cdot\mathrm{Im}(\bm{\xi})$. Hence, a fixed polytorus (with constant radii $\mathrm{Re}(\bm{\zeta}) = \log\!\bm{R}$) corresponds to a strictly constant real slice $\mathrm{Re}(\bm{\xi}) = \matr{U}_{\bm{c}}^\transp\cdot(\log\!\bm{R} \!-\! \mathrm{Re}(\bm{\zeta}^*))$. In what follows, we keep this real slice fixed and only adjust the angular (imaginary) fundamental domain.
\end{remark}

However, the standard fundamental domain, i.e., the hypercube $\bm{\theta} \in [0,2\pi)^{2m}$ within the imaginary subspace, is misaligned with the new invariance axis $\xi_{\hat{\bm{c}}}$. Consequently, when the domain is expressed in the rotated coordinates $\bm{\xi}$, the line segment in the $\xi_{\hat{\bm{c}}}$-direction contained in the domain has a length that depends on the transverse coordinates $\bm{\xi}_{\perp}$. To separate the integration and obtain a constant interval length along $\xi_{\hat{\bm{c}}}$, we leverage the lattice periodicity to select an adapted fundamental domain.

\begin{definition}
\label{app:def:adapted_domain}
Let $\mathcal{D}_{\bm{c}}$ denote a fundamental domain for the angular lattice $2\pi \mathbb{Z}^{2m}$ in $\bm{\theta}$. We choose $\mathcal{D}_{\bm{c}}$ as a fundamental parallelotope such that one of its spanning lattice vectors is the primitive period $2\pi \bm{c}$ established in Lemma~\ref{app:lem:log_imag_periodicity}. We denote by $\mathcal{C}_{\bm{c}}$ the $(2m\!-\!1)$-dimensional cycle in $H_{\bm{c}}$ obtained from the transverse coordinates $\bm{\xi}_\perp$ induced by $\mathcal{D}_{\bm{c}}$.
\end{definition}

\begin{remark}
Since $\bm{c}$ is a primitive integer vector, it can be extended to a basis of $\mathbb{Z}^{2m}$, or equivalently completed to a unimodular matrix (see~\cite[Chapter~1, Corollary~4]{cassels1997geometry}). We denote by $\mathcal{D}_{\bm{c}}$ the corresponding fundamental parallelotope.
\end{remark}

At this point, the reader might want to skim through the upcoming Example~\ref{app:ex:fundamental_domain_m1}, which visualizes the geometric setup for the special case $m\!=\!1$.

\begin{remark}
The set $\mathcal{D}_{\bm{c}}$ is a possible, convenient choice of a fundamental parallelotope for the angular lattice; its transverse edges need not align with the orthogonal axes $\bm{\xi}_\perp$. Orthogonality is only used to factorize the differential form, whereas the parallelotope is chosen to yield a constant interval length along $\xi_{\hat{\bm{c}}}$.
\end{remark}

\begin{lemma}
\label{app:lem:constant_interval_length}
In the orthogonal coordinates $\bm{\xi}=(\xi_{\hat{\bm{c}}};\bm{\xi}_\perp)$, the integration interval along $\xi_{\hat{\bm{c}}}$ over the adapted domain $\mathcal{D}_{\bm{c}}$ has a constant length of $2\pi\|\bm{c}\|_2$ for any fixed transverse coordinate.
\end{lemma}
\begin{proof}
By construction, the parallelotope $\mathcal{D}_{\bm{c}}$ has $2\pi\bm{c}$ as one of its spanning vectors. Since $\mathcal{D}_{\bm{c}}$ is a parallelotope with $2\pi\bm{c}$ as one spanning edge, any affine line parallel to $\bm{c}$ intersects $\mathcal{D}_{\bm{c}}$ in a translate of that edge; hence the resulting one-dimensional intersection is always parallel to $\bm{c}$ and has the same length as $2\pi\bm{c}$. Consequently, the length of the interval along the normalized direction $\hat{\bm{c}}$ is exactly the orthogonal projection of this spanning vector onto $\hat{\bm{c}}$, yielding $2\pi\bm{c} \cdot \hat{\bm{c}} \!=\! 2\pi\|\bm{c}\|_2$.
\end{proof}

\begin{definition}
\label{app:def:deff}
Define the effective dimension by
\begin{align}
    d_{\mathrm{eff}} \defeq 2m-1.
\end{align}
\end{definition}

\begin{lemma}
\label{app:lem:dimension_reduction}
Adopting the orthogonal coordinates from Definition~\ref{app:def:orthogonal_decomposition}, the $2m$-dimensional coefficient integral reduces to an integral over the $d_{\mathrm{eff}}$-dimensional cycle $\mathcal{C}_{\bm{c}}$ on the orthogonal slice $H_{\bm{c}}$ as
\begin{align}
    a_{\bm{r}} = \frac{\|\bm{c}\|_2}{(2\pi i)^{d_{\mathrm{eff}}}} \int_{\mathcal{C}_{\bm{c}}} \e^{-\bm{r}\cdot\bm{\zeta}}\,\tilde{P}(\bm{\zeta})\,\tilde{Q}(\bm{\zeta})^{-\alpha}\, \d\bm{\xi}_{\perp},
    \label{app:eq:reduced_integral}
\end{align}
where $\bm{\zeta}$ is understood as the restriction $\bm{\zeta}=\bm{\zeta}(\bm{\xi}_\perp)$ on the slice $\xi_{\hat{\bm{c}}}=0$, i.e., via the natural embedding of the slice $H_{\bm{c}}$ into the ambient log-space.
\end{lemma}
\begin{proof}
We evaluate the integral over the adapted fundamental domain $\mathcal{D}_{\bm{c}}$. Using the factorization of the volume form from Lemma~\ref{app:lem:volume_factorization} and Fubini's theorem, we decompose the integral into the invariance direction $\xi_{\hat{\bm{c}}}$ and the transverse directions $\bm{\xi}_\perp$.

By Lemma~\ref{app:lem:log_imag_continuous}, the entire integrand is constant along the coordinate $\xi_{\hat{\bm{c}}}$. To evaluate the inner integral rigorously, we parameterize the integration path using $\varphi \in [0, 2\pi)$ defined in Lemma~\ref{app:lem:log_imag_periodicity}. Recall from Definition~\ref{app:def:orthogonal_decomposition} that the orthogonal coordinate is given by $\xi_{\hat{\bm{c}}} = \hat{\bm{c}}^\transp (\bm{\zeta} \!-\! \bm{\zeta}^*)$. Since the critical point $\bm{\zeta}^*$ is a constant vector, taking the differential with respect to the integration variables yields $\d\xi_{\hat{\bm{c}}} = \hat{\bm{c}}^\transp \d\bm{\zeta}$. Along this invariance translation, the shift is strictly parameterized as $\d\bm{\zeta} = i\bm{c}\,\d\varphi$, which gives
\begin{align}
    \d\xi_{\hat{\bm{c}}} = \hat{\bm{c}}^\transp (i\bm{c}\,\d\varphi) = i (\hat{\bm{c}}^\transp\bm{c})\,\d\varphi = i\|\bm{c}\|_2\,\d\varphi.
\end{align}
Consequently, the inner integral evaluates to
\begin{align}
    \int_{0}^{2\pi} \d\xi_{\hat{\bm{c}}} = \int_{0}^{2\pi} i\|\bm{c}\|_2\,\d\varphi = 2\pi i\|\bm{c}\|_2.
\end{align}
Multiplying the restricted integrand by this constant factor yields
\begin{align}
    a_{\bm{r}} &= \frac{1}{(2\pi i)^{2m}} \cdot (2\pi i \|\bm{c}\|_2) \int_{\mathcal{C}_{\bm{c}}} \e^{-\bm{r}\cdot\bm{\zeta}}\,\tilde{P}(\bm{\zeta})\,\tilde{Q}(\bm{\zeta})^{-\alpha}\, \d\bm{\xi}_{\perp} \nonumber \\
    &= \frac{\|\bm{c}\|_2}{(2\pi i)^{2m-1}} \int_{\mathcal{C}_{\bm{c}}} \e^{-\bm{r}\cdot\bm{\zeta}}\,\tilde{P}(\bm{\zeta})\,\tilde{Q}(\bm{\zeta})^{-\alpha}\, \d\bm{\xi}_{\perp}.
\end{align}
This concludes the proof.
\end{proof}

We illustrate the dimension reduction mechanism using a concrete $2$-dimensional example ($m\!=\!1$) as follows.

\begin{example}
\label{app:ex:fundamental_domain_m1}
Consider the coefficient integral over the polytorus in logarithmic coordinates $\bm{\zeta}\!=\!(\zeta_t;\zeta_u)\!\in\!\mathbb{C}^2$, which relate to the original variables $\bm{z}\!=\!(t;u)$ via $t\!=\!\e^{\zeta_t}$ and $u\!=\!\e^{\zeta_u}$. Fix a point $\bm{\zeta}(0)$ on the logarithmic parametrization of the polytorus and define
\begin{align}
    \bm{\zeta}(\varphi)\;\defeq\;\bm{\zeta}(0)+i\varphi\,\bm{c},
    \qquad
    \bm{c}\defeq(1;-1),
    \qquad
    \varphi\in\mathbb{R}.
\end{align}
Assume the integrand $f(\bm{\zeta})$ is invariant under varying $\varphi$, i.e., $f(\bm{\zeta}(\varphi))$ is constant in $\varphi$.

On the flat torus of angular variables $(\theta_t;\theta_u)=(\mathrm{Im}(\zeta_t);\mathrm{Im}(\zeta_u))$ (mod $2\pi$), varying $\varphi$ moves $(\theta_t;\theta_u)$ along the direction $\bm{c}$, as illustrated in Fig.~\ref{app:fig:fundamental_domain}.

Define the unit direction
\begin{align}
    \hat{\bm{c}}\;\defeq\;\frac{\bm{c}}{\|\bm{c}\|_2}
    \;=\;\frac{1}{\sqrt{2}}(1;-1).
\end{align}
Without loss of generality, we set the critical point to be $\bm{\zeta}^* \!=\! \bm{0}$, and introduce the orthogonal change of variables
\begin{align}
    \matr{U}_{\bm{c}}\;\defeq\;
    \begin{pmatrix}
        \frac{1}{\sqrt{2}} & \frac{1}{\sqrt{2}}\\[2pt]
        -\frac{1}{\sqrt{2}} & \frac{1}{\sqrt{2}}
    \end{pmatrix},
    \qquad
    \begin{pmatrix}
        \zeta_t\\[2pt]
        \zeta_u
    \end{pmatrix}
    \;=\;
    \matr{U}_{\bm{c}}\cdot
    \begin{pmatrix}
        \xi_{\hat{\bm{c}}}\\[2pt]
        \xi_{\perp}
    \end{pmatrix}.
    \label{eq:ex_xi_norm_def_revised}
\end{align}
Then $\d\zeta_t\wedge\d\zeta_u=\d\xi_{\hat{\bm{c}}}\wedge\d\xi_{\perp}$.

The key point is how the choice of fundamental domain affects the interval in the invariance direction $\bm{c}$ (equivalently, along the $\xi_{\hat{\bm{c}}}$-axis). In the standard domain $[0,1)^2$ (blue in Fig.~\ref{app:fig:fundamental_domain}), the direction $\bm{c}\!=\!(1;-1)$ is not aligned with the domain edges. As a result, for different transverse positions, the interval parallel to $\bm{c}$ contained in the domain has a length that depends on the transverse position. Consequently, after switching to the rotated coordinate system $(\xi_{\hat{\bm{c}}};\xi_{\perp})$, the integration interval along $\xi_{\hat{\bm{c}}}$ has a length that depends on the transverse coordinate.

By choosing instead an invariance-aligned fundamental domain $\mathcal{D}_{\bm{c}}$ (red in Fig.~\ref{app:fig:fundamental_domain}) with one of the edges parallel to $\bm{c}$, the interval parallel to $\bm{c}$ has a constant length inside $\mathcal{D}_{\bm{c}}$, independent of the transverse position. In the rotated coordinate $\xi_{\hat{\bm{c}}}$, this common length is
\begin{align}
    \Delta\xi_{\hat{\bm{c}}}=2\pi i\,\|\bm{c}\|_2=2\pi i\sqrt{2},
\end{align}
which is exactly the factor integrated out in Lemma~\ref{app:lem:dimension_reduction}. Consequently,
\begin{align}
    & \frac{1}{(2\pi i)^2}\oint_{\mathcal{C}} f(\bm{\zeta})\,\d\zeta_t\wedge\d\zeta_u \nonumber\\
    & \;=\;
    \frac{1}{(2\pi i)^2}\Bigl(2\pi i\|\bm{c}\|_2\Bigr)
    \int_{\mathcal{C}_{\bm{c}}} f(\bm{\zeta}(\xi_{\perp}))\,\d\xi_{\perp} \nonumber\\
    & \;=\;
    \frac{\sqrt{2}}{2\pi i}\int_{\mathcal{C}_{\bm{c}}} f(\bm{\zeta}(\xi_{\perp}))\,\d\xi_{\perp}.
\end{align}
Note that $\bm{\zeta}(\varphi)$ is the parametrization along the invariance direction, while $\bm{\zeta}(\xi_{\perp})$ denotes the restriction $\bm{\zeta}\!=\!\bm{\zeta}(\bm{\xi}_{\perp})$ on the slice $\xi_{\hat{\bm{c}}}\!=\!0$ in Lemma~\ref{app:lem:dimension_reduction}.
\end{example}

\begin{figure}[t]
    \centering
    \includegraphics[width=0.8\linewidth]{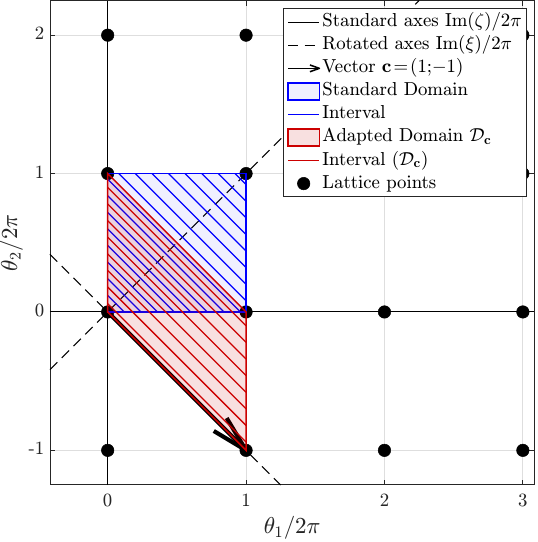} 
    \caption{An invariance-aligned fundamental domain on the flat torus ($m\!=\!1$). The plot is in the normalized angular coordinates $(\theta_t/2\pi,\theta_u/2\pi)$. The blue square is the standard fundamental domain, and the red parallelogram is an alternative choice $\mathcal{D}_{\bm{c}}$ aligned with $\bm{c}\!=\!(1;-1)$. Solid and dashed black lines indicate the standard and rotated axes, respectively.}
    \label{app:fig:fundamental_domain}
\end{figure}

\subsection{Local Geometry and Tangential Decomposition}
\label{app:subsec:local_geometry}

In order to evaluate the reduced integral in~\eqref{app:eq:reduced_integral}, we introduce local coordinates on the slice $H_{\bm{c}}$ near the critical point $\bm{\zeta}^*\!=\!\log\bm{w}$.

\begin{remark}
\label{rem:local_nature}
In this subsection we work only in a neighborhood of the dominant critical point $\bm{\zeta}^*$. This is sufficient for the leading asymptotics: the saddle-point contribution localizes to this neighborhood, so the global shape of the reduced cycle $\mathcal{C}_{\bm{c}}$ away from $\bm{\zeta}^*$ does not affect the main term. Our goal is therefore to introduce local coordinates on the slice $H_{\bm{c}}$ and to express the reduced integral in these coordinates.
\end{remark}

\begin{definition}
\label{app:def:normal_direction}
Define the logarithmic normal direction at the critical point by
\begin{align}
    \bm{v} \;\defeq\; \nabla_{\log}Q(\bm{w}).
    \label{app:eq:vdef}
\end{align}
\end{definition}

\begin{lemma}
\label{app:lem:local_basis}
The logarithmic normal direction $\bm{v}$ lies in the orthogonal slice $H_{\bm{c}}$ in Definition~\ref{app:def:orthogonal_slice}. There exists a matrix $\matr{V}$ of size $2m\times(2m\!-\!2)$ with orthonormal columns satisfying $\matr{V}^\transp \bm{v}=\bm{0}$ and being such that every $\bm\zeta\in H_{\bm{c}}$ in a neighborhood of $\bm\zeta^*$ can be written uniquely as
\begin{align}
    \bm{\zeta}=\bm{\zeta}^*+\omega\,\bm{v}+\matr{V}\,\bm{b}_{\mathrm{p}},
    \label{app:eq:zeta_decomposition}
\end{align}
for some $\omega\in\mathbb{C}$ and $\bm{b}_{\mathrm{p}}\in\mathbb{C}^{d_{\mathrm{eff}}-1}$.\footnote{We emphasize that the coordinate vector $\bm{b}_{\mathrm{p}}$ denotes local tangential parameters and is distinct from the entries $b_{i,j}$ of the base matrix $\matr{B}$ defined in Assumption~\ref{assump:general_block}.}
\end{lemma}

\begin{proof}
From the criticality condition~\eqref{app:eq:critical_eq}, $\bm{r}$ is parallel to $\nabla_{\log}Q(\bm{w})$, or, equivalently, $\bm{r}$ is parallel to $\bm{v}$. Since admissible partitionings satisfy $\bm{c}\cdot\bm{r}=0$, we have $\bm{c}\cdot\bm{v}=0$, hence $\bm{v}\in H_{\bm{c}}$. The orthogonal complement of $\mathrm{span}\{\bm{v}\}$ inside $H_{\bm{c}}$ has dimension $d_{\mathrm{eff}}-1$, so we may choose an orthonormal basis $\{\bm{v}^{(2)},\dots,\bm{v}^{(d_{\mathrm{eff}})}\}$ for it and set $\matr{V} \defeq [\bm{v}^{(2)},\dots,\bm{v}^{(d_{\mathrm{eff}})}]$. This yields \eqref{app:eq:zeta_decomposition}. Uniqueness of the decomposition in~\eqref{app:eq:zeta_decomposition} follows because $\{\bm{v},\bm{v}^{(2)},\dots,\bm{v}^{(d_{\mathrm{eff}})}\}$ is a basis of $H_{\bm{c}}$.
\end{proof}

\begin{definition}
\label{app:def:full_hessian}
Let $\matr{H}_{\bm{\zeta}^*}$ be the full Hessian matrix of the transformed singular function $\tilde{Q}$ at the critical point $\bm{\zeta}^*$, i.e.,
\begin{align}
    \matr{H}_{\bm{\zeta}^*} \defeq \nabla^2 \tilde{Q}(\bm{\zeta}^*).
\end{align}
\end{definition}

\begin{lemma}
\label{app:lem:local_expansion}
In the scaling regime $\omega = O(\|\bm{b}_{\mathrm{p}}\|_2^2)$, which characterizes the zero set $\bigl\{\bm{\zeta}\in\mathbb{C}^{2m} \bigm| \tilde{Q}(\bm{\zeta})=0\bigr\}$ near $\bm{\zeta}^*$, $\tilde{Q}(\bm{\zeta})$ admits the simplified asymptotic expansion
\begin{align}
    \tilde{Q}(\bm{\zeta}) = \omega\|\bm{v}\|_2^2 + \frac{1}{2}\bm{b}_{\mathrm{p}}^\transp \matr{H}_{\mathrm{p}} \bm{b}_{\mathrm{p}} + O(\|\bm{b}_{\mathrm{p}}\|_2^3),
    \label{app:eq:Q_full_expansion}
\end{align}
where $\matr{H}_{\mathrm{p}} \defeq \matr{V}^\transp \matr{H}_{\bm{\zeta}^*} \matr{V}$ is the tangential Hessian matrix.
\end{lemma}
\begin{proof}
Consider the local coordinates $\bm{\zeta} = \bm{\zeta}^* + \omega\bm{v} + \matr{V}\bm{b}_{\mathrm{p}}$ defined in Lemma~\ref{app:lem:local_basis}. Expanding $\tilde{Q}$ around $\bm{\zeta}^*$ with displacement $\bm{\delta} \defeq \omega\bm{v} + \matr{V}\bm{b}_{\mathrm{p}}$, and noting $\nabla\tilde{Q}(\bm{\zeta}^*)\!=\!\nabla_{\log}Q(\bm{w})\!=\!\bm{v}$ and $\matr{V}^\transp \bm{v}=\bm{0}$, yields the general form
\begin{align}
    \tilde{Q}(\bm{\zeta}) = \omega\|\bm{v}\|_2^2 + \frac{1}{2}\bm{\delta}^{\!\transp} \matr{H}_{\bm{\zeta}^*} \bm{\delta} + O(\|\bm{\delta}\|_2^3).
\end{align}
Restricting this expansion to the zero set $\bigl\{\bm{\zeta}\in\mathbb{C}^{2m} \!\bigm| \tilde{Q}(\bm{\zeta})\!=\!0\bigr\}$ and separating the quadratic form provides the implicit condition
\begin{align}
    0 & = \omega\|\bm{v}\|_2^2 + \frac{1}{2}\omega^2 \bm{v}^\transp \matr{H}_{\bm{\zeta}^*} \bm{v} + \omega \bm{v}^\transp \matr{H}_{\bm{\zeta}^*} \matr{V} \bm{b}_{\mathrm{p}} \nonumber\\ 
    & \qquad + \frac{1}{2}\bm{b}_{\mathrm{p}}^\transp \matr{H}_{\mathrm{p}} \bm{b}_{\mathrm{p}}+ O(\|\bm{\delta}\|_2^3).
\end{align}
Since we describe the zero set locally near $\bm{\zeta}^*$, we restrict to the small-$\omega$ branch (with $\omega\to 0$ as $\bm{b}_{\mathrm{p}}\to\bm{0}$). Given that components of the gradient vector and Hessian matrix are fixed at $\bm{\zeta}^*$, the linear term $\omega\|\bm{v}\|_2^2$ must balance the leading tangential quadratic contribution to satisfy this identity, yielding the scaling $\omega \!=\! O(\|\bm{b}_{\mathrm{p}}\|_2^2)$. Under this scaling, $\|\bm{\delta}\|_2 \!=\! O(\|\bm{b}_{\mathrm{p}}\|_2)$, resulting in the error term $O(\|\bm{b}_{\mathrm{p}}\|_2^3)$. Mixed and higher-order normal terms (e.g., $\omega^2 \!=\! O(\|\bm{b}_{\mathrm{p}}\|_2^4)$ and $\omega\|\bm{b}_{\mathrm{p}}\|_2 \!=\! O(\|\bm{b}_{\mathrm{p}}\|_2^3)$) are negligible and absorbed into the error, yielding \eqref{app:eq:Q_full_expansion}.
\end{proof}

\begin{lemma}
\label{app:lem:volume_transform}
In the local coordinates $(\omega; \bm{b}_{\mathrm{p}})$, the holomorphic volume form $\d\bm{\xi}_{\perp}$ defined in Lemma~\ref{app:lem:volume_factorization} transforms near the critical point as
\begin{align}
    \d\bm{\xi}_{\perp} \approx \frac{1}{\|\bm{v}\|_2} \, \d\tilde{Q} \wedge \d\bm{b}_{\mathrm{p}},
    \label{app:eq:omega_transform}
\end{align}
where $\d\bm{b}_{\mathrm{p}} \defeq \d b_2 \wedge \dots \wedge \d b_{d_{\mathrm{eff}}}$. Here $\approx$ means that, in the local regime $(\omega,\bm{b}_{\mathrm{p}})\!\to\!(0,\bm{0})$, i.e., $\bm{\zeta}\!\to\!\bm{\zeta}^*$, the difference between the two sides is of higher order and vanishes in the corresponding local limit; equivalently, the approximation becomes exact at leading order in a sufficiently small neighborhood of $\bm{\zeta}^*$.
\end{lemma}
\begin{proof}
Using the orthogonal basis from Lemma~\ref{app:lem:local_basis}, note that $\bm{v}$ has length $\|\bm{v}\|_2$ while the columns of $\matr{V}$ are orthonormal. Consequently, the Jacobian is $\|\bm{v}\|_2$, giving $\d\bm{\xi}_{\perp} = \|\bm{v}\|_2 \, \d \omega \wedge \d\bm{b}_{\mathrm{p}}$. Differentiating the expansion~\eqref{app:eq:Q_full_expansion} in Lemma~\ref{app:lem:local_expansion} with respect to the coordinates $(\omega, \bm{b}_{\mathrm{p}})$ shows that $\d\tilde{Q}$ is dominated by the normal term, i.e., $\d\tilde{Q} \approx \|\bm{v}\|_2^2 \, \d \omega$. Substituting $\d \omega \approx \|\bm{v}\|_2^{-2} \d\tilde{Q}$ yields
\begin{align}
    \d\bm{\xi}_{\perp} \approx \|\bm{v}\|_2 \left(\frac{1}{\|\bm{v}\|_2^2} \d\tilde{Q}\right) \wedge \d\bm{b}_{\mathrm{p}} = \frac{1}{\|\bm{v}\|_2} \d\tilde{Q} \wedge \d\bm{b}_{\mathrm{p}}.
\end{align}
\end{proof}

Combining the local expansion and the transformed volume form allows us to separate the singular and tangential components of the integral.

\begin{lemma}
\label{app:lem:integral_separation}
The dimensionally reduced coefficient integral \eqref{app:eq:reduced_integral} admits the asymptotically equivalent separated form
\begin{align}
    a_{\bm{r}} \sim \frac{\|\bm{c}\|_2}{(2\pi i)^{d_{\mathrm{eff}}} \|\bm{v}\|_2} \int_{\mathcal{C}_{\bm{c}}} \exp(-\bm{r} \!\cdot\! \bm{\zeta}) \tilde{P}(\bm{\zeta}) \tilde{Q}(\bm{\zeta})^{-\alpha} \, \d\tilde{Q} \wedge \d\bm{b}_{\mathrm{p}}.
    \label{app:eq:ar_separated}
\end{align}
\end{lemma}
\begin{proof}
Substituting the volume form approximation from Lemma~\ref{app:lem:volume_transform} into the reduced integral representation in Lemma~\ref{app:lem:dimension_reduction}, directly yields the result.
\end{proof}

\subsection{Phase Expansion and Steepest Descent}

Having established the local geometry and the reduced integral form, we now evaluate the integral \eqref{app:eq:ar_separated} via the method of steepest descent. We first identify the gradient in the transformed coordinates and define the relevant phase and normal functions.

\begin{remark}
\label{rem:transformed_gradient_ident}
Combining Definitions~\ref{app:def:normal_direction} and~\ref{app:def:log_transformed_funcs}, the log-transformed gradient at $\bm{\zeta}^*$ satisfies
\begin{align}
    \bm{v}\;\defeq\; \nabla_{\log}Q(\bm{w}) = \nabla \tilde Q(\bm{\zeta}^*).
\end{align}
\end{remark}

\begin{definition}
\label{def:phase_and_normal}
Define the linear phase function and the normal coordinate relevant to the saddle point integration as
\begin{align}
    \Psi(\bm{\zeta}) \defeq \bm{r} \cdot \bm{\zeta} \quad \text{and} \quad q \defeq -\tilde{Q}(\bm{\zeta}).
\end{align}
\end{definition}

\begin{remark}
\label{rem:phase_normal_preview}
The $\Psi(\bm{\zeta})$ and $q$ defined in Definition~\ref{def:phase_and_normal} are used to separate the tangential and normal contributions in the saddle-point analysis. In particular, combining the local phase expansion (Lemma~\ref{app:lem:phase_expansion}) with the steepest-descent parametrization (Lemma~\ref{app:lem:phase_factorization}) gives
\begin{align}
    \exp\bigl(-\Psi(\bm{\zeta})\bigr)
    \;\approx\;
    \bm{w}^{-\bm{r}}\,
    \exp(-nq)\,
    \exp\!\left(-\frac{\|\bm{r}\|_2}{2}\,\bm{y}^\transp \mathcal{H}\,\bm{y}\right).
\end{align}
Consequently, later in Lemma~\ref{lem:integral_factorization} the reduced coefficient integral factorizes into a $(d_{\mathrm{eff}}-1)$-dimensional Gaussian integral and a one-dimensional integral in $q$.
\end{remark}

\begin{lemma}
\label{app:lem:phase_expansion}
In the scaling regime $\omega = O(\|\bm{b}_{\mathrm{p}}\|_2^2)$ from Lemma~\ref{app:lem:local_expansion}, the linear phase function $\Psi(\bm{\zeta})$ admits the asymptotic expansion
\begin{align}
    \Psi(\bm{\zeta}) = \bm{r} \cdot \bm{\zeta}^* + n\,q + \frac{n}{2}\,\bm{b}_{\mathrm{p}}^\transp \matr{H}_{\mathrm{p}}\bm{b}_{\mathrm{p}} + O(n\|\bm{b}_{\mathrm{p}}\|_2^3).
    \label{eq:phase_expansion_local}
\end{align}
\end{lemma}
\begin{proof}
We analyze the phase function directly in the local coordinates. Using the linearity of the inner product and the local expansion of the normal coordinate $q$ derived in Lemma~\ref{app:lem:local_expansion}, we obtain
\begin{align}
    \Psi(\bm{\zeta}) &= \bm{r}\!\cdot\!\bm{\zeta}^* + \bm{r}\cdot(\bm{\zeta}\!-\!\bm{\zeta}^*) \nonumber\\
    &\overset{\mathrm{(a)}}{=} \bm{r}\!\cdot\!\bm{\zeta}^* - n\bm{v}^\transp \bigl(\omega\bm{v} + \matr{V}\bm{b}_{\mathrm{p}}\bigr) \nonumber\\
    &\overset{\mathrm{(b)}}{=} \bm{r}\!\cdot\!\bm{\zeta}^* - n \omega \|\bm{v}\|_2^2 \nonumber\\
    &\overset{\mathrm{(c)}}{=} \bm{r}\!\cdot\!\bm{\zeta}^* - n \Bigl( -q - \frac{1}{2}\bm{b}_{\mathrm{p}}^\transp \matr{H}_{\mathrm{p}}\bm{b}_{\mathrm{p}} + O(\|\bm{b}_{\mathrm{p}}\|_2^3) \Bigr) \nonumber\\
    &= \bm{r}\!\cdot\!\bm{\zeta}^* + n q + \frac{n}{2}\bm{b}_{\mathrm{p}}^\transp \matr{H}_{\mathrm{p}}\bm{b}_{\mathrm{p}} + O(n\|\bm{b}_{\mathrm{p}}\|_2^3),
\end{align}
where $(\mathrm{a})$ substitutes the saddle-point alignment $\bm{r} = -n\bm{v}$ and the coordinate displacement $\bm{\zeta}-\bm{\zeta}^* = \omega\bm{v} + \matr{V}\bm{b}_{\mathrm{p}}$, step $(\mathrm{b})$ follows from the orthogonality $\matr{V}^\transp \bm{v}=0$, and step $(\mathrm{c})$ substitutes the relation $\omega\|\bm{v}\|_2^2 = -q - \frac{1}{2}\bm{b}_{\mathrm{p}}^\transp \matr{H}_{\mathrm{p}}\bm{b}_{\mathrm{p}} + O(\|\bm{b}_{\mathrm{p}}\|_2^3)$ in the regime $\omega\!=\!O(\|\bm{b}_{\mathrm{p}}\|_2^2)$ implied by Lemma~\ref{app:lem:local_expansion}, with $q \!=\! -\tilde Q(\bm{\zeta})$.
\end{proof}

\begin{definition}
\label{def:effective_hessian}
To facilitate the Gaussian integral formulation, we define the normalized effective Hessian $\mathcal{H}$ by rescaling the tangential Hessian $\matr{H}_{\mathrm{p}}$ from Lemma~\ref{app:lem:local_expansion} as
\begin{align}
    \mathcal{H} \;\defeq\; -\|\bm{v}\|_2^{-1} \matr{H}_{\mathrm{p}}.
\end{align}
This normalization ensures that $\mathcal{H}$ captures the strictly positive curvature of $\log\mathcal{V}$ at the critical point.
\end{definition}

\begin{lemma}
\label{app:lem:phase_factorization}
Under the steepest descent parametrization $\bm{b}_{\mathrm{p}} = i\bm{y}$ (with $\bm{y}\in\mathbb{R}^{d_{\mathrm{eff}}-1}$), the term $\exp(-\Psi(\bm{\zeta}))$ takes the local form
\begin{align}
    \exp\bigl(-\Psi(\bm\zeta)\bigr)
    \approx
    \bm{w}^{-\bm{r}}\,\exp(-nq)\,
    \exp\!\left(-\frac{\|\bm{r}\|_2}{2}\,\bm{y}^\transp \mathcal{H}\,\bm{y}\right),
    \label{eq:phase_factorization}
\end{align}
where $\mathcal{H}$ is positive definite ($\mathcal{H} \succ 0$).
\end{lemma}
\begin{proof}
We perform the phase analysis by deforming the tangential slice of the global cycle $\mathcal{C}_{\bm{c}}$ to a local steepest descent cycle $\tilde{\sigma}$ lying within the intersection $\log\mathcal{V}\cap H_{\bm{c}}$ at $\bm{\zeta}^*$. Let $\mathcal{P}$ denote the real tangent space of this intersection, spanned by the basis $\{\bm{v}^{(2)},\dots,\bm{v}^{(d_{\mathrm{eff}})}\}$. Following the construction in~\cite[Theorem~9.45]{pemantle2024analytic}, we align the tangential directions of $\tilde{\sigma}$ with the purely imaginary subspace $i\mathcal{P}$. This justifies the parametrization $\bm{b}_{\mathrm{p}} = i\bm{y}$ for $\bm{y} \in \mathbb{R}^{d_{\mathrm{eff}}-1}$.

To derive the factorization, we substitute the parametrization $\bm{b}_{\mathrm{p}} = i\bm{y}$ directly into the phase expansion \eqref{eq:phase_expansion_local} established in Lemma~\ref{app:lem:phase_expansion}. The quadratic term transforms as
\begin{align}
    \frac{n}{2}(i\bm{y})^\transp \matr{H}_{\mathrm{p}}(i\bm{y})
    \;=\; -\frac{n}{2} \bm{y}^\transp \matr{H}_{\mathrm{p}} \bm{y}
    \;=\; \frac{\|\bm{r}\|_2}{2} \bm{y}^\transp \mathcal{H} \bm{y},
\end{align}
where the last equality follows from $\mathcal{H} \!=\!  -\|\bm{v}\|_2^{-1}\matr{H}_{\mathrm{p}}$ and $n = \|\bm{r}\|_2/\|\bm{v}\|_2$. Substituting this result back into \eqref{eq:phase_expansion_local} gives
\begin{align}
    \Psi(\bm{\zeta}) = \bm{r}\cdot\bm{\zeta}^* + nq + \frac{\|\bm{r}\|_2}{2} \bm{y}^\transp \mathcal{H} \bm{y} + O(n\|\bm{y}\|_2^3).
\end{align}
Under the dominant local scale, the normal factor $\exp(-nq)$ localizes $q$ to $q=O(n^{-1})$, while the tangential Gaussian factor localizes $\bm{y}$ to $\bm{y}=O(\|\bm{r}\|_2^{-1/2})=O(n^{-1/2})$. The higher-order remainder in the phase then scales as $O(n\|\bm{y}\|_2^3) = O(n^{-1/2})$. Exponentiating $-\Psi(\bm{\zeta})$ then yields the form in \eqref{eq:phase_factorization}.

Finally, since $\tilde\sigma$ is a steepest-descent cycle, the modulus $|\exp(-\Psi)|=\exp(-\mathrm{Re}(\Psi))$ is locally maximized at $\bm\zeta^*$ along $\tilde\sigma$. Equivalently, $\mathrm{Re}(\Psi)$ attains a strict local minimum at $\bm{y}=\bm{0}$. Since along $\tilde{\sigma}$ one can locally express $\mathrm{Re}(\Psi)$ as a constant plus the quadratic term $\frac{\|\bm{r}\|_2}{2}\bm{y}^\transp\mathcal{H}\,\bm{y}$ with $\|\bm{r}\|_2>0$, this strict minimality requires the quadratic form $\bm{y}^\transp\mathcal{H}\,\bm{y}$ to be strictly positive for all $\bm{y}\neq\bm{0}$, implying $\mathcal{H}\succ0$.
\end{proof}

\subsection{Asymptotic Evaluation and Main Results}
\label{app:subsec:asymptotic_evaluation}

We now consolidate the local expansions and geometric transformations derived in the preceding sections to evaluate the reduced Cauchy integral.

\begin{lemma}
\label{lem:integral_factorization}
The coefficient integral $a_{\bm{r}}$ in~\eqref{app:eq:ar_separated} admits the asymptotic approximation
\begin{align}
    a_{\bm{r}}
    & \sim
    \frac{-\|\bm{c}\|_2 \, \bm{w}^{-\bm{r}} P(\bm{w})\, i^{d_{\mathrm{eff}}-1}}{(2\pi i)^{d_{\mathrm{eff}}} \|\bm{v}\|_2} \nonumber \\
    & \quad \cdot \int_{\mathcal{C}_{\bm{c}}}(-q)^{-\alpha} \e^{-nq}\e^{\left( -\frac{\|\bm{r}\|_2}{2} \bm{y}^\transp \mathcal{H} \bm{y} \right)} \d q \wedge \d \bm{y}.
    \label{eq:integral_factorized}
\end{align}
\end{lemma}
\begin{proof}
The result follows directly by combining the reduced integral form in Lemma~\ref{app:lem:integral_separation} with the local phase factorization in Lemma~\ref{app:lem:phase_factorization}. Specifically, since the saddle-point contribution is localized to a neighborhood of $\bm{\zeta}^*$ asymptotically, we may evaluate the leading-order amplitude by $\tilde{P}(\bm{\zeta}) \sim P(\bm{w})$, and substitute the variables $(\tilde{Q}, \bm{b}_{\mathrm{p}})$ with $(-q, i\bm{y})$. The factor $i^{d_{\mathrm{eff}}-1}$ arises from the tangential parametrization $\d\bm{b}_{\mathrm{p}} = i^{d_{\mathrm{eff}}-1}\d\bm{y}$, and the leading negative sign results from the orientation change $\d\tilde{Q} = -\d q$.
\end{proof}


\begin{lemma}
\label{lem:integral_evaluation}
Combining the factorized representation in Lemma~\ref{lem:integral_factorization} with the standard Gaussian and Hankel integral evaluations, the coefficient integral $a_{\bm{r}}$ satisfies the asymptotic estimate as
\begin{align}
    a_{\bm{r}} \sim
    \frac{\|\bm{c}\|_2\bm{w}^{-\bm{r}}P(\bm{w})}{\sqrt{(2\pi\|\bm{r}\|_2)^{d_{\mathrm{eff}}-1}\det(\mathcal{H})}} \cdot \frac{1}{\|\nabla_{\log}Q(\bm{w})\|_2}\cdot \frac{n^{\alpha-1}}{\Gamma(\alpha)}.
    \label{eq:final_asymptotics}
\end{align}
\end{lemma}
\begin{proof}
Starting from the factorized representation in~\eqref{eq:integral_factorized}, the saddle contribution is asymptotically localized to a neighborhood of the critical point $\bm{\zeta}^*$, which corresponds to $(q;\bm{y})=(0;\bm{0})$ under the local coordinates used in Lemma~\ref{lem:integral_factorization}. Hence, for the leading asymptotics we may restrict attention to this local neighborhood and deform the cycle $\mathcal{C}_{\bm{c}}$ accordingly. In this localized setting, $\mathcal{C}_{\bm{c}}$ can be taken to agree, to leading order, with a product contour in the $(q,\bm{y})$-coordinates: the tangential directions are deformed to the steepest-descent slice $\bm{y}\in\mathbb{R}^{d_{\mathrm{eff}}-1}$, while the normal $q$-path is extended to a (truncated) Hankel contour around the branch cut of $(-q)^{-\alpha}$. Letting the truncation radius tend to infinity only changes the integral by an exponentially small amount with the decay factor $\e^{-nq}$. We then evaluate the two components separately:
\begin{itemize}
    \item \textbf{Tangential Component (Geometry)}: 
    Since $\mathcal{H} \succ 0$, the quadratic form defines a convergent Gaussian integral on $\mathbb{R}^{d_{\mathrm{eff}}-1}$. Following the standard result in~\cite[Theorem 5.3]{pemantle2024analytic}, we have
    \begin{align}
        \int_{\mathbb{R}^{d_{\mathrm{eff}}-1}} \exp\left(-\frac{\|\bm{r}\|_2}{2} \bm{y}^\transp \mathcal{H} \bm{y}\right) \d \bm{y} = \frac{(2\pi)^{(d_{\mathrm{eff}}-1)/2}}{\sqrt{\det(\|\bm{r}\|_2 \mathcal{H})}}.
    \end{align}
    Regarding the Cauchy prefactor $(2\pi i)^{-d_{\mathrm{eff}}}$, its combination with the factor $i^{d_{\mathrm{eff}}-1}$ from the volume element $\d\bm{b}_{\mathrm{p}}$ yields
    \begin{align}
        \frac{1}{(2\pi i)^{d_{\mathrm{eff}}}} i^{d_{\mathrm{eff}}-1} = \frac{1}{(2\pi)^{d_{\mathrm{eff}}-1}} \cdot \frac{1}{2\pi i}.
    \end{align}
    Thus, the tangential part carries the real prefactor $(2\pi)^{-(d_{\mathrm{eff}}-1)}$, while the remaining $(2\pi i)^{-1}$ is reserved for the normal Hankel integral. Combining these terms, the net geometric contribution is
    \begin{align}
        \frac{(2\pi)^{(d_{\mathrm{eff}}-1)/2}}{(2\pi)^{d_{\mathrm{eff}}-1} \sqrt{\det(\|\bm{r}\|_2 \mathcal{H})}} = \frac{1}{\sqrt{(2\pi \|\bm{r}\|_2)^{d_{\mathrm{eff}}-1} \det(\mathcal{H})}}.
    \end{align}

    \item \textbf{Normal Component (Singularity)}: 
    To provide a unified treatment for both algebraic singularities ($\alpha \notin \mathbb{Z}$) and polar singularities ($\alpha \in \mathbb{Z}_{>0}$), we adopt the Hankel contour framework (see, e.g., the transfer theorem in~\cite[Theorem~VI.3]{flajolet2009analytic}). Recalling that the normal variable in~\eqref{eq:integral_factorized} is $q=-\tilde Q(\bm{\zeta})$, the singular factor is $(-q)^{-\alpha}$. In the local analysis we may extend the normal $q$-path to the standard Hankel contour $\mathcal{C}_{\mathrm{H}}$ (formally written with endpoints at infinity); this should be understood in the truncated sense discussed above, with the added tails contributing only exponentially small terms because of the factor $\e^{-nq}$. 
    
    The normal integral, together with the leading minus sign from $\d\tilde{Q}=-\d q$ and the reserved prefactor $(2\pi i)^{-1}$, matches the standard representation of the reciprocal Gamma function~\cite[Theorem~B.1]{flajolet2009analytic} as
    \begin{align}
        \mathcal{I}_{\mathrm{norm}} = -\frac{1}{2\pi i}\int_{\mathcal{C}_{\mathrm{H}}} (-q)^{-\alpha} \e^{-n q} \d q.
    \end{align}
    By substituting $s = n q$, we obtain
    \begin{align}
        \mathcal{I}_{\mathrm{norm}} = n^{\alpha-1} \left( -\frac{1}{2\pi i}\int_{\mathcal{C}_{\mathrm{H}}} (-s)^{-\alpha} \e^{-s} \d s \right) = \frac{n^{\alpha-1}}{\Gamma(\alpha)}.
    \end{align}
\end{itemize}

Finally, combining the invariance and geometric factor $\|\bm{c}\|_2 / \|\bm{v}\|_2$, the constant amplitude $\bm{w}^{-\bm{r}} P(\bm{w})$, the tangential Gaussian result, and the explicit evaluation of $\mathcal{I}_{\mathrm{norm}}$, we obtain the final asymptotic formula~\eqref{eq:final_asymptotics}.
\end{proof}

\begin{remark}
\label{rem:grad-scaling-test}
We verify the correctness of the gradient factor by checking invariance under the rescaling $Q \mapsto Q_c \defeq c\,Q$ for any constant $c>0$.
Since the generating function scales as $F_c(\bm{z}) = P(\bm{z})(cQ(\bm{z}))^{-\alpha} = c^{-\alpha}F(\bm{z})$, the asymptotic coefficient must satisfy the exact scaling law
\begin{align}
\label{eq:target-scaling}
    a_{\bm{r}}^{(c)} = c^{-\alpha}\,a_{\bm{r}}.
\end{align}
We now examine how the components of the derived asymptotic formula transform:
\begin{itemize}
    \item The gradient norm scales linearly as $\|\nabla_{\log}Q_c(\bm{w})\|_2 = c\,\|\nabla_{\log}Q(\bm{w})\|_2$.
    \item The scaling parameter $n = \|\bm{r}\|_2 / \|\nabla_{\log}Q(\bm{w})\|_2$ transforms inversely as
    \begin{align*}
         n_c = \frac{\|\bm{r}\|_2}{\|\nabla_{\log}Q_c(\bm{w})\|_2} = c^{-1}n.
    \end{align*}
    \item Consequently, the normal Hankel integral contribution scales as
    \begin{align*}
        \frac{n_c^{\alpha-1}}{\Gamma(\alpha)} = c^{-(\alpha-1)} \frac{n^{\alpha-1}}{\Gamma(\alpha)}.
    \end{align*}
\end{itemize}
To recover the required total scaling $c^{-\alpha}$ in \eqref{eq:target-scaling}, the remaining geometric prefactor must contribute exactly a factor of $c^{-\alpha} / c^{-(\alpha-1)} = c^{-1}$.
Since the gradient norm scales as $c^1$, it must appear with \emph{power one} in the denominator
\begin{align*}
    \frac{1}{\|\nabla_{\log}Q_c(\bm{w})\|_2} = c^{-1} \frac{1}{\|\nabla_{\log}Q(\bm{w})\|_2}.
\end{align*}
In contrast, a squared denominator would scale as $c^{-2}$, leading to an incorrect total scaling of $c^{-(\alpha+1)}$.
This confirms that the formula in the first edition \cite[Theorem~9.5.4]{pemantle2013analytic} (power one) is dimensionally consistent, whereas the expression in the second edition \cite[Theorem~9.45]{pemantle2024analytic} (squared norm) is incorrect.
\end{remark}

\subsection{Specialization to Gibbs and Bethe Partition Sums}
\label{app:subsec:specialization}

We now apply the general asymptotic formula from Lemma~\ref{lem:integral_evaluation} to the Gibbs and Bethe cases to complete the proof of Proposition~\ref{prop:block_asymp}. In both instances, the invariance vector $\bm{c}\!=\!(\bm{1}_m;-\bm{1}_m)$ yields the geometric factor $\|\bm{c}\|_2 \!=\! \sqrt{2m}$. The analytic part $P(\bm{w})$ is evaluated at the critical point where $\lambda_1(\bm{w})\!=\!1$, as specified below.

\subsubsection{Gibbs case ($\alpha\!=\!1$)}
Here the amplitude is $P\!=\!\PGn$ with
\begin{align*}
    \PGn(\bm{w})=\prod_{i=2}^m \frac{1}{1\!-\!\eratio_i}.
\end{align*}
Since $\alpha\!=\!1$, the singularity is a simple pole. The normal component contribution from Lemma~\ref{lem:integral_evaluation} is
\begin{align*}
    \frac{n^{\alpha-1}}{\Gamma(\alpha)} = \frac{n^{0}}{\Gamma(1)} = 1.
\end{align*}
Substituting these into the general formula yields
\begin{align}
    \ZGn \sim
    \frac{n \sqrt{2m} \cdot \bm{w}^{-\bm{r}}}{(2\pi)^{m-1} \|\bm{r}\|_2^m \sqrt{\det(\mathcal{H})}}
    \cdot
    \prod_{i=2}^{m}\frac{1}{1\!-\!\eratio_i},
    \label{app:eq:ZGn}
\end{align}
where we used $\|\nabla_{\log} Q(\bm{w})\|_2 = \|\bm{r}\|_2/n$, which recovers the first expression in Proposition~\ref{prop:block_asymp}.

\subsubsection{Bethe double-cover case ($\alpha\!=\!1/2$)}
Here the amplitude is $P\!=\!\PDBn$, and evaluating at $\lambda_1(\bm{w})\!=\!1$ gives
\begin{align*}
    \PDBn(\bm{w}) = \sqrt{\e} \cdot \prod_{i=2}^m \sqrt{\frac{\e^{\eratio_i}}{1\!-\!\eratio_i}}.
\end{align*}
Since $\alpha\!=\!1/2$, the singularity is of square-root type. The normal component contribution from Lemma~\ref{lem:integral_evaluation} is
\begin{align*}
    \frac{n^{\alpha-1}}{\Gamma(\alpha)} = \frac{n^{-1/2}}{\Gamma(1/2)} = \frac{1}{\sqrt{\pi n}}.
\end{align*}
Combining these components yields
\begin{align}
    \ZDBn \sim
    \frac{n \sqrt{2m} \cdot \bm{w}^{-\bm{r}}}{(2\pi)^{m-1} \|\bm{r}\|_2^m \sqrt{\det(\mathcal{H})}}
    \cdot
    \sqrt{\frac{\e}{\pi n}}
    \cdot
    \prod_{i=2}^{m}\sqrt{\frac{\e^{\eratio_i}}{1\!-\!\eratio_i}},
    \label{app:eq:ZBn}
\end{align}
again using $\|\nabla_{\log} Q(\bm{w})\|_2 = \|\bm{r}\|_2/n$, which recovers the second expression in Proposition~\ref{prop:block_asymp}.

\qed

To end this section of proving Proposition~\ref{prop:block_asymp}, we discuss a special case to verify the consistency of our result.

\begin{example}
\label{app:ex:all_one_matrix}
Consider the case where the matrix of interest $\matr{A}$ is the all-one matrix $\Ones_n$, where the base matrix is the all-one matrix $\matr{B} = \Ones_m$. We evaluate the asymptotic Gibbs permanent of the block matrix $\matr{A}$ constructed with uniform block sizes $k_i = \ell_j = \bar{n} \defeq n/m$. In this setting, $\matr{A}$ is effectively the all-one matrix of size $n \times n$. While setting $m\!=\!1$ is the most direct approach, keeping a general $m$ serves as a robust consistency check to verify that the parameter cancels out in the final asymptotic expression.

The defining function simplifies to
\begin{align*}
    Q(\bm{z}) = 1 \!-\! \Biggl(\sum_{i\in[m]} t_i\Biggr)\cdot \Biggl(\sum_{j\in[m]} u_j\Biggr).
\end{align*}
By symmetry and scaling invariance, we set the critical point to be $\bm{w} \!=\! m^{-1} \bm{1}_{2m}$. We summarize the key analytic quantities evaluated at this point:
\begin{itemize}
    \item Amplitude: $\PGn(\bm{w}) = 1$.
    \item Gradient norm: $\|\nabla_{\log} Q(\bm{w})\|_2 \!=\! \|\bm{v}\|_2 \!=\! \sqrt{2/m}$.
    \item Spectral factor: Since $\lambda_1\!=\!1$ is the only non-zero eigenvalue of the base matrix $\matr{B}=\Ones_m$, the product of spectral factors simplifies to $\prod_{i=2}^m (1\!-\!\lambda_i)^{-1} = 1$.
\end{itemize}

To determine the Hessian determinant, we explicitly compute the derivatives in the log-domain at the symmetric critical point $\bm{w} = m^{-1}\bm{1}_{2m}$. First, for the gradient $\bm{v} = \nabla_{\log} Q(\bm{w})$, the components satisfy
\begin{align}
    v_{t_k} 
    &= \frac{\partial}{\partial \log t_k}\left(1 - \Biggl(\sum_{i\in[m]} t_i\Biggr)\cdot \Biggl(\sum_{j\in[m]} u_j\Biggr)\right) \nonumber \\
    &= -t_k \left(\sum_{j\in[m]} u_j\right)  = -\frac{1}{m}.
\end{align}
By symmetry, $v_{u_k} = -1/m$, yielding the gradient vector $\bm{v} = -m^{-1}(\bm{1}_m; \bm{1}_m)$ with norm $\|\bm{v}\|_2 = \sqrt{2/m}$.

Next, we evaluate the entries of the Hessian matrix $\matr{H}_{\bm{\zeta}^*} = \nabla^2 \tilde{Q}(\bm{\zeta}^*)$. The diagonal blocks ($\bm{t}$-$\bm{t}$ interactions) and off-diagonal blocks ($\bm{t}$-$\bm{u}$ interactions) are derived as
\begin{align}
    \frac{\partial^2 Q}{\partial \log t_i \, \partial \log t_j}
    & =  \frac{\partial}{\partial \log t_j} \left( \!-t_i \sum_{k} u_k\! \right) 
    \overset{\mathrm{(b)}}{=} -\frac{1}{m} \delta_{ij}, \\
    \frac{\partial^2 Q}{\partial \log t_i \, \partial \log u_j}
    &= \frac{\partial}{\partial \log u_j} \left( \!-t_i \sum_{k} u_k\! \right) 
    \overset{\mathrm{(c)}}{=} -t_i u_j = -\frac{1}{m^2},
\end{align}
where step $\mathrm{(b)}$ uses the linearity of the first derivative ($\partial_{t_j} t_i = \delta_{ij} t_i$), and step $\mathrm{(c)}$ captures the cross-variable interaction. Consequently, the full Hessian matrix assembles into the block form
\begin{align}
    \matr{H}_{\bm{\zeta}^*} = 
    \begin{pmatrix}
        -\frac{1}{m} \matr{I}_{m} & -\frac{1}{m^2} \Ones_m \\[6pt]
        -\frac{1}{m^2} \Ones_m & -\frac{1}{m} \matr{I}_{m}
    \end{pmatrix},
\end{align}
where $\matr{I}_m$ denotes the identity matrix of size $m \!\times\!m$ and $\Ones_m$ denotes the all-one matrix of size $m \!\times\! m$ (see Notation in Section~\ref{sec:intro}).

We project the full Hessian matrix onto the tangent space orthogonal to $\mathrm{span}(\bm{v}, \bm{c})$. Observing that the gradients and constraints span the block-wise constant vectors, i.e., $\mathrm{span}\{(\bm{1}_m; \bm{1}_m), (\bm{1}_m; -\bm{1}_m)\} = \mathrm{span}\{(\bm{1}_m; \bm{0}), (\bm{0}; \bm{1}_m)\}$, the tangent space decouples into two independent zero-sum subspaces for $\bm{t}$ and $\bm{u}$.

To construct the projection explicitly, let $\matr{E} \in \mathbb{R}^{m \times (m-1)}$ be a matrix whose columns form an orthonormal basis for the zero-sum subspace $\{\bm{x} \in \mathbb{R}^m \mid \bm{1}_m^\transp \bm{x} = 0\}$. The global projection matrix $\matr{V} \in \mathbb{R}^{2m \times (2m-2)}$ is then block-diagonal as
\begin{align}
    \matr{V} \defeq 
    \begin{pmatrix}
        \matr{E} & \matr{0} \\
        \matr{0} & \matr{E}
    \end{pmatrix}.
\end{align}
Utilizing the properties $\matr{E}^\transp \matr{E} = \matr{I}_{m-1}$ and $\Ones_m \matr{E} = \matr{0}$, the projection eliminates the off-diagonal blocks of $\matr{H}_{\bm{\zeta}^*}\!$ while preserving the identity components of the diagonal blocks. This yields the projected Hessian
\begin{align}
    \matr{H}_{\mathrm{p}} = \matr{V}^\transp \matr{H}_{\bm{\zeta}^*} \matr{V} 
    = \begin{pmatrix}
        -\frac{1}{m} \matr{I}_{m-1} & \matr{0} \\
        \matr{0} & -\frac{1}{m} \matr{I}_{m-1}
    \end{pmatrix}
    = -\frac{1}{m} \matr{I}_{2m-2}.
\end{align}
Finally, we compute the normalized Hessian matrix $\mathcal{H} = -\|\bm{v}\|_2^{-1} \matr{H}_{\mathrm{p}}$. With the scaling factor $\|\bm{v}\|_2^{-1} = \sqrt{m/2}$, the matrix becomes a scaled identity $\mathcal{H} = (2m)^{-1/2} \matr{I}_{2m-2}$, leading to the determinant
\begin{align}
    \det(\mathcal{H}) = \left( (2m)^{-1/2} \right)^{2m-2} = (2m)^{-(m-1)}.
\end{align}

We now apply the general asymptotic formula derived in~\eqref{app:eq:ZGn}. Substituting the specific geometric quantities derived above—namely the partition norm $\|\bm{r}\|_2 = \bar{n}\sqrt{2m}$, the saddle point value $\bm{w}^{-\bm{r}} = m^{2n}$, and the Hessian determinant $\det(\mathcal{H}) = (2m)^{-(m-1)}$ allows us to evaluate the coefficient $\ZGn$ directly as
\begin{align}
    \ZGn 
    &\sim \frac{n\sqrt{2m}\cdot\bm{w}^{-\bm{r}}}{(2\pi)^{m-1}\|\bm{r}\|_2^m\sqrt{\det(\mathcal{H})}} \cdot \prod_{i=2}^{m}\frac{1}{1-\eratio_i} \nonumber \\
    &\overset{\mathrm{(a)}}{=} \frac{n\sqrt{2m} \cdot m^{2n}}{(2\pi)^{m-1}\bigl(\bar{n}\sqrt{2m}\bigr)^m\sqrt{(2m)^{-(m-1)}}} \nonumber \\
    &\overset{\mathrm{(b)}}{=} \frac{n\sqrt{2m} \cdot m^{2n}}{(2\pi)^{m-1}\bar{n}^m\sqrt{(2m)^{m-(m-1)}}} \nonumber \\
    &\overset{\mathrm{(c)}}{=} \frac{m \cdot m^{2n}}{(2\pi \bar{n})^{m-1}},
\end{align}
where step $\mathrm{(a)}$ substitutes the parameter values, step $\mathrm{(b)}$ simplifies the $(2m)$ factor, and step $\mathrm{(c)}$ uses $n/\bar{n}^m = m/\bar{n}^{m-1}$.

Finally, we bridge the asymptotic coefficient to the permanent using the relation established in Lemma~\ref{lem:bridge-ogf-nfg}. Applying the Stirling approximation for the factorials $\bm{k}! \cdot \bm{\ell}! \sim (2\pi \bar{n})^m (\bar{n}/\e)^{2n}$ alongside our derived coefficient yields
\begin{align}
    \perm(\matr{A})^2 
    &\sim \bm{k}! \cdot \bm{\ell}! \cdot \ZGn \nonumber \\
    &\sim \left( (2\pi \bar{n})^m \left(\frac{\bar{n}}{\e}\right)^{2n} \right) \cdot \left( \frac{m \cdot m^{2n}}{(2\pi \bar{n})^{m-1}} \right) \nonumber \\
    &= (2\pi \bar{n}) \cdot m \cdot m^{2n} \cdot \left(\frac{\bar{n}}{\e}\right)^{2n} \nonumber \\
    &= 2\pi n \left(\frac{n}{\e}\right)^{2n},
\end{align}
where the last equality follows by substituting $\bar{n} = n/m$, demonstrating the complete cancellation of the auxiliary parameter $m$ via the identities $m^{-2n} \cdot m^{2n} = 1$ and $(n/m) \cdot m = n$. Taking the square root recovers
\begin{align}
    \perm(\matr{A}) \sim \sqrt{2\pi n} \left(\frac{n}{\e}\right)^{n} \sim n!.
\end{align}
\end{example}

\begin{remark}
\label{rem:direct_expansion}
For the simple case of the all-one matrix discussed in Example~\ref{app:ex:all_one_matrix}, where $P_G(\bm{z}) = 1$ and $Q(\bm{z}) = 1 \!-\! \bigl(\sum_{i\in[m]} t_i\bigr)\bigl(\sum_{j\in[m]} u_j\bigr)$, the generating function $F_G(\bm{z}) = \PGn(\bm{z})/Q(\bm{z})$ admits a direct expansion using the multinomial theorem
\begin{align}
    F_G(\bm{z}) & = \frac{\PGn(\bm{z})}{Q(\bm{z})} = \frac{1}{Q(\bm{z})} \nonumber\\
    &= \sum_{n=0}^{\infty} \left( \Biggl(\sum_{i\in[m]} t_i\Biggr) \Biggl(\sum_{j\in[m]} u_j\Biggr) \right)^n \nonumber \\
    &= \sum_{n=0}^{\infty} \left( \sum_{\bm{k}} \binom{n}{\bm{k}} \bm{t}^{\bm{k}} \right) \left( \sum_{\bm{\ell}} \binom{n}{\bm{\ell}} \bm{u}^{\bm{\ell}} \right),
\end{align}
where $\bm{t} = (t_1; \dots; t_m)$ and $\bm{u} = (u_1; \dots; u_m)$. The inner sums range over all multi-indices $\bm{k}, \bm{\ell} \in \mathbb{Z}_{\ge 0}^m$ that satisfy the multinomial degree constraint $|\bm{k}| = |\bm{\ell}| = n$, where $|\bm{k}| \defeq \sum_{i\in[m]} k_i$. Extracting the coefficient for $\bm{t}^{\bm{k}}\bm{u}^{\bm{\ell}}$ and multiplying by $\bm{k}! \cdot \bm{\ell}!$ yields the exact squared permanent
\begin{align}
    \perm(\matr{A})^2 = \left( \binom{n}{\bm{k}} \binom{n}{\bm{\ell}} \right) \cdot \bm{k}! \cdot \bm{\ell}! = (n!)^2.
\end{align}
This exact result matches the Stirling approximation derived via ACSV. It is worth noting that such a direct combinatorial expansion is feasible only for trivial base matrices like $\matr{B}=\Ones_m$. For general block matrices, the generating function lacks such simple structure, making the ACSV framework essential for deriving asymptotic approximations.
\end{remark}

\section{Deriving $(\bm{t^*}\!,\bm{u^*}\!)$ from Sinkhorn's Theorem}
\label{app:fp-derive-tu-m2}

\newcommand*{\vl}{\bm{\overleftarrow{V}}}
\newcommand*{\vr}{\bm{\overrightarrow{V}}}
\newcommand*{\vleft}[1]{\overleftarrow{V}_{\mkern-6mu #1}}
\newcommand*{\vright}[1]{\overrightarrow{V}_{\mkern-6mu #1}}

It turns out that the ACSV saddle-point parameters $(\bm{t^*}\!, \bm{u^*}\!)$ can be obtained with the help of Sinkhorn's theorem. We restrict attention to the two-type block-constant setting $m=2$, where $\matr{B}=(b_{i,j})\in\mathbb{R}_{>0}^{2\times2}$ and the expanded matrix $\matr{A}$ consists of four constant blocks of sizes $k_i\times \ell_j$ ($i,j\in\{1,2\}$), with every entry in block $(i,j)$ equal to $b_{i,j}$. Note, however, that the results in this appendix can be straightforwardly generalized to any positive integer $m$.

In this section, we use $\partial_x \defeq \partial/\partial x$ and $\partial_y \defeq \partial/\partial y$ for derivatives with respect to logarithmic coordinates $x,y \in \{\log t_1,\dots,\log t_m,\log u_1,\dots,\log u_m\}$. To avoid ambiguity, we specify that the operator $\partial$ applies strictly to the single object immediately following it.

We begin by recalling the classical Sinkhorn theorem~\cite{sinkhorn1964relationship}, which characterizes the diagonal scaling of strictly positive matrices.

\begin{lemma}
\label{app:lem:sinkhorn-general}
For any strictly positive square matrix $\matr{A}$, there exists a unique doubly stochastic matrix $\matr{U}$ expressible in the form
\begin{align}
    \matr{U} = \matr{D}_1 \matr{A} \matr{D}_2,
    \label{app:eq:sinkhorn-scaling-form}
\end{align}
where $\matr{D}_1$ and $\matr{D}_2$ are diagonal matrices with strictly positive diagonal entries. While $\matr{U}$ is unique, the scaling matrices $\matr{D}_1$ and $\matr{D}_2$ are unique up to a common scalar factor.
\end{lemma}

Since $\matr{U}$ is doubly stochastic, its row and column sums must equal $1$. Combined with Eq.~\eqref{app:eq:sinkhorn-scaling-form}, this implies constraints on the diagonal entries of $\matr{D}_1$ and $\matr{D}_2$. Leveraging the block-constant structure of $\matr{A}$ described in Assumption~\ref{assump:general_block}, we can explicitly characterize these constraints.

\begin{lemma}
\label{app:lem:sinkhorn-block-structure}
Let $\matr{A}$ be the block-constant matrix defined above. There exist positive scalars $\vright{1}, \vright{2}, \vleft{1}, \vleft{2}$ satisfying the coupled system
\begin{align}
    \vright{1} & = \frac{1}{\ell_1b_{1,1}\vleft{1}\!+\!\ell_2b_{1,2}\vleft{2}},\label{app:eq:fp-m2-vr1-derive}\\
    \vright{2} & = \frac{1}{\ell_1b_{2,1}\vleft{1}\!+\!\ell_2b_{2,2}\vleft{2}},\label{app:eq:fp-m2-vr2-derive}\\
    \vleft{1} & = \frac{1}{k_1b_{1,1}\vright{1}\!+\!k_2b_{2,1}\vright{2}},\label{app:eq:fp-m2-vl1-derive}\\
    \vleft{2} & = \frac{1}{k_1b_{1,2}\vright{1}\!+\!k_2b_{2,2}\vright{2}},\label{app:eq:fp-m2-vl2-derive}
\end{align}
such that the diagonal scaling matrices in \eqref{app:eq:sinkhorn-scaling-form} are given by
\begin{align}
    \matr{D}_1 & = \diag(\underbrace{\vright{1},\dots,\vright{1}}_{k_1}, \underbrace{\vright{2},\dots,\vright{2}}_{k_2}),
    \\
    \matr{D}_2 & = \diag(\underbrace{\vleft{1},\dots,\vleft{1}}_{\ell_1}, \underbrace{\vleft{2},\dots,\vleft{2}}_{\ell_2}).
\end{align}
\end{lemma}

\begin{proof}
The block-constant structure of $\matr{A}$ implies that all rows of type $i$ are identical, as are all columns of type $j$. By the uniqueness of the Sinkhorn scaling, the associated diagonal factors in $\matr{D}_1$ and $\matr{D}_2$ must be uniform within each type. Substituting these uniform values into the row-sum constraint $\sum_{j\in[n]} u_{1,j} = 1$ yields
\begin{align*}
    1 & = \sum_{j\in[n]} \vright{1} A_{1,j} (\matr{D}_2)_{jj} \\
      & = \vright{1} \left( \sum_{j \in \text{ Type } 1} b_{1,1} \vleft{1} + \sum_{j \in \text{ Type } 2} b_{1,2} \vleft{2} \right) \\
      & = \vright{1} (\ell_1 b_{1,1} \vleft{1} + \ell_2 b_{1,2} \vleft{2}).
\end{align*}
Rearranging this equation recovers \eqref{app:eq:fp-m2-vr1-derive}. The remaining equations \eqref{app:eq:fp-m2-vr2-derive}--\eqref{app:eq:fp-m2-vl2-derive} follow symmetrically from the constraints for row type 2 and both column types.
\end{proof}

Theorem~\ref{app:thm:sinkhorn-to-tu-m2} below expresses the ACSV parameters in terms of these Sinkhorn scaling factors.

\begin{theorem}
\label{app:thm:sinkhorn-to-tu-m2}
Let $m\!=\!2$ and $\matr{B}\in\mathbb{R}_{>0}^{2\times 2}$. Fix $\bm{r}=(k_1;k_2;\ell_1;\ell_2)$ with $k_1\!+k_2=\ell_1\!+\ell_2=n$. Let $\vright{1},\vright{2},\vleft{1},\vleft{2}$ be the positive scalars given by Lemma~\ref{app:lem:sinkhorn-block-structure}. Construct the ACSV parameters $(\bm{t^*}\!,\bm{u^*}\!)$ explicitly as
\begin{align}
    t_1^* = k_1 \vright{1}^2,\quad t_2^* = k_2 \vright{2}^2,\quad
    u_1^* = \ell_1 \vleft{1}^2,\quad u_2^* = \ell_2 \vleft{2}^2.
    \label{eq:thm-def-tu-explicit}
\end{align}
The choice $\bm{w}=(\bm{t^*}\!;\bm{u^*}\!)$ constitutes a solution to the general ACSV critical point equations \eqref{app:eq:critical_eq} with scaling factor $\eta=n$. Specifically:

\begin{itemize}
    \item \textbf{Singularity constraint:} The condition $Q(\bm{w})\!=\!0$ corresponds to the unit spectral radius requirement for the symmetric state-transition matrix $\matr{S}(\bm{t^*}\!,\bm{u^*}\!)$ in Definition~\ref{def:S-general}, i.e.,
    \begin{align}
        \lambda_1(\bm{t^*}\!,\bm{u^*}\!) = 1.
    \end{align}

    \item \textbf{Gradient alignment:} The condition $\bm{r} = -\eta \nabla_{\log} Q(\bm{w})$ is satisfied componentwise, meaning the scaled logarithmic derivatives of the dominant eigenvalue match the respective block dimensions. That is, with $\eta=n$, for all $i,j\in\{1,2\}$,
    \begin{align}
        n \left. \frac{\partial \log \lambda_1}{\partial \log t_i} \right|_{\bm{w}} = k_i,
        \qquad
        n \left. \frac{\partial \log \lambda_1}{\partial \log u_j} \right|_{\bm{w}} = \ell_j.
    \end{align}
\end{itemize}
\end{theorem}

\begin{proof}
We first define the scaling vectors
\begin{align*}
    \vr \defeq 
    \begin{pmatrix}
        \vright{1}\\ 
        \vright{2}
    \end{pmatrix},
    \qquad
    \vl \defeq 
    \begin{pmatrix}
        \vleft{1}\\ 
        \vleft{2}
    \end{pmatrix}.
\end{align*} 
The Sinkhorn system \eqref{app:eq:fp-m2-vr1-derive}--\eqref{app:eq:fp-m2-vl2-derive} can be rewritten in matrix form as
\begin{align}
    \begin{pmatrix}
        \frac{1}{\vright{1}}\\[2pt]
        \frac{1}{\vright{2}}
    \end{pmatrix}
    &=
    \matr{B}
    \begin{pmatrix}
        \ell_1 & 0\\
        0 & \ell_2
    \end{pmatrix}
    \vl,
    \label{app:eq:fp-mat-left-derive}\\[6pt]
    \begin{pmatrix}
        \frac{1}{\vleft{1}}\\[2pt]
        \frac{1}{\vleft{2}}
    \end{pmatrix}
    &=
    \matr{B}^{\transp}
    \begin{pmatrix}
        k_1 & 0\\
        0 & k_2
    \end{pmatrix}
    \vr.
    \label{app:eq:fp-mat-right-derive}
\end{align}
To facilitate the analysis, we introduce the auxiliary matrix $\matr{M}(\bm{t},\bm{u})$, given as
\begin{align}
    \matr{M}(t_1,t_2,u_1,u_2)
    & \defeq
    \begin{pmatrix}
        \ell_1^{-1}u_1 & 0\\
        0 & \ell_2^{-1}u_2
    \end{pmatrix}
    \matr{B}^{\transp}
    \begin{pmatrix}
        k_1 & 0\\
        0 & k_2
    \end{pmatrix}\nonumber\\
    &\qquad \cdot
    \begin{pmatrix}
        k_1^{-1}t_1 & 0\\
        0 & k_2^{-1}t_2
    \end{pmatrix}
    \matr{B}
    \begin{pmatrix}
        \ell_1 & 0\\
        0 & \ell_2
    \end{pmatrix}.
    \label{app:eq:def-M-m2-derive}
\end{align}
Note that $\matr{M}$ is similar to the weighted transition matrix $\matr{S}(\bm{t},\bm{u})$ and thus shares its eigenvalues.

Next, we claim that $\bm{\phi} \defeq \vl$ is the right Perron eigenvector of $\matr{M}(\bm{t^*}\!,\bm{u^*}\!)$ with the eigenvalue $1$. Substituting the constructed parameters into the definition of $\matr{M}$ yields
\begin{align*}
    \matr{M}\vl
    &=
    \begin{pmatrix}
        \ell_1^{-1}u_1^* & 0\\
        0 & \ell_2^{-1}u_2^*
    \end{pmatrix}
    \matr{B}^{\transp}
    \begin{pmatrix}
        k_1 & 0\\
        0 & k_2
    \end{pmatrix}\\
    & \qquad\cdot
    \begin{pmatrix}
        k_1^{-1}t_1^* & 0\\
        0 & k_2^{-1}t_2^*
    \end{pmatrix}
    \matr{B}
    \begin{pmatrix}
        \ell_1 & 0\\
        0 & \ell_2
    \end{pmatrix}
    \vl \\
    &\overset{\mathrm{(a)}}{=}
    \begin{pmatrix}
        \ell_1^{-1}u_1^* & 0\\
        0 & \ell_2^{-1}u_2^*
    \end{pmatrix}
    \matr{B}^{\transp}
    \begin{pmatrix}
        k_1 & 0\\
        0 & k_2
    \end{pmatrix}
    \vr \\
    &\overset{\mathrm{(b)}}{=}
    \vl,
\end{align*}
where step $\mathrm{(a)}$ follows from applying \eqref{app:eq:fp-mat-left-derive} to the right-most matrix-vector product, and step $\mathrm{(b)}$ follows from applying \eqref{app:eq:fp-mat-right-derive} to the remaining expression. Thus, $\matr{M}\vl = \vl$, establishing that $1$ is an eigenvalue of $\matr{M}(\bm{t^*}\!,\bm{u^*}\!)$. To identify this eigenvalue as the dominant one, we invoke the Perron--Frobenius theorem. Since $\matr{M}$ is strictly positive, the Perron--Frobenius theorem implies that the spectral radius of $\matr{M}$ is the unique eigenvalue admitting a strictly positive eigenvector, and the associated positive eigenvector is unique up to scaling. Because $\vl>0$ and $\matr{M}\vl=\vl$, we must have $\lambda_1(\bm{t^*}\!,\bm{u^*}\!)=1$.

To analyze the gradient alignment constraints in~\eqref{app:eq:eta-component-k}--\eqref{app:eq:eta-component-ell}, we need to compute the logarithmic derivatives of the dominant eigenvalue at $\bm{w}=(\bm{t^*};\bm{u^*}\!)$. Define the row vector $\bm{\psi}^{\transp} \defeq (\ell_1/\vleft{1}, \ell_2/\vleft{2})$. Its normalization with respect to $\bm{\phi}$ is
\begin{align*}
    \bm{\psi}^\transp \bm{\phi}
    =
    \begin{pmatrix}
        \ell_1/\vleft{1} & \ell_2/\vleft{2}
    \end{pmatrix}
    \begin{pmatrix}
        \vleft{1}\\ \vleft{2}
    \end{pmatrix}
    =\ell_1+\ell_2=n.
\end{align*}
A similar direct calculation using \eqref{app:eq:fp-mat-left-derive}--\eqref{app:eq:fp-mat-right-derive} verifies that $\bm{\psi}^\transp \matr{M} = \bm{\psi}^\transp$. Thus, $\bm{\psi}^{\transp}$ is the left Perron eigenvector of $\matr{M}$ with the eigenvalue $1$.

We seek to evaluate the sensitivity of the eigenvalue with respect to the parameter $x = \log t_1$. Although the ACSV gradient alignment condition requires the logarithmic derivative $\partial \log \lambda_1 / \partial \log t_1$, the manifold constraint $\lambda_1(\bm{t^*}\!,\bm{u^*}\!)=1$ ensures that this coincides with the derivative $\partial_x\lambda_1$ at the solution $\bm{w}\!=\!(\bm{t^*}\!;\bm{u^*}\!)$. We determine this value using standard eigenvalue perturbation theory (see, e.g.,~\cite{kato2013perturbation}).

Consider a small perturbation $\varepsilon$ in the parameter $x = \log t_1$. This induces perturbations in the matrix $\matr{M} \to \matr{M} \!+\! \varepsilon \partial_x\matr{M} \!+\! O(\varepsilon^2)$, the eigenvalue $\lambda_1 \to \lambda_1 \!+\! \varepsilon \partial_x\lambda_1 \!+\! O(\varepsilon^2)$, and the right Perron eigenvector $\bm{\phi} \to \bm{\phi} + \varepsilon \partial_x\bm{\phi} + O(\varepsilon^2)$. The eigenvalue equation for the perturbed system is
\begin{align*}
    & \bigl(\matr{M} + \varepsilon \partial_x\matr{M} + O(\varepsilon^2)\bigr) \bigl(\bm{\phi} + \varepsilon \partial_x\bm{\phi} + O(\varepsilon^2)\bigr) \\
    & = 
    \bigl(\lambda_1 + \varepsilon \partial_x\lambda_1 + O(\varepsilon^2)\bigr) \bigl(\bm{\phi} + \varepsilon \partial_x\bm{\phi} + O(\varepsilon^2)\bigr).
\end{align*}
Expanding both sides and retaining only terms linear in $\varepsilon$, we obtain
\begin{align*}
    \matr{M}\bm{\phi} + \varepsilon(\matr{M}\partial_x\bm{\phi} \!+\! \partial_x\matr{M}\bm{\phi}) 
    \approx 
    \lambda_1\bm{\phi} + \varepsilon(\lambda_1\partial_x\bm{\phi} \!+\! \partial_x\lambda_1\bm{\phi}).
\end{align*}
Since $\matr{M}\bm{\phi} = \lambda_1\bm{\phi}$, the zero-order terms cancel. Equating the coefficients of $\varepsilon$ yields
\begin{align}
    \matr{M}\partial_x\bm{\phi} + \partial_x\matr{M}\bm{\phi} = \lambda_1\partial_x\bm{\phi} + \partial_x\lambda_1\bm{\phi}.
\end{align}

To isolate $\partial_x\lambda_1$, we left-multiply the equation by the left Perron eigenvector $\bm{\psi}^\transp$, yielding
\begin{align*}
    \bm{\psi}^\transp \matr{M}\partial_x\bm{\phi}
    + \bm{\psi}^\transp \partial_x\matr{M}\bm{\phi}
    &=
    \lambda_1 \bm{\psi}^\transp \partial_x\bm{\phi}
    + \partial_x\lambda_1 \bm{\psi}^\transp \bm{\phi}.
\end{align*}
Using $\bm{\psi}^\transp \matr{M} = \lambda_1 \bm{\psi}^\transp$, the term $\lambda_1 \bm{\psi}^\transp \partial_x\bm{\phi}$ cancels on both sides. Solving for $\partial_x\lambda_1$, we arrive at
\begin{align}
    \partial_x\lambda_1
    =
    \frac{\bm{\psi}^{\transp}\partial_x\matr{M}\bm{\phi}}{\bm{\psi}^{\transp}\bm{\phi}}.
    \label{app:eq:pertur}
\end{align}
Using the explicit form of $\partial_x\matr{M}$ and the parameters, the numerator of the RHS of~\eqref{app:eq:pertur} becomes
\begin{align*}
    \bm{\psi}^{\transp}\partial_x\matr{M}\bm{\phi}
    &=
    \bm{\psi}^{\transp}
    \begin{pmatrix}
        \ell_1^{-1}u_1^* & 0\\
        0 & \ell_2^{-1}u_2^*
    \end{pmatrix}
    \matr{B}^{\transp}
    \begin{pmatrix}
        k_1 & 0\\
        0 & k_2
    \end{pmatrix} \nonumber \\
    &\quad \cdot
    \begin{pmatrix}
        k_1^{-1}t_1^* & 0\\
        0 & 0
    \end{pmatrix}
    \matr{B}
    \begin{pmatrix}
        \ell_1 & 0\\
        0 & \ell_2
    \end{pmatrix}
    \vl \nonumber \\
    &\overset{\mathrm{(a)}}{=}
    k_1 \vright{1}
    \begin{pmatrix}
        \ell_1/\vleft{1} & \ell_2/\vleft{2}
    \end{pmatrix}
    \matr{B}^{\transp}
    \begin{pmatrix}
        1 \\ 0
    \end{pmatrix} \nonumber \\
    &\overset{\mathrm{(b)}}{=}
    k_1,
\end{align*}
where step $\mathrm{(a)}$ uses the relation derived in \eqref{app:eq:fp-mat-left-derive} that the right-most product yields $(\vright{1}, 0)^\transp$, and step $\mathrm{(b)}$ simplifies the scalar product using the first component of \eqref{app:eq:fp-mat-left-derive}. Consequently, we obtain
\begin{align}
    \partial_x\lambda_1
    =
    \frac{\bm{\psi}^\transp \partial_x\matr{M}\bm{\phi}}{\bm{\psi}^\transp \bm{\phi}}
    =
    \frac{k_1}{n}.
\end{align}
This calculation yields $k_1/n$ in the $\log t_1$ direction, confirming that the parameters derived from the Sinkhorn solution satisfy the componentwise ACSV constraint $k_1 = \eta \frac{\partial \log \lambda_1}{\partial \log t_1}$ with the scaling factor $\eta=n$. The corresponding conditions for $t_2, u_1, u_2$ follow immediately due to the symmetry of the setup.
\end{proof}

Finally, we highlight a connection between the system of equations in Lemma~\ref{app:lem:sinkhorn-block-structure} and the scaled Sinkhorn permanent $\perm_{\scSink}(\matr{A})$. Since the parameters $\vright{i}$ and $\vleft{j}$ are precisely the scaling factors in Sinkhorn's matrix-scaling iterations, the ACSV critical point $\bm{w}\!=\!(\bm{t^*}\!;\bm{u^*}\!)$ is intrinsically linked to the value of $\perm_{\scSink}(\matr{A})$.

We start by introducing the Sinkhorn permanent and its scaled variant. Fundamentally, these quantities are rooted in a variational framework. We refer the reader to~\cite{6804280} and~\cite{pmlr-v134-anari21a} for the explicit expressions of the Sinkhorn permanent $\perm_{\Sink}(\matr{A})$ and the scaled Sinkhorn permanent $\perm_{\scSink}(\matr{A})$ as the minima of the unscaled and scaled Sinkhorn free energy minimization problems, respectively. Following the notation in~\cite{10463076}, the scaled Sinkhorn permanent is related to the (unscaled) Sinkhorn permanent by the scaling factor $\e^{-n}$, i.e.,
\begin{align}
    \perm_{\scSink}(\matr{A}) \defeq \e^{-n}\cdot\perm_{\Sink}(\matr{A}).
    \label{app:eq:scsink-def}
\end{align}

Rearranging terms in~\eqref{app:eq:sinkhorn-scaling-form} yields the decomposition $\matr{A} = \matr{D}_1^{-1} \matr{U} \matr{D}_2^{-1}$. Based on this factorization, the Sinkhorn approximation is defined as the product of the permanents of these inverse scaling matrices, i.e., $\perm_{\Sink}(\matr{A}) \defeq \perm(\matr{D}_1^{-1})\cdot \perm(\matr{D}_2^{-1})$. Since the exact permanent factorizes as $\perm(\matr{A}) = \perm_{\Sink}(\matr{A})\perm(\matr{U})$, the accuracy of this approximation depends entirely on the permanent of the doubly stochastic matrix $\matr{U}$. Utilizing the established bounds $n!/n^n \le \perm(\matr{U}) \le 1$~\cite{Egorychev1981,Falikman1981}, the range for the ratio of the permanent of $\matr{A}$ to its scaled Sinkhorn approximation is
\begin{align}
    \e^{n}\cdot \frac{n!}{n^{n}}
    \;\le\;
    \frac{\perm(\matr{A})}{\perm_{\scSink}(\matr{A})}
    \;\le\;
    \e^{n},
    \label{app:eq:scsink-bounds}
\end{align}
where the lower bound scales as $\e^{n}n!/n^{n}\approx \sqrt{2\pi n}$ by Stirling's formula.

With these results, we can rewrite the exact asymptotic expressions for $\ZGn$ and $\ZDBn$ in Proposition~\ref{prop:block_asymp}. 

\begin{proposition}
\label{prop:acsv_exact_general}
Let $m$ be a positive integer. As $n \to \infty$, we can rewrite $\ZGn(\bm{k},\bm{\ell})$ and $\ZDBn(\bm{k},\bm{\ell})$ in Proposition~\ref{prop:block_asymp} as
\small{
\begin{align}
    \ZGn
    & \!\sim\!
    \frac{n\sqrt{2m}\cdot \bm{k}^{-\bm{k}}\bm{\ell}^{-\bm{\ell}} \!\cdot\! \e^{2n} \!\cdot\! \bigl(\perm_{\scSink}(\matr{A})\bigr)^2}{(2\pi)^{m-1} \|\bm{r}\|_2^m \sqrt{\det(\mathcal{H})}}
    \cdot
    \prod_{i=2}^m \frac{1}{1\!-\!\eratio_i},
    \label{eq:ZGn_explicit_final} \\
    \ZDBn
    & \!\sim\!
    \frac{n\sqrt{2m}\cdot\! \bm{k}^{-\bm{k}}\bm{\ell}^{-\bm{\ell}} \!\cdot\! \e^{2n} \!\cdot\! \bigl(\perm_{\scSink}(\matr{A})\bigr)^2}{(2\pi)^{m-1} \|\bm{r}\|_2^m \sqrt{\det(\mathcal{H})}}
    \!\cdot\!
    \sqrt{\!\frac{\e}{\pi n}}
    \!\cdot\!
    \prod_{i=2}^m \sqrt{\!\frac{\e^{\eratio_i}}{1\!-\!\eratio_i}}.
    \label{eq:ZBn_explicit_final}
\end{align}
}
Here, $\|\bm{r}\|_2 = \sqrt{\sum_{i\in[m]} k_i^2 + \sum_{j\in[m]} \ell_j^2}$, $\det(\mathcal{H})$ denotes the determinant of the Hessian matrix of the constraint function restricted to the tangent space, and $\eratio_i = \lambda_i / \lambda_1$ are the spectral ratios.
\end{proposition}

\begin{proof}
The results follow directly by evaluating the saddle-point contribution $\bm{w}^{-\bm{r}}$ and substituting it into the general expansions of Proposition~\ref{prop:block_asymp}. Using the Sinkhorn parametrizations $t_i^* = k_i \vright{i}^2$ and $u_j^* = \ell_j \vleft{j}^2$ given in \eqref{eq:thm-def-tu-explicit}, the exponential factor expands as
\begin{align}
    \bm{w}^{-\bm{r}}
    & = \bm{k}^{-\bm{k}} \bm{\ell}^{-\bm{\ell}} \Biggl( \prod_{i\in[m]} \vright{i}^{-k_i} \prod_{j\in[m]} \vleft{j}^{-\ell_j} \Biggr)^2 \nonumber\\
    & = \bm{k}^{-\bm{k}} \bm{\ell}^{-\bm{\ell}} \bigl( \perm_{\Sink}(\matr{A}) \bigr)^2 \nonumber\\
    & = \bm{k}^{-\bm{k}} \bm{\ell}^{-\bm{\ell}} \cdot \e^{2n} \bigl( \perm_{\scSink}(\matr{A}) \bigr)^2.
\end{align}
Substituting this explicit form into the asymptotic formulas in Proposition~\ref{prop:block_asymp} yields \eqref{eq:ZGn_explicit_final} and \eqref{eq:ZBn_explicit_final}.
\end{proof}

With the explicit asymptotic forms for the coefficients $\ZGn$ and $\ZDBn$ established in Proposition~\ref{prop:acsv_exact_general}, we now translate these results back to the Gibbs permanent and Bethe double-cover permanent. The following lemma combines the ACSV outcomes with the factorial relations derived in Lemma~\ref{lem:bridge-ogf-nfg} to yield the final asymptotic formulas.

\begin{lemma}
\label{lem:perm_asymp_explicit}
Let $m$ be a positive integer. By combining the exact combinatorial relations in Lemma~\ref{lem:bridge-ogf-nfg} with the asymptotic expansions in Proposition~\ref{prop:acsv_exact_general}, it holds that
\begin{align}
    \frac{\bigl(\perm(\matr{A})\bigr)^2}{\bigl(\perm_{\scSink}(\matr{A})\bigr)^2}
    & \sim
    \frac{\sqrt{2m} \cdot2\pi n \sqrt{\bm{k}^{\bm{1}} \bm{\ell}^{\bm{1}}}}{\|\bm{r}\|_2^m \sqrt{\det(\mathcal{H})}}
    \!\cdot\!
    \prod_{i=2}^m \frac{1}{1\!-\!\eratio_i},
    \label{eq:perm_asymp_final}\\
    \frac{\bigl(\perm_{\Bethe,2}(\matr{A})\bigr)^2}{\bigl(\perm_{\scSink}(\matr{A})\bigr)^2}
    & \sim
    \frac{\sqrt{2m} \cdot2\pi n \sqrt{\bm{k}^{\bm{1}} \bm{\ell}^{\bm{1}}}}{\|\bm{r}\|_2^m \sqrt{\det(\mathcal{H})}}
    \!\cdot\!
    \sqrt{\frac{\e}{\pi n}} \!\cdot\! \prod_{i=2}^m\sqrt{\frac{\e^{\eratio_i}}{1\!-\!\eratio_i}},
    \label{eq:bethe_asymp_final}
\end{align}
where $\|\bm{r}\|_2 = \sqrt{\sum_{i\in[m]} k_i^2 + \sum_{j\in[m]} \ell_j^2}$, $\eratio_i\!=\!\lambda_i/\lambda_1$ denote the spectral ratios defined in Proposition~\ref{prop:acsv_exact_general}, and we use the multi-index notation $\bm{x}^{\bm{1}} \defeq \prod_i x_i$ (with $\bm{1}$ denoting the all-one vector).
\end{lemma}

\begin{proof}
Recall from Lemma~\ref{lem:bridge-ogf-nfg} that
\begin{align*}
    \bigl( \perm(\matr{A}) \bigr)^2 = \bm{k}! \cdot \bm{\ell}! \cdot \ZGn(\bm{k},\bm{\ell}).
\end{align*} 
Applying Stirling's approximation $n!\sim \sqrt{2\pi n} (n/\e)^n$ to $\bm{k}!$ and $\bm{\ell}!$, we obtain
\begin{align}
    \bm{k}! \cdot \bm{\ell}!
    \sim
    (2\pi)^m \cdot\sqrt{\bm{k}^{\bm{1}} \bm{\ell}^{\bm{1}}} \cdot \bm{k}^{\bm{k}} \bm{\ell}^{\bm{\ell}} \cdot \e^{-2n},
    \label{eq:stirling_explicit}
\end{align}
where we used $\sum_{i\in[m]} k_i \!+\! \sum_{j\in[m]} \ell_j = 2n$.

Substituting the explicit asymptotic formula for $\ZGn$ from \eqref{eq:ZGn_explicit_final} into this product, we observe that the terms $\bm{k}^{\bm{k}} \bm{\ell}^{\bm{\ell}} \e^{-2n}$ perfectly cancel the terms $\bm{k}^{-\bm{k}} \bm{\ell}^{-\bm{\ell}} \e^{2n}$ arising from the ACSV expansion. Collecting the remaining geometric terms and simplifying $(2\pi)^m / (2\pi)^{m-1}$ to $2\pi$ yields the expression in \eqref{eq:perm_asymp_final}. The derivation for $\bigl( \perm_{\Bethe,2}(\matr{A}) \bigr)^2$ follows an identical procedure using \eqref{eq:ZBn_explicit_final}.
\end{proof}

To conclude, we highlight two methodological contributions of this derivation beyond the standard ACSV framework. First, our analysis explicitly identifies the effective dimension of the problem. The global constraints inherent in the block-model formulation impose a dependency among the variables, reducing the effective degrees of freedom for the coefficient extraction integral. Correctly incorporating this reduction into the geometric prefactor is crucial for obtaining the exact asymptotic magnitude. Second, the explicit comparison between $\ZGn$ and $\ZDBn$ involves a non-integer exponent shift. While such cases generally fall under the scope of algebraic generating functions in ACSV theory, our derivation adopts a distinct approach by explicitly isolating the singular direction. This separation allows us to resolve the scaling behaviors using classical singularity analysis results in~\cite{flajolet2009analytic}. Specifically, this approach captures the additional decay factor of $n^{-1/2}$ in $\ZDBn$ relative to $\ZGn$, revealing the precise asymptotic ratio between the two components.

\end{document}